\documentclass[11pt, oneside]{article}   	% use "amsart" instead of "article" for AMSLaTeX format
\usepackage{geometry}                		% See geometry.pdf to learn the layout options. There are lots.
\usepackage{graphicx}				% Use pdf, png, jpg, or eps§ with pdflatex; use eps in DVI mode

\usepackage{import}
\usepackage{authblk}						% for affiliations
\usepackage{cite}
\usepackage{array}								% for tables
\usepackage{amssymb}
\usepackage{amsmath}
\usepackage{mathtools}
\usepackage{braket}
\usepackage{amsthm}
\usepackage{dsfont}							% for mathbb{1}
\usepackage{verbatim}						% for multi line comments
\usepackage{MnSymbol}						% for \rrangle
\usepackage{mathrsfs}						% for curlier C
\usepackage{mdframed}
\usepackage{caption}
\usepackage[dvipsnames]{xcolor}

\RequirePackage{doi}
\usepackage{hyperref}
\hypersetup{
  colorlinks   = true,    % Colours links instead of boxes
  urlcolor     = blue,    % Colour for external hyperlinks
  linkcolor    = blue,    % Colour of internal links
  citecolor    = green      % Colour of citations
}

\DeclareCaptionType{Protocol}

\usepackage{latexdef}

\newif\ifcomments
\commentsfalse

\title{Universal chain rules from entropic triangle inequalities}
\author{Ashutosh Marwah\footnote{email: ashutosh.marwah@outlook.com} \ and Fr\'ed\'eric Dupuis}
\affil{\small{D\'epartement d'informatique et de recherche op\'erationnelle,\\ Universit\'e de Montr\'eal,\\ Montr\'eal QC, Canada}}

\date{\today}							% Activate to display a given date or no date

\newcommand{\charFn}[1]{\chi_{#1}}

\begin{document}
\maketitle
\begin{abstract}
    The von Neumann entropy of an $n$-partite system $A_1^n$ given a system $B$ can be written as the sum of the von Neumann entropies of the individual subsystems $A_k$ given $A_1^{k-1}$ and $B$. While it is known that such a chain rule does not hold for the smooth min-entropy, we prove a counterpart of this for a variant of the smooth min-entropy, which is equal to the conventional smooth min-entropy up to a constant. This enables us to lower bound the smooth min-entropy of an $n$-partite system in terms of, roughly speaking, equally strong entropies of the individual subsystems. We call this a \emph{universal chain rule} for the smooth min-entropy, since it is applicable for all values of $n$. Using duality, we also derive a similar relation for the smooth max-entropy. Our proof utilises the entropic triangle inequality based technique developed in \cite{Marwah23} for analysing approximation chains. Additionally, we also prove an approximate version of the entropy accumulation theorem, which significantly relaxes the conditions required on the state to bound its smooth min-entropy. In particular, it does not require the state to be produced through a sequential process like previous entropy accumulation type bounds. In our upcoming companion paper \cite{Marwah24-DIQKD}, we use it to prove the security of parallel device independent quantum key distribution. 
\end{abstract}

\tableofcontents

\section{Introduction}
\begin{sloppypar}
    The rates of several interesting tasks with i.i.d (independent and identically distributed) structure, like coding and randomness extraction, are characterised by the von Neumann entropy \cite{Schumacher95,Wilde13} in the asymptotic regime. The von Neumann entropy is the most commonly used and understood measure of uncertainty. It generally follows the behaviour one intuitively expects from a measure of uncertainty. Central to this paper is the intuitive expectation that the total uncertainty of two registers $A_1$ and $A_2$ given a register $B$ should be the sum of the uncertainty of register $A_1$ given $B$ and the uncertainty of register $A_2$, conditioned on both $A_1$ and $B$. The von Neumann entropy, respects this intuition nicely and satisfies the following chain rule:
    \begin{align}
        H(A_1 A_2| B)_\rho = H(A_1 | B)_\rho + H(A_2 | A_1 B)_\rho.
        \label{eq:bi_von_N_ch_rule}
    \end{align}
    More generally, we can decompose the von Neumann entropy of a large system $A_1^n$\footnote{The notation $A_1^n$ denotes the set of registers $A_1, A_2, \cdots, A_n$.} given $B$ into a sum of the entropies of its parts as:
    \begin{align}
        H(A_1^n | B)_\rho = \sum_{k=1}^n H(A_k | A_1^{k-1} B)_\rho.
        \label{eq:lin_von_N_ch_rule}
    \end{align}
    As we move beyond the asymptotic and i.i.d regime and consider general processes, we move into the territory of one-shot information theory, where we need a multitude of different entropic measures to characterise and analyse various tasks. The smooth min-entropy \cite{Renner06,Renner05} is one of the most important information measures among these. It determines the amount of randomness that can be securely extracted from a (classical) register, when part of it is correlated with the adversary's state \cite{Tomamichel10}. This is particularly useful in cryptography, where no assumptions (specifically i.i.d assumptions) can be placed on the adversary's actions. For a state $\rho_{AB}$, the min-entropy of system $A$ conditioned on system $B$ is defined as\footnote{We use base $e$ for $\exp$ and $\log$, and \emph{nat} units for entropies throughout this paper.}
    \begin{align}
        H_{\min}(A|B)_{\rho} := \max\curlyBrk{\lambda \in \mathbb{R}: \text{there exists a state } {\sigma}_{B} \text{ such that } {\rho}_{AB} \leq e^{-\lambda} \Id_A \otimes \sigma_B}
    \end{align}
    and the smooth min-entropy is defined as 
    \begin{align}
        H_{\min}^{\epsilon}(A|B)_{\rho} := \max_{\tilde{\rho}}H_{\min}(A|B)_{\tilde{\rho}}
    \end{align}
    where the maximum is over all subnormalised states $\tilde{\rho}_{AB}$ which are $\epsilon$-close to the state $\rho$ in the purified distance. The smooth min-entropy behaves very differently from the von Neumann entropy. Informally speaking, it places a much larger weight on the worst case value of the conditioning register as compared to the von Neumann entropy, which places equal weight on all conditioning values. As a consequence, the smooth min-entropy exhibits behaviour that deviates from the intuitive expectations set by the von Neumann entropy. This is evident in the chain rule for smooth min-entropy \cite{Dupuis14, Vitanov13}, which is notably more complex than Eq. \ref{eq:bi_von_N_ch_rule}:
    \begin{align}
        H_{\min}^{\epsilon_1+ 2\epsilon_2+\delta}(A_1 A_2|B)_{\rho} \geq H_{\min}^{\epsilon_1}(A_1|B)_{\rho} + H_{\min}^{\epsilon_2}(A_2| A_1 B)_{\rho} - k(\delta)
        \label{eq:bi_Hmin_ch_rule}
    \end{align}
    where $k(\delta)=O\rndBrk{\log\frac{1}{\delta}}$. Moreover, it is easy to demonstrate that a smooth min-entropy counterpart for the relation in Eq. \ref{eq:lin_von_N_ch_rule} cannot hold true. Mathematically, a relation of the following form cannot be valid for $\epsilon$ in a neighborhood of $0$\footnote{Simply consider the classical distribution $p_{A_1^n B}$ where $B$ is chosen uniformly at random from $\{0,1\}$ and if $B=0$ then $A_1^n$ is set to the constant all zero string, otherwise each $A_i$ is independently and randomly chosen bit. In this case, $H_{\min}^{g_1(\epsilon)}(A_1^n | B)_p = O(1)$ and for each $k$, $H_{\min}(A_k | A_1^{k-1} B)_p \geq \log 4/3$. For small $\epsilon$, Eq. \ref{eq:lin_Hmin_ch_rule} would have a constant left-hand side and a linearly growing right-hand side.}:
    \begin{align}
        H_{\min}^{g_1(\epsilon)}(A_1^n | B)_\rho &\geq \sum_{k=1}^n H_{\min}^\epsilon(A_k | A_1^{k-1} B)_{\rho} - ng_2(\epsilon) - k(\epsilon),
        \label{eq:lin_Hmin_ch_rule}
    \end{align}
    where the functions $g_1$, $g_2$, and $k$ are dependent solely on $\epsilon$ and $|A|$ (the dimension of the $A_k$ registers, assumed to be constant in $n$) and are independent of $n$. Furthermore, $g_1(\epsilon)$ and $g_2(\epsilon)$ are required to be \emph{small} functions of $\epsilon$, meaning they are continuous and approach $0$ as $\epsilon$ tends to $0$.\\

    Eq. \ref{eq:lin_Hmin_ch_rule} is intimately related to the concept of \emph{approximation chains} introduced in \cite{Marwah23}. For a state $\rho_{A_1^n B}$, a sequence of states $(\sigma_{A_1^k B}^{(k)})_{k=1}^n$ is called an $\epsilon$-approximation chain of $\rho$ if for every $1 \leq k\leq n$, we have $\rho_{A_1^k B} \approx_{\epsilon} \sigma_{A_1^k B}^{(k)}$. Often, direct bounds on the conditional entropies of the state $\rho$ are unavailable; however, they can be obtained for the entropies of its approximation chain states. For instance, consider a state $\rho_{A_1^n B}$ where the $A_k$ registers are produced sequentially. Ideally, each register $A_k$ should be sampled independently from the rest. However, say imperfections introduce minor correlations between $A_k$ and the earlier registers. In such a case, we might only be able to confirm that $\rho_{A_1^k B} \approx_\epsilon \rho_{A_k} \otimes \rho_{A_1^{k-1} B} =: \sigma_{A_1^k B}^{(k)}$. Despite these correlations, we expect the entropy of $A_1^n$ given $B$ to be large, as the $A_k$ registers are almost independent of $B$. In such cases, a chain rule-like tool relating the entropy of $\rho$ to the entropies of its approximation chain states would be helpful.\\
    
    The impossibility of a bound of the form in Eq. \ref{eq:lin_Hmin_ch_rule} implies that $H_{\min}^{\epsilon}(A_1^n | B)_\rho$ cannot be lower bounded meaningfully in terms of the min-entropies of the approximation chain states of $\rho$, since these approximation chain states can simply be states satisfying $H_{\min}(A_k | A_1^{k-1} B)_{\sigma^{(k)}} = H_{\min}^\epsilon(A_k | A_1^{k-1} B)_{\rho}$ for every $k$. On the other hand, the von Neumann entropy of a state can easily be bounded in terms of its approximation chain by using the continuity of the conditional von Neumann entropy \cite{Alicki04,Winter16} to modify Eq. \ref{eq:lin_von_N_ch_rule} and derive: 
    \begin{align}
        H(A_1^n | B)_\rho &\geq \sum_{k=1}^n H(A_k | A_1^{k-1} B)_{\sigma^{(k)}} - nf(\epsilon)
    \end{align}
    where $f(\epsilon) = O\rndBrk{\epsilon\log\frac{|A|}{\epsilon}}$. The absence of a comparable bound for the smooth min-entropy severely limits us\footnote{It is often useful to initially prove novel results using von Neumann entropies before translating them into the corresponding one-shot entropies. This approach separates the complexity of the problem into two distinct phases. The impossibility of Eq. \ref{eq:lin_Hmin_ch_rule}, in particular, prevents us from porting arguments based on the von Neumann entropy into smooth min-entropy arguments.}.\\

    There are multiple alternative definitions of the smooth min-entropy, which are equal to the one we defined above up to a constant \cite{Renner06,Tomamichel10,Anshu20}. One of these, is the $H^{\downarrow}_{\min}$ min-entropy and its smoothed variant $H^{\downarrow, \epsilon}_{\min}$, defined as:
    \begin{align}
        H^{\downarrow}_{\min}(A|B)_{\rho} &:= \max\curlyBrk{\lambda \in \mathbb{R}:\rho_{AB} \leq e^{-\lambda}\Id_A \otimes \rho_B} \\ 
        H^{\downarrow, \epsilon}_{\min}(A|B)_{\rho} &:= \sup_{\tilde{\rho}}H^{\downarrow}_{\min}(A|B)_{\tilde{\rho}}
    \end{align}
    where the supremum is over all subnormalised states $\tilde{\rho}_{AB}$ which are $\epsilon$-close to the state $\rho$ in the purified distance. \cite[Lemma 20]{Tomamichel10} showed that this smooth min-entropy is equal to $H^{ \epsilon}_{\min}$ up to a constant:
    \begin{align}
        H_{\min}^{\epsilon/2}(A_1^n|B)_{\rho} - O\rndBrk{\log \frac{1}{\epsilon}} \leq H_{\min}^{\downarrow, \epsilon}(A_1^n|B)_{\rho} \leq H_{\min}^{\epsilon}(A_1^n|B)_{\rho}.
    \end{align}
    One can now ask whether $H_{\min}^{\downarrow, \epsilon}$ satisfies a chain rule like Eq. \ref{eq:lin_Hmin_ch_rule}, that is, does
    \begin{align}
        H_{\min}^{\downarrow, g_1(\epsilon)}(A_1^n | B)_\rho &\geq \sum_{k=1}^n H_{\min}^{\downarrow, \epsilon}(A_k | A_1^{k-1} B)_{\rho} - ng_2(\epsilon) - k(\epsilon),
        \label{eq:lin_dn_Hmin_ch_rule}
    \end{align}
    hold true for some $g_1, g_2$ and $k$ as in Eq. \ref{eq:lin_Hmin_ch_rule}? Remarkably, we prove that this is indeed the case. Establishing this in Theorem \ref{th:Hmin_eps_chain_rule} is the first main result of this paper. We call this a \emph{universal chain rule} for the smooth min-entropy to emphasise the fact that it is true and meaningful for a constant $\epsilon \in (0,1)$ and an arbitrary $n \in \mathbb{N}$. This universal chain can be viewed as a smoothed generalisation of the chain rule for $H_{\min}^{\downarrow}$:
    \begin{align}
        H_{\min}^{\downarrow}(A_1^n | B)_\rho &\geq \sum_{k=1}^n H_{\min}^{\downarrow}(A_k | A_1^{k-1} B)_{\rho},
        \label{eq:Hmin_dn_non_sm_ch_rule}
    \end{align}
    which is well known and fairly simple to prove \cite[Proposition 5.5]{TomamichelBook16}. The universal chain rule in Eq. \ref{eq:lin_dn_Hmin_ch_rule} is particularly interesting because it can be used to decompose the smooth min-entropy (both $H_{\min}^{\epsilon}$ and $H_{\min}^{\downarrow, \epsilon}$) into a sum of conditional entropies, which are equally strong. As a corollary to this, we also prove a universal chain rule for the smooth max-entropy (Corollary \ref{cor:H0_eps_chain_rule}). \\

    We provide two proofs for this chain rule. We briefly describe the first proof technique here, since it is simpler and we use it to prove the approximate entropy accumulation theorem as well. This proof follows the entropic triangle inequality based approach developed in \cite{Marwah23} for analysing approximation chains. The entropic triangle inequality allows us to bound the smooth min-entropy of the state $\rho$ in Eq. \ref{eq:lin_dn_Hmin_ch_rule} with the min-entropy of an auxiliary state $\sigma$ as
    \begin{align}
        H^{\delta}_{\min}(A_1^n|B)_{\rho} \geq H_{\min}(A_1^n|B)_{\sigma} - D^{\delta}_{\max}(\rho|| \sigma).
    \end{align}
    Using the Generalised Golden Thompson inequality \cite{Sutter17}, we identify a state $\sigma_{A_1^n B}$ satisfying
    \begin{align}
        & D_m(\rho_{A_1^n B} || \sigma_{A_1^n B}) \leq n g(\epsilon) \label{eq:intro_sigma_D_bd}\\
        & H_{\min}(A_1^n|B)_{\sigma} \geq \sum_{k=1}^n H_{\min}^{\downarrow, \epsilon}(A_k | A_1^{k-1} B)_{\rho} \label{eq:intro_sigma_Hmin_bd}
    \end{align}
    where $D_m$ is the measured relative entropy and $g(\epsilon)$ is a small function of $\epsilon$. Selecting an appropriate state $\sigma$ is the most non-trivial part of the proof. The bound in Eq. \ref{eq:intro_sigma_D_bd} can further be transformed into a smooth max-relative entropy bound using the substate theorem \cite{Jain02, Jain11}. Putting these bounds together into the triangle inequality then yields the universal chain rule. \\

    The second major result in our paper is an \emph{unstructured} approximate entropy accumulation theorem (Theorem \ref{th:approx_EAT}). The entropy accumulation theorem (EAT) \cite{Dupuis20}, similar to a chain rule, decomposes the smooth min-entropy of a large system into a sum of the entropies of its parts. EAT can only be used for states $\rho_{A_1^n B_1^n E}$ satisfying the Markov chain $A_1^{k-1} \leftrightarrow B_1^{k-1}E \leftrightarrow B_k$ for every $k \in [n]$, which are produced by applying maps $\cM_k : R_{k-1}\rightarrow A_k B_k R_k$ sequentially so that $\rho_{A_1^n B_1^n E} = \tr_{R_k} \circ \cM_n \circ \cdots \circ \cM_1 (\rho^{(0)}_{R_0 E})$. In Theorem \ref{th:approx_EAT}, we significantly relax the conditions on the structure of the state required to obtain an EAT like bound. We show that for any state as long as for every $k \in [n]$: $\rho_{A_1^k B_1^k E} \approx_{\epsilon} \cN_k(\tilde{\rho}^{(k)}_{A_1^{k-1} B_1^{k-1} E R_k})$ for some state $\tilde{\rho}^{(k)}_{A_1^{k-1} B_1^{k-1} E R_k}$ and a channel $\cN_k: R_k \rightarrow A_k B_k$ which samples $B_k$ independent of the previous registers\footnote{It is worth noting two points. First, this condition is typically satisfied in applications of EAT. Second, if one strengthens the Markov chain condition for EAT, such that the Markov chain $A_1^{k-1} \leftrightarrow B_1^{k-1}E \leftrightarrow B_k$ holds for all inputs to the channels $\cM_k$, then this seems to reduce to the independence condition used in this context.}, we have the bound 
    \begin{align}
        H^{\tilde{O}(\epsilon^{1/6})}_{\min} (A_1^n | B_1^n E)_\rho &\geq \sum_{k=1}^n \inf_{\omega_{R_k \tilde{R}_k}} H (A_k | B_k \tilde{R}_k)_{\cM_k(\omega)}  - n \tilde{O}(\epsilon^{1/12}) - \tilde{O}\rndBrk{\frac{1}{\epsilon^{5/12}}}
    \end{align}
    where $\tilde{O}$ hides the logarithmic factors in $\epsilon$ and the infimum is over all states $\omega_{R_k \tilde{R}_k}$. This bound holds irrespective of the process which produces the state. In fact, the central motivation behind this theorem was proving the security of parallel device independent quantum key distribution (DIQKD). We do this in our upcoming companion paper \cite{Marwah24-DIQKD}. The proof approach for this theorem is very similar to that of the universal chain rule.

    % spiritual sequel
    % Its behaviour underlies our intuition, reflexivity

\end{sloppypar}

\section{Background}
\begin{sloppypar}
\subsection{Notation}

For a classical probability distribution $p_{AB}$, the conditional probability distribution $p_{A|B}$ is defined as $p_{A|B}(a|b) := \frac{p_{AB}(a,b)}{p_B(b)}$ when $p_B(b) >0$. For the case when $p_B(b) =0$, we define $p_{A|B}$ to be the uniform distribution for our purposes. The probability distribution $q_B p_{A|B}$ for a distribution $q$ on random variable $B$ is defined as $q_B p_{A|B} (b,a) := q_B(b) p_{A|B}(a|b)$.\\

The term normalised (subnormalised) quantum states is used for positive semidefinite operators with unit trace (trace less than $1$). We denote the set of registers a quantum state describes (equivalently, its Hilbert space) using a subscript. Partial states of a quantum state are simply denoted by restricting the set of registers in the subscript. For example, if $\rho_{AB}$ denotes a quantum state on registers $A$ and $B$, then the partial states on registers $A$ and $B$, will be denoted as $\rho_{A}$ and $\rho_{B}$ respectively. \\

The term ``channel'' is used for completely positive trace preserving (CPTP) linear maps between two spaces of Hermitian operators. A channel $\cN$ mapping registers $A$ to $B$ will be denoted by $\cN_{A \rightarrow B}$. \\

We use the notation $[n]$ to denote the set $\{1,2, \cdots, n\}$. For $n$ quantum registers $(X_1, X_2, \cdots, X_n)$, the notation $X_i^j$ for $i <j$ refers to the set of registers $(X_i, X_{i+1}, \cdots, X_{j})$. For a register $A$, $|A|$ represents the dimension of the underlying Hilbert space. \\

For two Hermitian operators $X$ and $Y$, the operator inequality $X \geq Y$ is used to denote that $X-Y$ is a positive semidefinite operator and $X>Y$ denotes that $X-Y$ is a strictly positive operator. The notation $X \ll Y$ denotes that the support of operator $X$ is contained in the support of $Y$. The identity operator on register $A$ is denoted using $\Id_A$.\\

\subsection{Information theory}
The trace norm is defined as $\norm{X}_1 := \tr\big(\rndBrk{X^\dag X}^{\frac{1}{2}}\big)$. The fidelity between two positive operators $P$ and $Q$ is defined as $F(P,Q)= \norm{\sqrt{P}\sqrt{Q}}_1^2$. The generalised fidelity between two subnormalised states $\rho$ and $\sigma$ is defined as 
\begin{align}
    F_\ast(\rho, \sigma) := \rndBrk{\norm{\sqrt{\rho}\sqrt{\sigma}}_1 + \sqrt{(1- \tr\rho)(1- \tr\sigma)}}^2.
\end{align}
The purified distance between two subnormalised states $\rho$ and $\sigma$ is defined as 
\begin{align}
    P(\rho, \sigma) = \sqrt{1- F_{\ast}(\rho, \sigma)}.
\end{align}
Throughout this paper, we use base $e$ for both the functions $\log$ and $\exp$. Therefore, all the entropies in this paper are in \emph{nat} units. We follow the notation in Tomamichel's book~\cite{TomamichelBook16} for R\'enyi entropies. The sandwiched $\alpha$-R\'enyi relative entropy for $\alpha \in [\frac{1}{2},1) \cup (1,\infty]$ between the positive operator $P$ and $Q$ is defined as 
\begin{align}
    \tilde{D}_{\alpha}(P || Q) = \begin{cases}
        \frac{1}{\alpha-1} \log \frac{\tr(Q^{-\frac{\alpha'}{2}} P Q^{-\frac{\alpha'}{2}})^{\alpha}}{\tr(P)} & \text{ if } (\alpha<1 \text{ and } P \not\perp Q) \text{ or } (P \ll Q)\\
        \infty & \text{ else}.
    \end{cases}
\end{align}
where $\alpha' = \frac{\alpha-1}{\alpha}$. In the limit $\alpha \rightarrow 
\infty$, the sandwiched divergence becomes equal to the max-relative entropy, $D_{\max}$, which is defined as 
\begin{align}
    D_{\max} (P|| Q) := \inf \curlyBrk{\lambda \in \mathbb{R}: P \leq e^{\lambda}Q}.
\end{align}
In the limit of $\alpha \rightarrow 1$, the sandwiched relative entropies equals the quantum relative entropy, $D(P ||Q)$, which is defined as 
\begin{align}
    D(P|| Q) := \begin{cases}
        \frac{\tr\rndBrk{P \log P - P \log Q}}{\tr(P)} & \text{ if } (P \ll Q)\\
        \infty & \text{ else}.
    \end{cases}
\end{align}
In this work, we will also need the measured relative entropy $D_m$. For two positive operators $P$ and $Q$, it is defined as
\begin{align}
    D_m (P|| Q) := \sup_{\cM} D(\cM(P) || \cM(Q))
\end{align}
where the supremum is taken over all measurement channels $\cM$, that is, a channel such that $\cM(\rho) = \sum_{x \in \mathcal{X}} \tr(M_x \rho) \ket{x}\bra{x}$ for some POVM elements $(M_x)_x$ and an alphabet $\mathcal{X}$.\\

We can use the sandwiched divergence to define the following conditional entropies for a subnormalised state $\rho_{AB}$:
\begin{align}
    \tilde{H}_{\alpha}^{\uparrow} (A|B)_{\rho} &:= \max_{\sigma_B} - \tilde{D}_{\alpha}(\rho_{AB} || \Id_A \otimes \sigma_B) \\
    \tilde{H}_{\alpha}^{\downarrow} (A|B)_{\rho} &:= - \tilde{D}_{\alpha}(\rho_{AB} || \Id_A \otimes \rho_B)
\end{align}
for $\alpha \in [\frac{1}{2},1) \cup (1,\infty]$. The maximum in the definition for $ \tilde{H}_{\alpha}^{\uparrow}$ is over all quantum states $\sigma_B$ on register $B$. \\

For $\alpha \rightarrow 1$, both these conditional entropies are equal to the von Neumann conditional entropy $H(A|B)$. $\tilde{H}_{\infty}^{\uparrow} (A|B)_{\rho}$ is usually called the min-entropy and is denoted as $H_{\min}(A|B)_{\rho}$. For a subnormalised state, it can also be defined as 
\begin{align}
    H_{\min}(A|B)_{\rho} &:= \max \curlyBrk{\lambda \in \mathbb{R}: \text{ there exists state }\sigma_B \text{ such that } \rho_{AB} \leq e^{-\lambda} \Id_A \otimes \sigma_B}.
\end{align}  
Another version of the min-entropy, that will be relevant in this paper is given by 
$\tilde{H}_{\infty}^{\downarrow} (A|B)_{\rho}$. We denote it in this paper using $H^{\downarrow}_{\min}(A|B)_{\rho}$. It can alternatively be defined as
\begin{align}
    H^{\downarrow}_{\min}(A|B)_{\rho} :=& \max \curlyBrk{\lambda \in \mathbb{R}:  \rho_{AB} \leq e^{-\lambda} \Id_A \otimes \rho_B}\\
    =&  - \log \norm{\rho_B^{-\frac{1}{2}} \rho_{AB} \rho_B^{-\frac{1}{2}}}_{\infty} 
\end{align} 
For a classical distribution $p_{AB}$, the min-entropies above simplify to 
\begin{align}
    & H_{\min}(A|B)_{p} = - \log \sum_{b} p(b)\max_{a}p(a|b) \label{eq:Hmin_cl_defn}\\
    & H^{\downarrow}_{\min}(A|B)_p = - \log \max_{a,b} p(a|b). \label{eq:Hmin_down_cl_defn}
\end{align}
The expressions for these entropies in the classical case clearly demonstrates the difference between $H_{\min}$ and $H_{\min}^{\downarrow}$. Since $H_{\min}$ is the negative log of the guessing probability \cite{Konig09}, it averages over the values of $B$. Whereas $H_{\min}^{\downarrow}$ selects the worst possible value of $B$ and then calculates the guessing probability for this value\footnote{Note that $H_{\min}^{\downarrow}$ is not continuous.}. \\

The entropy $\tilde{H}_{1/2}^{\uparrow} (A|B)_{\rho}$ is commonly referred to as the max-entropy and is denoted as $H_{\max}(A|B)_{\rho}$. It can also be written as
\begin{align}
    H_{\max}(A|B)_{\rho} := \max_{\sigma_B} \log F(\rho_{AB}, \Id_A \otimes \sigma_B)
\end{align}
where the maximum is over all states $\sigma_B$. \cite{Renner06,Tomamichel10} also define the following alternative max-entropy
\begin{align}
  \bar{H}^{\uparrow}_{0}(A|B)_{\rho} :=& \max_{\sigma_B} \log \tr \rndBrk{\rho_{AB}^0 \sigma_B} \\
  =& \log \norm{\tr_A \rho_{AB}^0}_{\infty}
\end{align}
where $\rho_{AB}^0$ is the projector on the support of $\rho_{AB}$ and the maximum in the first line is over all states $\sigma_B$ on register $B$.\\

For the purpose of smoothing, we define the $\epsilon$-ball around a subnormalised state $\rho$ as the set
\begin{align}
    B_{\epsilon}(\rho) = \{ \tilde{\rho} \geq 0 : P(\rho, \tilde{\rho}) \leq \epsilon \text{ and } \tr\tilde{\rho} \leq 1\}.
\end{align}
The smooth max-relative entropy is defined as 
\begin{align}
    D_{\max}^{\epsilon}(\rho || \sigma) := \min_{\tilde{\rho} \in B_{\epsilon}(\rho)} D_{\max}(\tilde{\rho} || \sigma)
\end{align}
We define the two version of smooth min-entropy of $\rho_{AB}$ using the min-entropies defined above as 
\begin{align}
    H_{\min}^{\epsilon}(A|B)_{\rho} := \max_{\tilde{\rho} \in B_{\epsilon}(\rho)} H_{\min}(A|B)_{\tilde{\rho}} \\
    H_{\min}^{\downarrow, \epsilon}(A|B)_{\rho} := \sup_{\tilde{\rho} \in B_{\epsilon}(\rho)} H^{\downarrow}_{\min}(A|B)_{\tilde{\rho}}.
\end{align} 
In \cite[Lemma 20]{Tomamichel10} it is shown that these two are related as 
\begin{align}
    H_{\min}^{\epsilon}(A|B)_{\rho} - \log\rndBrk{\frac{2}{\epsilon^2} + \frac{1}{1-\epsilon}} \leq H_{\min}^{\downarrow, 2\epsilon}(A|B)_{\rho} \leq H_{\min}^{2\epsilon}(A|B)_{\rho}.
    \label{eq:Hmin_up_dn_reln}
\end{align}
That is, these entropies are essentially the same up to a constant term and a constant factor in the smoothing parameter. \\

\noindent We define the smooth max-entropy as 
\begin{align}
  H_{\max}^{\epsilon}(A|B)_{\rho} := \min_{\tilde{\rho} \in B_{\epsilon}(\rho)} H_{\max}(A|B)_{\tilde{\rho}}
\end{align}
Following \cite{Tomamichel10}, we define the smoothed version of $\bar{H}^{\uparrow}_{0}$ as 
\begin{align}
  \bar{H}^{\uparrow, \epsilon}_{0}(A|B)_{\rho} := \inf \curlyBrk{ \bar{H}^{\uparrow}_{0}(A|B')_{\tilde{\rho}} : \text{ for a Hilbert space } B' \geq B \text{ and } \tilde{\rho}_{AB'} \in B_{\epsilon}(\rho_{AB'})}.
\end{align}
The notation $B'\geq B$ for Hilbert spaces represents that $B$ is a Hilbert subspace of the Hilbert space $B'$. $\rho_{AB'}$ above is an embedding of the state $\rho_{AB}$ in $AB'$. It is sufficient to restrict the dimension of $B'$ above to $|A|^2 |B|$. \\

For a pure state $\rho_{ABC}$, the following duality relations are shown in \cite{Tomamichel10}:
\begin{align}
  & H^{\epsilon}_{\min} (A|B) =  -H_{\max}^{\epsilon} (A|C) \\
  & H^{\downarrow, \epsilon}_{\min} (A|B) =  -\bar{H}^{\uparrow, \epsilon}_{0} (A|C)
  \label{eq:Hmin_dn_dual}
\end{align}
As a result, the following relation between the two smooth max-entropies can be derived from Eq. \ref{eq:Hmin_up_dn_reln} (\cite[Lemma 20]{Tomamichel10}):
\begin{align}
    H_{\max}^{2\epsilon}(A|B)_{\rho} \leq \bar{H}^{\uparrow, \epsilon}_{0}(A|B)_{\rho} \leq H_{\max}^{\epsilon}(A|B)_{\rho} + \log\rndBrk{\frac{2}{\epsilon^2} + \frac{1}{1-\epsilon}}.
    \label{eq:Hmax_up_dn_reln}
\end{align}

\end{sloppypar}

\subsection{Key lemmas and theorems}
\begin{sloppypar}
In this section, we state some important results, which we use throughout this paper. As stated in the introduction, the entropic triangle inequalities developed in \cite{Marwah23} are central to analysing approximation chains and hence to this work. The following lemma states the two inequalities we use. 

\begin{lemma}[Entropic triangle inequality {\cite[Introduction and Lemma 3.5]{Marwah23}}]
    For two states $\rho_{AB}$ and $\eta_{AB}$ and $\epsilon \in [0,1)$, we have
    \begin{align}
        H_{\min}^{\epsilon}(A|B)_{\rho} \geq H_{\min}(A|B)_{\eta} - D^{\epsilon}_{\max}(\rho_{AB} || \eta_{AB}).
        \label{eq:ent_tri_ineq_simp}
    \end{align}
    Further, for $\alpha \in (1,2]$ and $\delta \in (0,1)$ such that $\epsilon + \delta < 1$, we also have 
    \begin{align}
        H_{\min}^{\epsilon+\delta}(A|B)_{\rho} &\geq \tilde{H}^{\uparrow}_{\alpha}(A|B)_{\eta} - \frac{\alpha}{\alpha-1} D^\epsilon_{\max} (\rho_{AB}|| \eta_{AB}) - \frac{g_1(\delta, \epsilon)}{\alpha-1}
        \label{eq:ent_tri_ineq_alpha_renyi}
    \end{align}
    where $g_1(x, y):= - \log(1- \sqrt{1-x^2}) - \log (1-y^2)$. 
    \label{lemm:ent_tri_ineq}
\end{lemma}

The quantum substate theorem stated below is one of the major results used to bound the smooth max-relative entropy in this paper. This theorem serves as a tool to convert bounds on the measured relative entropy to bounds on the smooth max-relative entropy. 

\begin{theorem}[Quantum substate theorem \cite{Jain02,Jain11}]
    \label{thm:substate_th}
    Let $\rho$ and $\sigma$ be two normalised states on the same Hilbert space. Then for any $\epsilon \in (0,1)$, we have
    \begin{align}
        D^{\sqrt{\epsilon}}_{\max}(\rho || \sigma) \leq \frac{D_m (\rho || \sigma)+1}{\epsilon} + \log \frac{1}{1-\epsilon}.
        \label{eq:subset_th_eq2}
    \end{align}
\end{theorem}
Usually, the above theorem is stated with the relative entropy $D (\rho || \sigma)$ on the right-hand side instead of the measured relative entropy $D_m (\rho || \sigma)$. Since, in the proofs in this paper we are only able to derive bounds on the measured relative entropy, we need the stronger version stated above. It is, therefore, instructive to understand how it is possible to use $D_m (\rho || \sigma)$ in the bound above. \\

\cite{Jain02,Jain11} actually prove a bound on $D^{\sqrt{\epsilon}}_{\max}(\rho || \sigma)$ in terms of a divergence they call the \emph{observational divergence} $D_{obs}$, which is defined as 
\begin{align}
    D_{obs}(\rho || \sigma) := \sup\curlyBrk{\tr(P\rho) \log \frac{\tr(P\rho)}{\tr(P\sigma)} : \text{ for } 0\leq P\leq \Id}.
\end{align}
\cite[Theorem 1]{Jain11} proves that 
\begin{align}
    D^{\sqrt{\epsilon}}_{\max}(\rho || \sigma) \leq \frac{D_\text{obs} (\rho || \sigma)}{\epsilon} + \log \frac{1}{1-\epsilon}.
    \label{eq:subset_th_obs_div}
\end{align}
Observe that for any binary measurement $(P, \Id - P)$, we have 
\begin{align*}
    D_m(\rho || \sigma)
    &\geq \tr(P\rho) \log \frac{\tr(P\rho)}{\tr(P\sigma)} + (1-\tr(P\rho)) \log \frac{1- \tr(P\rho)}{1-\tr(P\sigma)} \\
    &\geq \tr(P\rho) \log \frac{\tr(P\rho)}{\tr(P\sigma)} + (1-\tr(P\rho)) \log \rndBrk{1- \tr(P\rho)} \\
    &\geq \tr(P\rho) \log \frac{\tr(P\rho)}{\tr(P\sigma)} - 1 
\end{align*}
which implies, that 
\begin{align}
    D_{obs}(\rho || \sigma) \leq D_m(\rho || \sigma) +1.
\end{align}
This allows us to use $D_m(\rho || \sigma)$ in the bound for the substate theorem. \\

In addition to the quantum substate theorem, the following generalisation of the Golden-Thompson (GT) inequality is another major result enabling the proof approach used in this paper. 

\begin{theorem}[Generalised Golden-Thompson (GT) Inequality {\cite{Sutter17}}]
    \label{thm:gen_golden_thompson}
    For a collection of Hermitian matrices $\{H_k\}_{k=1}^n$, we have 
    \begin{align}
        \tr\exp{\rndBrk{\sum_{k=1}^n H_k}} \leq \int_{-\infty}^{\infty} dt \beta_0 (t) \tr\rndBrk{e^{H_n} e^{\frac{1-it}{2}H_{n-1}}\cdots\ e^{\frac{1-it}{2}H_{2}} e^{H_1} e^{\frac{1+it}{2}H_{2}} \cdots\ e^{\frac{1+it}{2}H_{n-1}}}
    \end{align}
    where $\beta_0 (t) := \frac{\pi}{2} (\cosh (\pi t) + 1)^{-1}$ is a probability density function. 
\end{theorem}
\ifcomments \textcolor{red}{\noindent Note that by simply multiplying the Hermitian matrices by $\log(c)$, we see that a similar inequality holds when the base of the exponent is $c$ instead of $e$.}\fi
\begin{proof}
    Using \cite[Corollary 3.3]{Sutter17}, we have 
    \begin{align}
        & \log\norm{\exp\rndBrk{\sum_{k=1}^n \frac{1}{2}H_k}}_2 \leq \int_{-\infty}^{\infty} dt \beta_0 (t) \log\norm{\prod_{k=1}^n \exp\rndBrk{\frac{1+it}{2} H_k}}_2.        
    \end{align}
    Expanding the norm gives
    \begin{align}
        \frac{1}{2} \log\tr \rndBrk{\exp\rndBrk{\sum_{k=1}^n H_k}} \leq \int_{-\infty}^{\infty} dt \beta_0 (t) \frac{1}{2} \log\tr\rndBrk{e^{H_n} e^{\frac{1-it}{2}H_{n-1}}\cdots e^{\frac{1-it}{2}H_{2}} e^{H_1} e^{\frac{1+it}{2}H_{2}} \cdots e^{\frac{1+it}{2}H_{n-1}}}.
    \end{align}
    Using the concavity and monotonicity of $\log$, we get the statement in the Theorem. 
\end{proof}

The GT inequality above is often used in conjunction with the following variational expressions for the relative entropies. This is also the case in this paper.

\begin{lemma}[\cite{Petz88,Berta17}]
    For a normalised state $\rho$ and a positive operator $Q$, the following variational forms hold true:
    % (from \cite[Lemma 2.29, Lemma 2.34]{Sutter18}), 
    \begin{align}
        &D(\rho || Q) = \sup_{\omega > 0} \curlyBrk{\tr(\rho \log \omega) + 1 - \tr\exp(\log Q + \log \omega)} \label{eq:rel_ent_var_rep}\\
        &D_{m}(\rho || Q) = \sup_{\omega > 0} \curlyBrk{\tr(\rho \log \omega) + 1 - \tr (Q \omega)}. \label{eq:meas_rel_ent_var_rep}
    \end{align}
\end{lemma}
\ifcomments \textcolor{red}{(Need base $e$ for this)}\fi

We will also use the following lemma, which is a quantum generalisation of the fact that if two probability distributions $p_{AB}$ and $q_{AB}$ are close to each other, then the probability distributions $p_{AB}$ and $p_B q_{A|B}$ are also close to each other. 

\begin{lemma}
    For a normalised state $\rho_{AB}$ and a subnormalised state $\tilde{\rho}_{AB}$ such that $P(\rho_{AB}, \tilde{\rho}_{AB}) \leq \epsilon$, the state $\eta_{AB} := \rho_{B}^{1/2}\tilde{\rho}_{B}^{-1/2}\tilde{\rho}_{AB}\tilde{\rho}_{B}^{-1/2}\rho_{B}^{1/2}$ (Moore-Penrose pseudo-inverse) satisfies $P(\rho_{AB}, \eta_{AB}) \leq (\sqrt{2} + 1)\epsilon$. Note that if $\tilde{\rho}_B$ is full rank, then $\eta_B = \rho_B$. 
    \label{lemm:cond_state_dist}
\end{lemma}
  
\begin{proof}
Note that since $\rho_{AB}$ is normalised, we have $F(\tilde{\rho}_{AB}, \rho_{AB})\geq 1- \epsilon^2$. Let $\ket{\tilde{\rho}}_{ABR}$ be an arbitrary purification of $\tilde{\rho}_{AB}$. Observe that the pure state $\ket{\eta} := \rho_{B}^{1/2} \tilde{\rho}_{B}^{-1/2} \ket{\tilde{\rho}}_{ABR}$ is a purification of $\eta_{AB}$. \\

Let $\ket{\tilde{\rho}}_{ABR} = \sum_{i} \sqrt{\tilde{p}_i} \ket{u_i}_{B}\otimes \ket{v_i}_{AR}$ where all $p_i > 0$ be the Schmidt decomposition of $\tilde{\rho}_{ABR}$. This implies $\tilde{\rho}_B = \sum_i \tilde{p}_i \ket{u_i} \bra{u_i}_B$. Then, using Uhlmann's theorem \cite[Theorem 3.22]{Watrous18}, we have 
\begin{align*}
    F(\tilde{\rho}_{AB}, \eta_{AB}) &\geq |\braket{\tilde{\rho} | \eta}|^2 \\
    &= \left\vert \braket{\tilde{\rho} | \rho_{B}^{1/2} \tilde{\rho}_{B}^{-1/2} | \tilde{\rho}} \right\vert^2 \\
    &= \left\vert \tr \rndBrk{\rho_{B}^{1/2} \tilde{\rho}_{B}^{-1/2} \tilde{\rho}_{ABR}}\right\vert^2 \\
    &= \left\vert\tr\rndBrk{\rho_{B}^{1/2} \tilde{\rho}_B^{1/2} } \right\vert^2 \\
    &\geq F(\tilde{\rho}_B, \rho_{B})^2 \\
    &\geq (1-\epsilon^2)^2\\
    &\geq 1-2\epsilon^2
\end{align*}
where we have used the relation between the \emph{pretty good} fidelity and fidelity \cite[Eq. 44]{Iten17} for the first inequality and the fact that $F(\tilde{\rho}_B, \rho_{B}) \geq F(\tilde{\rho}_{AB}, \rho_{AB})$. Further, we have 
\begin{align*}
    P(\tilde{\rho}_{AB}, \eta_{AB}) &= \sqrt{1- F_\ast(\tilde{\rho}_{AB}, \eta_{AB})} \\
    &\leq \sqrt{1- F(\tilde{\rho}_{AB}, \eta_{AB})} \\
    &\leq \sqrt{2}\epsilon.
\end{align*}
Using the triangle inequality, we get
\begin{align*}
    P({\rho}_{AB}, \eta_{AB}) &\leq P({\rho}_{AB}, \tilde{\rho}_{AB}) + P(\tilde{\rho}_{AB}, \eta_{AB})\\
    &\leq (\sqrt{2} + 1 )\epsilon.
\end{align*}
\end{proof}

\end{sloppypar}

% first proof for the chain rule
\section{A universal chain rule for smooth min-entropy}

In this section, we will prove the universal chain rule for $H^{\downarrow, \epsilon}_{\min}$. Specifically, we will show that for a state $\rho_{A_1^n B}$,
\begin{align}
  H^{\downarrow, g_1(\epsilon)}_{\min}(A_1^n | B)_{\rho} \geq \sum_{k=1}^n H^{\downarrow, \epsilon}_{\min}(A_k | A_1^{k-1} B)_{\rho} - n g_2(\epsilon) - k(\epsilon),
  \label{eq:target_Hmin_ch_rule}
\end{align}
where $g_1$ and $g_2$ are small functions of $\epsilon$ ($g_1$ and $g_2$ are continuous and tend to $0$ as $\epsilon \rightarrow 0$) and $k$ is a general function of $\epsilon$. It should be noted that these functions may depend on $|A|$, which is the size of the individual registers $A_k$. We begin by sketching the proof of such a chain rule in the classical case first. This will be beneficial for understanding the challenges that need to be solved in order to prove the statement in the quantum case. We will then generalise this proof to the quantum case. In Sec. \ref{sec:alternate_proof}, we provide an alternate proof for this chain rule. We compare the techniques used to prove this chain rule with the previous chain rules and their proofs in Appendix \ref{sec:ch_rule_comp}.

\subsection{Proof sketch for classical distributions}
\label{sec:uni_ch_rul_cl1}

Consider the probability distribution $p_{A_1^n B}$. We will sketch a proof for a chain rule of the form in Eq. \ref{eq:target_Hmin_ch_rule} for $\rho = p$. Let us first broadly describe the proof strategy. We will identify an auxiliary distribution $p^{(\text{aux})}_{A_1^n B}$, such that 
\begin{align}
  &D^{g_1(\epsilon)}_{\max}(p_{A_1^n B} || p^{(\text{aux})}_{A_1^n B}) \leq n g_2(\epsilon) + k(\epsilon) \text{ and } \label{eq:cl_eq_needed1}\\
  &H_{\min}(A_1^n | B)_{p^{(\text{aux})}} \geq \sum_{k=1}^n H_{\min}^{\downarrow, \epsilon}(A_k | A_1^{k-1} B)_p \label{eq:cl_eq_needed2}
\end{align}
where $g_1, g_2$ and $k$ are as in Eq. \ref{eq:target_Hmin_ch_rule}. Then, we can simply use the entropic triangle inequality to show that
\begin{align}
  H^{g_1(\epsilon)}_{\min}(A_1^n | B)_{p} &\geq H_{\min}(A_1^n | B)_{p^{(\text{aux})}} - D_{\max}^{g_1(\epsilon)}(p_{A_1^n B} || p^{(\text{aux})}_{A_1^n B}) \\
  &\geq \sum_{k=1}^n H_{\min}^{\downarrow, \epsilon}(A_k | A_1^{k-1} B)_p - n g_2(\epsilon) - k(\epsilon).
\end{align}
$p^{(\text{aux})}$ will be defined using the distributions which achieve $H_{\min}^{\downarrow, \epsilon}(A_k | A_1^{k-1} B)_p$. For $k \in [n]$, let $q^{(k)}_{A_1^k B}$ be the distribution, such that 
\begin{align}
    & \frac{1}{2}\norm{p_{A_1^k B} - q^{(k)}_{A_1^k B}}_1 \leq \epsilon \\
    & H^{\epsilon, \downarrow}_{\min}(A_k | A_1^{k-1} B)_p = H^{\downarrow}_{\min}(A_k | A_1^{k-1} B)_{q^{(k)}}
\end{align}
It is easy to show that we also have 
\begin{align}
  \frac{1}{2} \norm{p_{A_1^k B} - p_{A_1^{k-1} B} q^{(k)}_{A_k|A_1^{k-1} B}}_1 \leq 2\epsilon.
  \label{eq:cl_ch_rule_dist_bd}
\end{align}
We can define an auxiliary distribution as $q_{A_1^n B} := p_B \prod_{k=1}^n q^{(k)}_{A_k|A_1^{k-1} B}$. For this distribution, the conditional min-entropy satisfies (Eq. \ref{eq:Hmin_down_cl_defn}), 
\begin{align}
  H^{\downarrow}_{\min}(A_k | A_1^{k-1} B)_{q} = H^{\downarrow}_{\min}(A_k | A_1^{k-1} B)_{q^{(k)}} = H^{\epsilon, \downarrow}_{\min}(A_k | A_1^{k-1} B)_p
\end{align}
for every $k$ since $q_{A_k | A_1^{k-1} B} = q^{(k)}_{A_k|A_1^{k-1} B}$. Further, we can bound the min-entropy for this distribution as desired in Eq. \ref{eq:cl_eq_needed2} as 
\begin{align}
  H_{\min}(A_1^n | B)_q \geq H^{\downarrow}_{\min}(A_1^n | B)_q \geq \sum_{k=1}^n H^{\downarrow}_{\min}(A_k | A_1^{k-1} B)_q
\end{align}
Next, we need to ensure that the smooth max-divergence between $p$ and $q$ is relatively small. Our strategy will be to bound the relative entropy between these two distributions and use the substate theorem to convert that to a bound for the smooth max-divergence. We expect to be able to use the following chain rule for the relative entropy
\begin{align}
  D(p_{A_1^k B} || q_{A_1^k B}) = D(p_{A_1^{k-1} B} || q_{A_1^{k-1} B}) + D(p_{A_1^k B} || p_{A_1^{k-1} B} q_{A_k | A_1^{k-1} B}) \label{eq:cl_rel_ent_ch_rule_use}
\end{align} 
to prove that the relative entropy distance between $p_{A_1^n B}$ and $q_{A_1^n B}$ is small. Using this repeatedly along with the fact that $q_{A_k | A_1^{k-1} B} = q^{(k)}_{A_k|A_1^{k-1} B}$, gives us
\begin{align}
  D(p_{A_1^n B} || q_{A_1^n B}) &= \sum_{k=1}^n D(p_{A_1^k B} || p_{A_1^{k-1} B} q^{(k)}_{A_k | A_1^{k-1} B}).
\end{align}
If we could now show that each term inside the summation is some small function in $\epsilon$, $g(\epsilon)$, we could show that the $D(p_{A_1^n B} || q_{A_1^n B})$ is bounded by $n g(\epsilon)$. Eq. \ref{eq:cl_ch_rule_dist_bd} shows that the two distributions in each of the terms are close to each other. However, the relative entropy between them need not be small or even bounded, since it could be that $p_{A_1^k B} \not\ll p_{A_1^{k-1} B} q^{(k)}_{A_k | A_1^{k-1} B}$. To circumvent this technicality, we need to \emph{massage} the distributions $q^{(k)}_{A_k | A_1^{k-1} B}$ by a small amount. This is done by mixing these distributions with the uniform distribution $u_{A_k}$ to produce the distributions
\begin{align}
  r^{(k)}_{A_k|A_1^{k-1} B} = (1-\delta)q^{(k)}_{A_k|A_1^{k-1} B} + \delta u_{A_k}.
\end{align}
For these distributions, we can show that (Lemma \ref{lemm:approx_chain_rule_for_D})
\begin{align}
    & \frac{1}{2}\norm{p_{A_1^k B} - p_{A_1^{k-1} B}r^{(k)}_{A_k|A_1^{k-1} B}}_1 \leq 2\epsilon + \delta 
\end{align}
and
\begin{align}
    & D(p_{A_1^k B} || p_{A_1^{k-1} B} r^{(k)}_{A_k | A_1^{k-1} B}) \leq z(\epsilon, \delta)
\end{align}
where $z(\epsilon, \delta) = O\rndBrk{(\epsilon + \delta) \log \frac{|A|}{\epsilon \delta}}$. Note that $z(\epsilon, \delta)$ can be made $O\rndBrk{\epsilon \log \frac{|A|}{\epsilon}}$ for $\delta= \epsilon$. Instead of using $q_{A_1^n B}$ as the auxiliary distribution, we can use $r_{A_1^n B} := p_B \prod_{k=1}^n r^{(k)}_{A_k|A_1^{k-1} B} $. We can use the same argument as above to bound the relative entropy between $p_{A_1^n B}$ and $r_{A_1^n B}$. This gives us
\begin{align*}
    D(p_{A_1^n B} || r_{A_1^n B}) &= \sum_{k=1}^n D(p_{A_1^k B} || p_{A_1^{k-1} B} r^{(k)}_{A_k | A_1^{k-1} B}) \\
    &\leq n z(\epsilon, \delta).
    \numberthis
    \label{eq:cl_ch_rule_rel_ent_bd}
\end{align*}
Let $\mu:= {z(\epsilon, \delta)}^{1/3}$, then using the substate theorem, we have
\begin{align*}
    D^{\mu}_{\max} (p_{A_1^n B} || r_{A_1^n B}) \leq n \mu + \frac{1}{\mu^2} + \log \frac{1}{1 - \mu^2}. 
\end{align*}
Using the entropic triangle inequality, we get
\begin{align*}
    H_{\min}^{\mu}(A_1^n | B)_p &\geq H_{\min} (A_1^n | B)_r - n \mu - \frac{1}{\mu^2} - \log \frac{1}{1 - \mu^2}\\
    &\geq \sum_{k=1}^n H_{\min}^{\downarrow}(A_k | A_1^{k-1} B)_r - n \mu - \frac{1}{\mu^2} - \log \frac{1}{1 - \mu^2}\\
    &\geq \sum_{k=1}^n H_{\min}^{\downarrow}(A_k | A_1^{k-1} B)_{q^{(k)}} - n \mu - \frac{1}{\mu^2} - \log \frac{1}{1 - \mu^2} \\
    &= \sum_{k=1}^n H_{\min}^{\downarrow, \epsilon}(A_k | A_1^{k-1} B)_{p} - n \mu - \frac{1}{\mu^2} - \log \frac{1}{1 - \mu^2}
\end{align*}
where in the third line, we use the quasi-concavity of $H_{\min}^{\downarrow}$\cite[Pg 73]{TomamichelBook16}. \\

Thus, we have proven a chain rule for $H_{\min}^{\downarrow,\epsilon}$ of the desired form in the classical case. The primary challenge in generalizing this approach to the quantum setting lies in defining the auxiliary state. In the classical case, we could use the product of (suitably massaged) conditional distributions $q^{(k)}_{A_k | A_1^{k-1} B}$ to define the auxiliary distribution $r_{A_1^n B}$. However, the quantum generalisation of such a distribution is not unique. Moreover, such generalisations are all quite difficult to manipulate. A closely related problem is that we need to be able to prove a relative entropy bound similar to that in Eq. \ref{eq:cl_ch_rule_rel_ent_bd} for the auxiliary state. This also turns out to be quite challenging, especially because the quantum chain rules \cite[Theorem 4.1 and Proposition F.1]{Sutter17} for the relative entropy do not yield something as simple and convenient as Eq. \ref{eq:cl_rel_ent_ch_rule_use}. We solve both of these problems indirectly. We will use a simple relative entropy bound between the state $\rho_{A_1^n B}$ and an unwieldy, yet simple exponential generalisation of the auxiliary probability distribution used above. Then, in order to convert this bound to a (measured) relative entropy bound between the state $\rho_{A_1^n B}$ and a simpler and more convenient auxiliary state, we employ the Generalised GT inequality (Lemma \ref{lemm:cond_st_rel_ent_bd}). Through this approach, we are able to delegate the pesky question of the correct generalisation of the auxiliary state to the GT inequality.

\subsection{Proof for the quantum case}

The following lemma uses the GT inequality as described above to create a generalisation of the auxiliary state used in the classical proof and also prove a measured relative entropy bound between the original state and this auxiliary state. 

\begin{lemma}
  \label{lemm:cond_st_rel_ent_bd}
  Suppose $\rho_{A_1^n B}$ is a normalised state and for $1 \leq k\leq n$ the normalised states $\bar{\rho}^{(k)}_{A_1^k B}$ are full rank and satisfy
  \begin{align}
    D\rndBrk{\rho_{A_1^k B}||\bar{\rho}^{(k)}_{A_1^k B}} \leq \epsilon
    \label{eq:cond_st_lemm_rel_ent_cond}
  \end{align}
  for some $\epsilon>0$. Let $\bar{\rho}^{(0)}_{ B}:= \rho_B$ for notational simplicity. Then, the subnormalised state\footnote{In the following and throughout this paper, the product notation $\Pi_{j=0}^k M_j$ represents the operator product $M_0 M_1 \cdots M_k$ and the notation $\Pi_{j=k}^0 M_j$ represents the operator product $M_k M_{k-1} \cdots M_0$.}
  \begin{align}
    \sigma_{A_1^n B} &:= \int_{-\infty}^{\infty} dt \beta_0 (t) \prod_{k=0}^{n-1} \sqBrk{\rndBrk{\bar{\rho}^{(k)}_{A_1^k B}}^{\frac{1-it}{2}} \rndBrk{\bar{\rho}^{(k+1)}_{A_1^{k} B}}^{-\frac{1-it}{2}}} \cdot \bar{\rho}^{(n)}_{A_1^n B} \cdot \prod_{k=n-1}^{0} \sqBrk{\rndBrk{\bar{\rho}^{(k+1)}_{A_1^{k} B}}^{-\frac{1+it}{2}} \rndBrk{\bar{\rho}^{(k)}_{A_1^k B}}^{\frac{1+it}{2}}} \label{eq:sigma_def1}\\
    &= \int_{-\infty}^{\infty} dt \beta_0 (t) \rho_B^{\frac{1-it}{2}} \rndBrk{\bar{\rho}^{(1)}_{B}}^{-\frac{1-it}{2}} \rndBrk{\bar{\rho}^{(1)}_{A_1 B}}^{\frac{1-it}{2}} \rndBrk{\bar{\rho}^{(2)}_{A_1 B}}^{-\frac{1-it}{2}} \rndBrk{\bar{\rho}^{(2)}_{A_1^2 B}}^{\frac{1-it}{2}}\ \cdots \ \rndBrk{\bar{\rho}^{(n)}_{A_1^{n-1} B}}^{-\frac{1-it}{2}} \nonumber\\
    & \bar{\rho}^{(n)}_{A_1^n B} \cdot \rndBrk{\bar{\rho}^{(n)}_{A_1^{n-1} B}}^{-\frac{1+it}{2}}\ \cdots \ \rndBrk{\bar{\rho}^{(2)}_{A_1^2 B}}^{\frac{1+it}{2}} \rndBrk{\bar{\rho}^{(2)}_{A_1 B}}^{-\frac{1+it}{2}} \rndBrk{\bar{\rho}^{(1)}_{A_1 B}}^{\frac{1+it}{2}} \rndBrk{\bar{\rho}^{(1)}_{B}}^{-\frac{1+it}{2}} \rho_B^{\frac{1+it}{2}}. \label{eq:sigma_def2}
  \end{align}
  is such that
  \begin{align}
    D_m (\rho_{A_1^n B} || \sigma_{A_1^n B}) \leq n \epsilon. 
  \end{align}
  Further, the partial states of $\sigma$ are
  \begin{align*}
    \sigma_{A_1^k B} &= \int_{-\infty}^{\infty} dt \beta_0 (t) \prod_{j=0}^{k-1} \sqBrk{\rndBrk{\bar{\rho}^{(j)}_{A_1^j B}}^{\frac{1-it}{2}} \rndBrk{\bar{\rho}^{(j+1)}_{A_1^{j} B}}^{-\frac{1-it}{2}}} \cdot \bar{\rho}^{(k)}_{A_1^k B} \cdot \prod_{j=k-1}^{0} \sqBrk{\rndBrk{\bar{\rho}^{(j+1)}_{A_1^{j} B}}^{-\frac{1+it}{2}} \rndBrk{\bar{\rho}^{(j)}_{A_1^j B}}^{\frac{1+it}{2}}} \numberthis\\
    &= \int_{-\infty}^{\infty} dt \beta_0 (t) \rho_B^{\frac{1-it}{2}} \rndBrk{\bar{\rho}^{(1)}_{B}}^{-\frac{1-it}{2}} \rndBrk{\bar{\rho}^{(1)}_{A_1 B}}^{\frac{1-it}{2}} \rndBrk{\bar{\rho}^{(2)}_{A_1 B}}^{-\frac{1-it}{2}} \rndBrk{\bar{\rho}^{(2)}_{A_1^2 B}}^{\frac{1-it}{2}}\ \cdots \ \rndBrk{\bar{\rho}^{(k)}_{A_1^{k-1} B}}^{-\frac{1-it}{2}} \\
    & \qquad \qquad \bar{\rho}^{(k)}_{A_1^{k} B} \cdot \rndBrk{\bar{\rho}^{(k)}_{A_1^{k-1} B}}^{-\frac{1+it}{2}}\ \cdots \ \rndBrk{\bar{\rho}^{(2)}_{A_1^2 B}}^{\frac{1+it}{2}} \rndBrk{\bar{\rho}^{(2)}_{A_1 B}}^{-\frac{1+it}{2}} \rndBrk{\bar{\rho}^{(1)}_{A_1 B}}^{\frac{1+it}{2}} \rndBrk{\bar{\rho}^{(1)}_{B}}^{-\frac{1+it}{2}} \rho_B^{\frac{1+it}{2}}.
    \numberthis
    \label{eq:sigma_partial_st}
  \end{align*}
  Hence, $\sigma_B = \rho_B$, and $\sigma$ is normalised.
\end{lemma}
\begin{proof}
  \noindent Our approach here is broadly based on the proof of \cite[Theorem 4.1]{Sutter17}. Define the positive operator\footnote{By restricting ourselves to strictly positive operators $\bar{\rho}^{(k)}_{A_1^k B}$, we ensure that this expression is always well-defined.}
  \begin{align}
    Q_{A_1^n B} := \exp\rndBrk{\sum_{k=1}^n \rndBrk{\log \bar{\rho}^{(k)}_{A_1^k B} - \log\bar{\rho}^{(k)}_{A_1^{k-1} B}} + \log \rho_B}.
  \end{align}
  Note that if all the operators above commuted, $Q_{A_1^n B}$ (and also $\sigma_{A_1^n B}$) would be the probability distribution $\bar{\rho}^{(n)}_{A_{n} | A_1^{n-1}B} \bar{\rho}^{(n-1)}_{A_{n-1} | A_1^{n-2}B}\ \cdots\ \bar{\rho}^{(1)}_{A_1 | B}\rho_B$. As is often times the case with the quantum relative entropy, the best quantum generalisation of the above conditional distributions turns out to be an exponential\footnote{Classically, we have that $D(P_{AB}|| Q_{AB}) - D(P_A || Q_A) = D(P_{AB} || P_{A}Q_{B|A})$. Quantumly, this equation is directly best generalised as $D(\rho_{AB}|| \sigma_{AB}) - D(\rho_A || \sigma_A) = D(\rho_{AB} || \exp\rndBrk{\log \sigma_{AB} - \log \sigma_A + \log \rho_A})$. Though, this is not very useful on its own.}. For this operator, we have 
  \begin{align*}
    D(\rho_{A_1^n B} || Q_{A_1^n B}) &= \tr\rndBrk{\rho_{A_1^n B} \rndBrk{\log\rho_{A_1^n B} -  \log Q_{A_1^n B}}} \\
    &= \tr\rho_{A_1^n B} \rndBrk{\log\rho_{A_1^n B} -  \sum_{k=1}^n \rndBrk{\log\bar{\rho}^{(k)}_{A_1^k B} - \log\bar{\rho}^{(k)}_{A_1^{k-1} B}} - \log\rho_{B}} \\
    &= \sum_{k=1}^n \tr\rndBrk{\rho_{A_1^n B} \rndBrk{\log\rho_{A_1^k B} - \log\bar{\rho}^{(k)}_{A_1^k B}}} - \sum_{k=1}^n \tr\rndBrk{\rho_{A_1^n B} \rndBrk{\log\rho_{A_1^{k-1} B} - \log\bar{\rho}^{(k)}_{A_1^{k-1} B}}} \\
    &= \sum_{k=1}^n \rndBrk{D(\rho_{A_1^k B} || \bar{\rho}^{(k)}_{A_1^k B}) - D(\rho_{A_1^{k-1} B} || \bar{\rho}^{(k)}_{A_1^{k-1} B})} \\
    &\leq \sum_{k=1}^n D(\rho_{A_1^k B} || \bar{\rho}^{(k)}_{A_1^k B}) \\
    &\leq n \epsilon
    \numberthis
  \end{align*}
  where in the second last line, we have used the fact that the relative entropy of two normalised states is positive and in the last line, we have used the bound in Eq. \ref{eq:cond_st_lemm_rel_ent_cond}. Next, by using the variational expression for the relative entropy (Eq. \ref{eq:rel_ent_var_rep}), we see that
  \begin{align*}
    n \epsilon &\geq D(\rho_{A_1^n B} || Q_{A_1^n B})\\
    &= \sup_{\omega_{A_1^n B} > 0} \curlyBrk{\tr(\rho_{A_1^n B} \log \omega_{A_1^n B}) + 1 - \tr\exp \rndBrk{\sum_{k=1}^n \rndBrk{\log \bar{\rho}^{(k)}_{A_1^k B} - \log\bar{\rho}^{(k)}_{A_1^{k-1} B}} + \log \rho_B + \log \omega_{A_1^n B}}}. \numberthis
    \label{eq:var_exp_int}
  \end{align*}
  The trace exponential above can be bounded using the Generalised GT inequality (Theorem \ref{thm:gen_golden_thompson}) as 
  \begin{align*}
    \tr&\exp \rndBrk{\sum_{k=1}^n \rndBrk{\log \bar{\rho}^{(k)}_{A_1^k B} - \log\bar{\rho}^{(k)}_{A_1^{k-1} B}} + \log \rho_B + \log \omega_{A_1^n B}} \\
    &\leq \int_{-\infty}^{\infty} dt \beta_0 (t) \tr \bigg(\omega_{A_1^n B}\  \rho_B^{\frac{1-it}{2}} \rndBrk{\bar{\rho}^{(1)}_{B}}^{-\frac{1-it}{2}} \rndBrk{\bar{\rho}^{(1)}_{A_1 B}}^{\frac{1-it}{2}} \rndBrk{\bar{\rho}^{(2)}_{A_1 B}}^{-\frac{1-it}{2}} \rndBrk{\bar{\rho}^{(2)}_{A_1^2 B}}^{\frac{1-it}{2}}\ \cdots \ \rndBrk{\bar{\rho}^{(n)}_{A_1^{n-1} B}}^{-\frac{1-it}{2}} \cdot \\
    & \qquad \qquad \bar{\rho}^{(n)}_{A_1^n B} \cdot \rndBrk{\bar{\rho}^{(n)}_{A_1^{n-1} B}}^{-\frac{1+it}{2}}\ \cdots \ \rndBrk{\bar{\rho}^{(2)}_{A_1^2 B}}^{\frac{1+it}{2}} \rndBrk{\bar{\rho}^{(2)}_{A_1 B}}^{-\frac{1+it}{2}} \rndBrk{\bar{\rho}^{(1)}_{A_1 B}}^{\frac{1+it}{2}} \rndBrk{\bar{\rho}^{(1)}_{B}}^{-\frac{1+it}{2}} \rho_B^{\frac{1+it}{2}} \bigg) \\
    &= \int_{-\infty}^{\infty} dt \beta_0 (t) \tr \bigg(\omega_{A_1^n B}\  
    \prod_{k=0}^{n-1} \sqBrk{\rndBrk{\bar{\rho}^{(k)}_{A_1^k B}}^{\frac{1-it}{2}} \rndBrk{\bar{\rho}^{(k+1)}_{A_1^{k} B}}^{-\frac{1-it}{2}}} \cdot \bar{\rho}^{(n)}_{A_1^n B} \cdot \prod_{k=n-1}^{0} \sqBrk{\rndBrk{\bar{\rho}^{(k+1)}_{A_1^{k} B}}^{-\frac{1+it}{2}} \rndBrk{\bar{\rho}^{(k)}_{A_1^k B}}^{\frac{1+it}{2}}}\bigg)\\
    &= \tr \rndBrk{\omega_{A_1^n B} \sigma_{A_1^n B}}
    \numberthis
    \label{eq:gen_gol_thomp_app}
    % &\sqBrk{\rho_B^{\frac{1-it}{2}} \rndBrk{\bar{\rho}^{(1)}_{B}}^{-\frac{1-it}{2}} \rndBrk{\bar{\rho}^{(1)}_{A_1 B}}^{\frac{1-it}{2}} \rndBrk{\bar{\rho}^{(2)}_{A_1 B}}^{-\frac{1-it}{2}} \rndBrk{\bar{\rho}^{(2)}_{A_1^2 B}}^{\frac{1-it}{2}}\ \cdots \ \rndBrk{\bar{\rho}^{(n)}_{A_1^{n-1} B}}^{-\frac{1-it}{2}}} \rndBrk{\bar{\rho}^{(n)}_{A_1^n B}}
    %  \rndBrk{\bar{\rho}^{(n)}_{A_1^{n-1} B}}^{-\frac{1+it}{2}}\ \cdots \ \rndBrk{\bar{\rho}^{(1)}_{A_1 B}}^{\frac{1+it}{2}} \rndBrk{\bar{\rho}^{(1)}_{B}}^{-\frac{1+it}{2}} \rho_B^{\frac{1+it}{2}} 
  \end{align*}
  for the state $\sigma$ defined in the statement of the lemma. Plugging the bound above into Eq. \ref{eq:var_exp_int}, we get 
  \begin{align*}
    n \epsilon &\geq D(\rho_{A_1^n B} || Q_{A_1^n B})\\
    &\geq \sup_{\omega_{A_1^n B} > 0} \curlyBrk{\tr(\rho_{A_1^n B} \log \omega_{A_1^n B}) + 1 - \tr\rndBrk{\omega_{A_1^n B} \sigma_{A_1^n B}}} \\
    &=D_m (\rho_{A_1^n B} || \sigma_{A_1^n B}). \numberthis
    \label{eq:Dmeas_bd}
  \end{align*}
\end{proof}

We use the following theorem and its corollary to bound the relative entropy between two states which are close to each other and have a bounded max-relative entropy. 
\begin{theorem}[{\cite[Theorem 4]{Audenaert14}}]
  \label{thm:TRE_dist_bd}
  For two normalised quantum states $\rho$ and $\sigma$, and $\delta \in (0,1]$, we have 
  \begin{align}
    D(\rho || (1-\delta) \sigma + \delta \rho ) \leq \frac{1}{2}\norm{\rho - \sigma}_1\log \frac{1}{\delta} .
  \end{align}
\end{theorem}

\begin{corollary}
  \label{cor:TRE_bd}
  Suppose two normalised quantum states $\rho_{AB}$ and $\sigma_{AB}$ are such that 
  \begin{align}
    \frac{1}{2}\norm{\rho_{AB} - \sigma_{AB}}_1 \leq \epsilon
  \end{align}
  for $\epsilon \in [0,1]$ and 
  \begin{align}
    \delta \tau_A \otimes \rho_B \leq \sigma_{AB} 
  \end{align}
  where $\delta \in (0,1)$ and $\tau_A$ is the completely mixed state on register $A$. Then, 
  \begin{align}
    D(\rho_{AB} || \sigma_{AB}) \leq \frac{\epsilon}{1- \delta/ |A|^2} \log \frac{|A|^2}{\delta}.
  \end{align}
\end{corollary}
\begin{proof}
  Note that 
  \begin{align*}
    \frac{\delta}{|A|^2} \rho_{AB} \leq \delta \tau_A \otimes \rho_B \leq \sigma_{AB}.
  \end{align*}
  So, we can write $\sigma$ as
  \begin{align}
    \sigma_{AB} = \rndBrk{1- \frac{\delta}{|A|^2}} \sigma'_{AB}+ \frac{\delta}{|A|^2} \rho_{AB}
  \end{align}
  for some normalised quantum state $\sigma'_{AB}$. For this state, we have 
  \begin{align*}
    \frac{1}{2} \norm{\rho_{AB} - \sigma'_{AB}}_1 &= \frac{1}{2} \norm{\rho_{AB} - \frac{1}{1- \delta/ |A|^2}(\sigma_{AB} - \frac{\delta}{|A|^2} \rho_{AB})}_1 \\
    &= \frac{1}{2} \frac{1}{1- \delta/ |A|^2} \norm{\rho_{AB} - \frac{\delta}{|A|^2} \rho_{AB} - (\sigma_{AB} - \frac{\delta}{|A|^2} \rho_{AB})}_1 \\
    &\leq \frac{\epsilon}{1- \delta/ |A|^2}.
  \end{align*}
  Using Theorem \ref{thm:TRE_dist_bd}, we have that 
  \begin{align*}
    D(\rho_{AB} || \sigma_{AB}) &\leq \frac{1}{2} \norm{\rho_{AB} - \sigma'_{AB}}_1 \log \frac{|A|^2}{\delta}\\
    &\leq \frac{\epsilon}{1- \delta/ |A|^2} \log \frac{|A|^2}{\delta}.
  \end{align*}
\end{proof}
Now we are ready to state and prove the universal chain rule for the smooth min-entropy for quantum states. 
\begin{theorem}
  \label{th:Hmin_eps_chain_rule}
  For a normalised quantum state $\rho_{A_1^n B}$ such that for all $k \in [n]$ the dimension $|A_k| = |A|$, and $\epsilon \in (0,1)$ such that $\mu = \rndBrk{\frac{2\epsilon }{1- \epsilon/ |A|^2} \log \frac{|A|^2}{\epsilon}}^{1/3} = O\rndBrk{\rndBrk{\epsilon\log\frac{|A|}{\epsilon}}^{1/3}}$ lies in $(0,1)$, we have the chain rule
  \begin{align}
    H^{\downarrow, 2\mu + \epsilon/2}_{\min}(A_1^n | B)_{\rho} &\geq \sum_{k=1}^n H^{\downarrow, \epsilon/2}_{\min}(A_k | A_1^{k-1} B)_{\rho} - n \mu - \frac{1}{\mu^2} - \log \frac{1}{1 - \mu^2} - \log\rndBrk{\frac{2}{\mu^2} + \frac{1}{1-\mu}}
  \end{align}
\end{theorem}
\begin{proof}
  \textbf{Case 1:} Let's first consider states $\rho_{A_1^n B}$, which are full rank. \\

  \noindent For every $k \in [n]$, define $\lambda_k := H^{\downarrow, \epsilon}_{\min}(A_k | A_1^{k-1}B)_{\rho}$ and let the state $\tilde{\rho}^{(k)}_{A_1^k B}$ be such that\footnote{Since $H^{\downarrow}_{\min}$ is discontinuous, strictly speaking, states achieving $H^{\downarrow, \epsilon}_{\min}(A_k | A_1^{k-1}B)_{\rho}$ may not exist. To take this into account, we can consider $\lambda_k$ to be arbitrarily close to $H^{\downarrow, \epsilon}_{\min}(A_k | A_1^{k-1}B)_{\rho}$. However, for simplicity we assume that such states exist throughout this paper. }
  \begin{align}
    & P(\rho_{A_1^k B}, \tilde{\rho}^{(k)}_{A_1^k B}) \leq \epsilon \\
    & \tilde{\rho}^{(k)}_{A_1^k B} \leq e^{-\lambda_k} \Id_{A_k} \otimes \tilde{\rho}^{(k)}_{A_1^{k-1} B}.
  \end{align}
  Without loss of generality, we will choose the states $\tilde{\rho}^{(k)}_{A_1^k B}$ to be normalised states, since dividing $\tilde{\rho}^{(k)}_{A_1^k B}$ by a normalising factor $N \leq 1$ only decreases its purified distance from $\rho$ and leaves the operator inequality above invariant. Further, let $\delta \in (0,1)$ be a small parameter. For every $ k \in [n]$, we define the normalised states, 
  \begin{align}
    \bar{\rho}^{(k)}_{A_1^k B} := (1- \delta) \tilde{\rho}^{(k)}_{A_1^k B} + \delta \tau_{A_k} \otimes \rho_{A_1^{k-1} B}
  \end{align}
  where $\tau_{A_k}$ is the completely mixed state on register $A_k$. Also, define $\bar{\rho}_B^{(0)} := \rho_B$ for notational convenience. Since, $\rho_{A_1^{k-1} B}$ is full rank, $\bar{\rho}^{(k)}_{A_1^k B}$ are also full rank for all $k$.\\

  For each $k \in [n]$, these states satisfy
  \begin{align*}
    \frac{1}{2}\norm{\rho_{A_1^{k} B} - \bar{\rho}^{(k)}_{A_1^k B}}_1 &\leq \frac{1-\delta}{2}\norm{\rho_{A_1^{k} B} - \tilde{\rho}^{(k)}_{A_1^k B}}_1 + \frac{\delta}{2}\norm{\rho_{A_1^{k} B} - \tau_{A_k} \otimes \rho_{A_1^{k-1} B}}_1 \\
    &\leq \epsilon + \delta 
    \numberthis
  \end{align*}
  and
  \begin{align}
    D(\rho_{A_1^k B} || \bar{\rho}^{(k)}_{A_1^k B}) \leq \frac{\epsilon + \delta}{1- \delta/ |A|^2} \log \frac{|A|^2}{\delta}.
    \label{eq:Drho_rho_bar_bd}
  \end{align}
  using Corollary \ref{cor:TRE_bd}. Let's define the right-hand side above as 
  \begin{align} 
    z(\epsilon, \delta) := \frac{\epsilon + \delta}{1- \delta/ |A|^2} \log \frac{|A|^2}{\delta}.
  \end{align}
  Note that this can be made small, say $O(\epsilon\log 1/\epsilon)$, by choosing $\delta = \epsilon$ for example. \\

  \noindent Using Lemma \ref{lemm:cond_st_rel_ent_bd}, we have that the state $\sigma$ defined as 
  \begin{align}
    \sigma_{A_1^n B} &:= \int_{-\infty}^{\infty} dt \beta_0 (t) \prod_{k=0}^{n-1} \sqBrk{\rndBrk{\bar{\rho}^{(k)}_{A_1^k B}}^{\frac{1-it}{2}} \rndBrk{\bar{\rho}^{(k+1)}_{A_1^{k} B}}^{-\frac{1-it}{2}}} \cdot \bar{\rho}^{(n)}_{A_1^n B} \cdot \prod_{k=n-1}^{0} \sqBrk{\rndBrk{\bar{\rho}^{(k+1)}_{A_1^{k} B}}^{-\frac{1+it}{2}} \rndBrk{\bar{\rho}^{(k)}_{A_1^k B}}^{\frac{1+it}{2}}} \label{eq:Hmin_ch_rule_th_sigma_def1}\\
    &= \int_{-\infty}^{\infty} dt \beta_0 (t) \rho_B^{\frac{1-it}{2}} \rndBrk{\bar{\rho}^{(1)}_{B}}^{-\frac{1-it}{2}} \rndBrk{\bar{\rho}^{(1)}_{A_1 B}}^{\frac{1-it}{2}} \rndBrk{\bar{\rho}^{(2)}_{A_1 B}}^{-\frac{1-it}{2}} \rndBrk{\bar{\rho}^{(2)}_{A_1^2 B}}^{\frac{1-it}{2}}\ \cdots \ \rndBrk{\bar{\rho}^{(n)}_{A_1^{n-1} B}}^{-\frac{1-it}{2}} \nonumber\\
    & \bar{\rho}^{(n)}_{A_1^n B} \cdot \rndBrk{\bar{\rho}^{(n)}_{A_1^{n-1} B}}^{-\frac{1+it}{2}}\ \cdots \ \rndBrk{\bar{\rho}^{(2)}_{A_1^2 B}}^{\frac{1+it}{2}} \rndBrk{\bar{\rho}^{(2)}_{A_1 B}}^{-\frac{1+it}{2}} \rndBrk{\bar{\rho}^{(1)}_{A_1 B}}^{\frac{1+it}{2}} \rndBrk{\bar{\rho}^{(1)}_{B}}^{-\frac{1+it}{2}} \rho_B^{\frac{1+it}{2}} \label{eq:Hmin_ch_rule_th_sigma_def2}
  \end{align}
  is normalised and satisfies
  \begin{align*}
    D_m (\rho_{A_1^n B} || \sigma_{A_1^n B}) \leq n z(\epsilon, \delta). \numberthis
   % \label{eq:Dmeas_bd}
 \end{align*}
  Further, 
  \begin{align*}
    \sigma_{A_1^k B} &= \int_{-\infty}^{\infty} dt \beta_0 (t) \prod_{j=0}^{k-1} \sqBrk{\rndBrk{\bar{\rho}^{(j)}_{A_1^j B}}^{\frac{1-it}{2}} \rndBrk{\bar{\rho}^{(j+1)}_{A_1^{j} B}}^{-\frac{1-it}{2}}} \cdot \bar{\rho}^{(k)}_{A_1^k B} \cdot \prod_{j=k-1}^{0} \sqBrk{\rndBrk{\bar{\rho}^{(j+1)}_{A_1^{j} B}}^{-\frac{1+it}{2}} \rndBrk{\bar{\rho}^{(j)}_{A_1^j B}}^{\frac{1+it}{2}}}
    \numberthis
    \label{eq:eq:Hmin_ch_rule_th_sigma_partial_st}
  \end{align*}
  for $k \in [n]$. Using the substate theorem\footnote{It is quite remarkable that the substate theorem works with a $D_m$ bound and the generalised GT inequality only yields a $D_m$ bound. All our proofs exploit this fact. Whether it is an incredible coincidence or an indication of the tightness of these bounds: you decide.}, we have that for $\mu := z(\epsilon, \delta)^{1/3}$ (we omit the dependence of $\mu$ on $\epsilon$ and $\delta$ for clarity)
  \begin{align}
    D^{\mu}_{\max}(\rho_{A_1^n B} || \sigma_{A_1^n B}) &\leq \frac{D_m (\rho_{A_1^n B} || \sigma_{A_1^n B}) + 1}{\mu^2} + \log \frac{1}{1 - \mu^2}\\
    &\leq n \mu + \frac{1}{\mu^2} + \log \frac{1}{1 - \mu^2}.
  \end{align}
  A simple application of the entropic triangle inequality in Eq. \ref{eq:ent_tri_ineq_simp} gives us that
  \begin{align*}
    H^{\mu}_{\min}(A_1^n | B)_{\rho} &\geq H_{\min}(A_1^n | B)_{\sigma} - D^{\mu}_{\max}(\rho_{A_1^n B} || \sigma_{A_1^n B}) \\
    &\geq H_{\min}(A_1^n | B)_{\sigma} - n \mu - \frac{1}{\mu^2} - \log \frac{1}{1 - \mu^2} \\
    &\geq H^{\downarrow}_{\min}(A_1^n | B)_{\sigma}  - n \mu - \frac{1}{\mu^2} - \log \frac{1}{1 - \mu^2}\\
    &\geq \sum_{k=1}^n H^{\downarrow}_{\min}(A_k | A_1^{k-1} B)_{\sigma}  - n \mu - \frac{1}{\mu^2} - \log \frac{1}{1 - \mu^2}
    \numberthis
    \label{eq:Hmin_rho_bd_in_sigma}
  \end{align*}
  where we have used the chain rule for $H^{\downarrow}_{\min}$ (Eq. \ref{eq:Hmin_dn_non_sm_ch_rule}). We will now show that $\sigma$ is such that $H^{\downarrow}_{\min}(A_k | A_1^{k-1} B)_{\sigma} \geq H^{ \downarrow, \epsilon}_{\min}(A_k | A_1^{k-1} B)_{\rho}$. \\

  Using the quasi-concavity of $H_{\min}^{\downarrow}$ for $\bar{\rho}^{(k)}$ \cite[Pg 73]{TomamichelBook16}, we have
  \begin{align*}
    H_{\min}^{\downarrow}(A_k | A_1^{k-1} B)_{\bar{\rho}^{(k)}} &\geq \min \curlyBrk{H_{\min}^{\downarrow}(A_k | A_1^{k-1} B)_{\tilde{\rho}^{(k)}}, H_{\min}^{\downarrow}(A_k | A_1^{k-1} B)_{\tau_{A_k}\otimes \rho_{A_1^{k-1}B}}} \\
    &\geq \min \curlyBrk{H_{\min}^{\downarrow}(A_k | A_1^{k-1} B)_{\tilde{\rho}^{(k)}}, \log |A|} \\
    &\geq H_{\min}^{\downarrow}(A_k | A_1^{k-1} B)_{\tilde{\rho}^{(k)}}.
  \end{align*}
  Therefore, we have that 
  \begin{align}
    \bar{\rho}^{(k)}_{A_1^k B} &\leq e^{-\lambda_k} \Id_{A_k} \otimes \bar{\rho}_{A_1^{k-1} B}^{(k)}.
  \end{align}
  This implies that 
  \begin{align*}
    \sigma_{A_1^k B} &= \int_{-\infty}^{\infty} dt \beta_0 (t) \prod_{j=0}^{k-1} \sqBrk{\rndBrk{\bar{\rho}^{(j)}_{A_1^j B}}^{\frac{1-it}{2}} \rndBrk{\bar{\rho}^{(j+1)}_{A_1^{j} B}}^{-\frac{1-it}{2}}} \cdot \bar{\rho}^{(k)}_{A_1^k B} \cdot \prod_{j=k-1}^{0} \sqBrk{\rndBrk{\bar{\rho}^{(j+1)}_{A_1^{j} B}}^{-\frac{1+it}{2}} \rndBrk{\bar{\rho}^{(j)}_{A_1^j B}}^{\frac{1+it}{2}}} \\
    &\leq e^{-\lambda_k} \int_{-\infty}^{\infty} dt \beta_0 (t) \prod_{j=0}^{k-1} \sqBrk{\rndBrk{\bar{\rho}^{(j)}_{A_1^j B}}^{\frac{1-it}{2}} \rndBrk{\bar{\rho}^{(j+1)}_{A_1^{j} B}}^{-\frac{1-it}{2}}} \cdot  \Id_{A_k} \otimes \bar{\rho}^{(k)}_{A_1^{k-1} B} \cdot \prod_{j=k-1}^{0} \sqBrk{\rndBrk{\bar{\rho}^{(j+1)}_{A_1^{j} B}}^{-\frac{1+it}{2}} \rndBrk{\bar{\rho}^{(j)}_{A_1^j B}}^{\frac{1+it}{2}}} \\
    &\leq e^{-\lambda_k} \Id_{A_k} \otimes \int_{-\infty}^{\infty} dt \beta_0 (t) \prod_{j=0}^{k-2} \sqBrk{\rndBrk{\bar{\rho}^{(j)}_{A_1^j B}}^{\frac{1-it}{2}} \rndBrk{\bar{\rho}^{(j+1)}_{A_1^{j} B}}^{-\frac{1-it}{2}}} \cdot \bar{\rho}^{(k-1)}_{A_1^{k-1} B} \cdot \prod_{j=k-2}^{0} \sqBrk{\rndBrk{\bar{\rho}^{(j+1)}_{A_1^{j} B}}^{-\frac{1+it}{2}} \rndBrk{\bar{\rho}^{(j)}_{A_1^j B}}^{\frac{1+it}{2}}} \\
    &= e^{-\lambda_k} \Id_{A_k} \otimes \sigma_{A_1^{k-1} B}.
  \end{align*}
  Or, equivalently that $H^{\downarrow}_{\min}(A_k | A_1^{k-1} B)_{\sigma} \geq \lambda_k = H^{\downarrow, \epsilon}_{\min}(A_k | A_1^{k-1} B)_{\rho}$. Plugging this bound in Eq. \ref{eq:Hmin_rho_bd_in_sigma}, we get
  \begin{align}
    H^{\mu}_{\min}(A_1^n | B)_{\rho} &\geq \sum_{k=1}^n H^{\downarrow, \epsilon}_{\min}(A_k | A_1^{k-1} B)_{\rho} - n \mu - \frac{1}{\mu^2} - \log \frac{1}{1 - \mu^2}.
  \end{align}
  Using Eq. \ref{eq:Hmin_up_dn_reln}, we can upper bound the left-hand side with $H^{\downarrow, 2\mu}_{\min}(A_1^n | B)$ to derive 
  \begin{align}
    H^{\downarrow, 2\mu}_{\min}(A_1^n | B)_{\rho} &\geq \sum_{k=1}^n H^{\downarrow, \epsilon}_{\min}(A_k | A_1^{k-1} B)_{\rho} - n \mu - \frac{1}{\mu^2} - \log \frac{1}{1 - \mu^2} - \log\rndBrk{\frac{2}{\mu^2} + \frac{1}{1-\mu}}.
  \end{align}
  \textbf{Case 2:} The above proves the Theorem for the case when $\rho$ was full rank. If $\rho$ was not full rank, then for an arbitrary $\nu \in (0,1)$, the state $\rho_{A_1^n B}' := (1-\nu) \rho_{A_1^n B} + \nu \tau_{A_1^n B}$, which has full support, satisfies
  \begin{align}
    H^{\downarrow,2\mu}_{\min}(A_1^n | B)_{\rho'} &\geq \sum_{k=1}^n H^{\downarrow, \epsilon}_{\min}(A_k | A_1^{k-1} B)_{\rho'} - n \mu - \frac{1}{\mu^2} - \log \frac{1}{1 - \mu^2} - \log\rndBrk{\frac{2}{\mu^2} + \frac{1}{1-\mu}}
  \end{align}
  which implies that 
  \begin{align}
    H^{\downarrow, 2\mu + \sqrt{2\nu}}_{\min}(A_1^n | B)_{\rho} &\geq \sum_{k=1}^n H^{\downarrow, \epsilon - \sqrt{2\nu}}_{\min}(A_k | A_1^{k-1} B)_{\rho} - n \mu - \frac{1}{\mu^2} - \log \frac{1}{1 - \mu^2} - \log\rndBrk{\frac{2}{\mu^2} + \frac{1}{1-\mu}}
  \end{align}
  for every $\nu \in (0,1)$. Unfortunately, since $H^{\downarrow}_{\min}$ is not continuous, we cannot use continuity to claim the above for $\rho$ itself. We can, however, simply choose $\nu = \frac{\epsilon^2}{8}$, which gives us
  \begin{align}
    H^{\downarrow, 2\mu + \epsilon/2}_{\min}(A_1^n | B)_{\rho} &\geq \sum_{k=1}^n H^{\downarrow, \epsilon/2}_{\min}(A_k | A_1^{k-1} B)_{\rho'} - n \mu - \frac{1}{\mu^2} - \log \frac{1}{1 - \mu^2} - \log\rndBrk{\frac{2}{\mu^2} + \frac{1}{1-\mu}}
  \end{align}
  where $\mu = z(\epsilon, \delta)^{1/3} = \rndBrk{\frac{\epsilon + \delta}{1- \delta/ |A|^2} \log \frac{|A|^2}{\delta}}^{1/3}$. To derive the bound in the Theorem, we make the concrete choice $\delta = \epsilon$. 
\end{proof}

\subsection{Dual relation for the smooth max-entropy}

In this section, we will use duality to prove a universal chain rule for the alternate smooth max-entropy $\bar{H}^{\uparrow, \epsilon}_{0}$. Using Eq. \ref{eq:Hmax_up_dn_reln}, one can use this to decompose the conventional smooth max-entropy as well. 

\begin{corollary}
  \label{cor:H0_eps_chain_rule}
  For a normalised quantum state $\rho_{A_1^n B}$ such that for all $k \in [n]$ the dimension $|A_k| = |A|$, and $\epsilon \in (0,1)$ such that $\mu = \rndBrk{\frac{2\epsilon }{1- \epsilon/ |A|^2} \log \frac{|A|^2}{\epsilon}}^{1/3} = O\rndBrk{\rndBrk{\epsilon\log\frac{|A|}{\epsilon}}^{1/3}}$ lies in $(0,1)$, we have the chain rule
  \begin{align}
    \bar{H}_0^{\uparrow, 2\mu + \epsilon/2}(A_1^n | B)_{\rho} &\leq \sum_{k=1}^n \bar{H}_0^{\uparrow, \epsilon/2}(A_k | A_1^{k-1} B)_{\rho} + n \mu + \frac{1}{\mu^2} + \log \frac{1}{1 - \mu^2} + \log\rndBrk{\frac{2}{\mu^2} + \frac{1}{1-\mu}}.
  \end{align}
\end{corollary}
\begin{proof}
  Let $\rho_{A_1^n B C}$ be a purification of $\rho_{A_1^n B}$. Then, using the universal chain rule for the smooth min-entropy, we have 
  \begin{align}
    H^{\downarrow, 2\mu + \epsilon/2}_{\min}(A_1^n | C)_{\rho} &\geq \sum_{k=n}^1 H^{\downarrow, \epsilon/2}_{\min}(A_k | A_{k+1}^n C)_{\rho} - n \mu - \frac{1}{\mu^2} - \log \frac{1}{1 - \mu^2} - \log\rndBrk{\frac{2}{\mu^2} + \frac{1}{1-\mu}}.
  \end{align}
  Using the duality relation in Eq. \ref{eq:Hmin_dn_dual} for $H^{\downarrow, \epsilon}_{\min}$, we get 
  \begin{align}
    -\bar{H}_0^{\uparrow, 2\mu + \epsilon/2}(A_1^n | B)_{\rho} &\geq \sum_{k=1}^n -\bar{H}_0^{\uparrow, \epsilon/2}(A_k | A_1^{k-1} B)_{\rho} - n \mu - \frac{1}{\mu^2} - \log \frac{1}{1 - \mu^2} - \log\rndBrk{\frac{2}{\mu^2} + \frac{1}{1-\mu}}
  \end{align}
  which implies the statement in the corollary. 
\end{proof}

% approx-EAT proof
\section{Unstructured approximate entropy accumulation}
\begin{sloppypar}
  \begin{figure}
    \centering
    \includegraphics[scale=0.1]{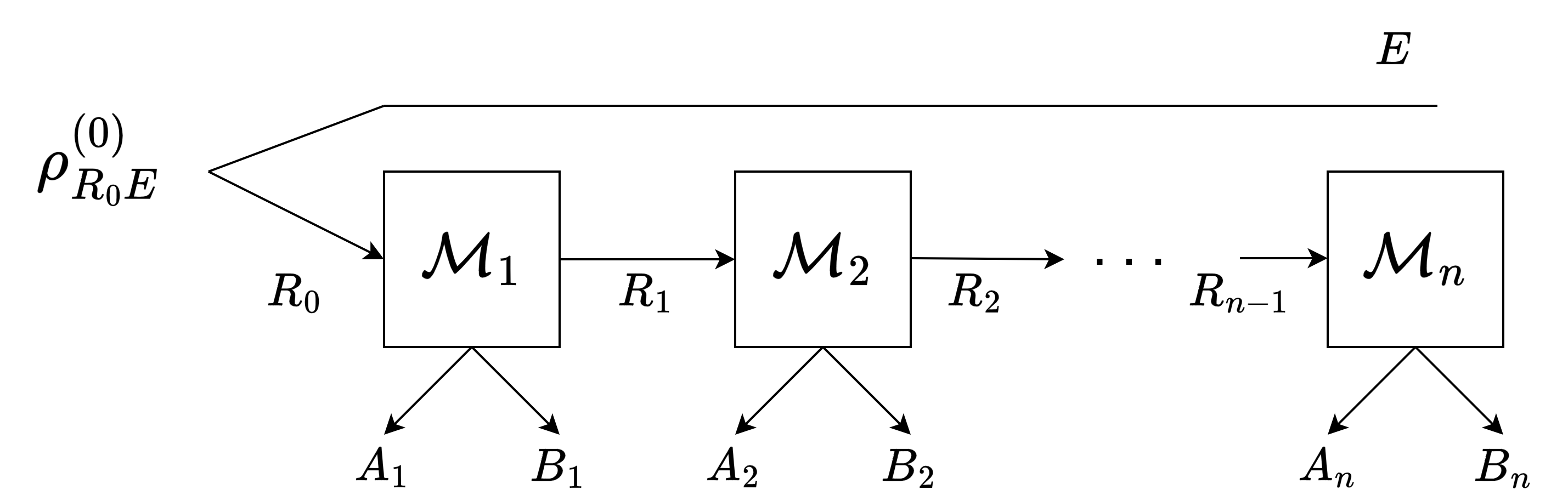}
    \caption{The setting for entropy accumulation. The channels $\cM_k$ are sequentially applied to the registers $R_{k-1}$ and produce the secret information $A_k$ and the side information $B_k$.}
    \label{fig:EAT_setup}
  \end{figure}

  The entropy accumulation theorem (EAT) \cite{Dupuis20} (also see \cite{Metger22}) is one of the most important tools available to bound the smooth min-entropy produced by a sequential process. It is useful for proving the security of quantum key distribution (QKD) \cite{Dupuis20}, device-independent QKD (DIQKD) and randomness expansion \cite{Friedman18,Arnon_Friedman19, Liu21}. EAT can be used to bound the smooth min-entropy of a state produced by a process which has the structure shown in Fig. \ref{fig:EAT_setup}. In such a process, the output state $\rho_{A_1^n B_1^n E} = \mathcal{M}_n \circ \cdots \circ \mathcal{M}_1(\rho^{(0)}_{R_0 E})$ is produced by sequentially applying the channels $\cM_k$ to the registers $R_{k-1}$ for $k \in [n]$ to produce the registers $A_k, B_k$ and $R_k$. In cryptographic applications of this theorem, the registers $A_k$ represent partially secret information produced during the $k$th round of the protocol and the registers $B_k$ denote the side information about these registers, which is leaked to the adversary. EAT provides a lower bound for the smooth min-entropy of the secret information conditioned on the adversary's state in these applications. Specifically, according to the theorem, as long as the output of the process in Fig. \ref{fig:EAT_setup}, $\rho_{A_1^n B_1^n E}$ satisfies the Markov chain 
  \begin{align}
    A_1^{k-1} \leftrightarrow B_1^{k-1} E \leftrightarrow B_k
  \end{align}
  for every $k \in [n]$, the smooth min-entropy satisfies 
  \begin{align}
    H^{\delta}_{\min}(A_1^n | B_1^n E)_{\rho} \geq \sum_{k=1}^n \inf_{\omega_{R_{k-1}\tilde{R}}} H(A_k | B_k \tilde{R})_{\cM_k(\omega)} - c \sqrt{n}
  \end{align}
  where the infimum is over all input states to the channel $\cM_k$ and $c>0$ is a constant that only depends on $|A|$ (the size of registers $A_k$) and $\delta$. In this section, we will prove Theorem \ref{th:approx_EAT}, which is an approximate version of EAT. We show that for any state $\rho_{A_1^n B_1^n E}$ whose partial states $\rho_{A_1^k B_1^k E}$ can be $\epsilon$-approximated as the output of channels $\cM_k$, which sample the side information $B_k$ independent of the previous registers $A_1^{k-1} B_1^{k-1} E$ we can recover a statement similar to EAT\footnote{This seems to be equivalent to requiring that all outputs of the channel $\cM_k$ satisfy $A_1^{k-1} \leftrightarrow B_1^{k-1}E \leftrightarrow B_k$ (see Appendix \ref{sec:all_op_MC_cond}).}. Crucially, Theorem \ref{th:approx_EAT} does not require $\rho_{A_1^n B_1^n E}$ to be produced by a structured process like in Fig. \ref{fig:EAT_setup}; the state may even be produced by a completely parallel process. In fact, our central motivation to develop this theorem was to prove the security of parallel device independent QKD. We do this in our companion work \cite{Marwah24-DIQKD}. For this reason, we call this theorem the \emph{unstructured} approximate EAT.\\

  Previously, we also proved an approximate entropy accumulation theorem in \cite[Theorem 5.1]{Marwah23}. The key difference between these two theorems lies in their applicability: \cite[Theorem 5.1]{Marwah23} applies only to states $\rho$ produced by a sequential process of the same structure as in EAT (Fig. \ref{fig:EAT_setup}), while Theorem \ref{th:approx_EAT} can be applied to states generated by a completely unstructured or parallel process. Furthermore, \cite{Marwah23} considers approximation at the level of channels. Specifically, it considers the diamond norm approximation of the channels $\cM_k$ producing the output state $\rho$ in Fig. \ref{fig:EAT_setup} by \emph{nicer} channels $\cM'_k$, which satisfy certain Markov chain conditions. This is a strong approximation condition, since it requires the outputs of $\cM_k$ and $\cM'_k$ to be approximately equal for all input states. In contrast, in Theorem \ref{th:approx_EAT}, we only need the trace distance between $\rho_{A_1^k B_1^k E}$ and the output of $\cM_k$ to be small for a single input state. We pay the price for these weaker conditions in terms of a much stronger condition on the side information produced by the approximation channels. Additionally, the smoothing parameter in Theorem \ref{th:approx_EAT} depends on the approximation parameter $\epsilon$ and cannot be made arbitrarily small. In comparison, the smoothing parameter for the smooth min-entropy in \cite[Theorem 5.1]{Marwah23} is independent of the approximation parameter and can be chosen arbitrarily small. Consequently, in its current form, Theorem \ref{th:approx_EAT} may have limited utility in cryptographic scenarios where experimental noise or imperfections need to be accounted for, as cryptographic protocols typically require the smoothing parameter to be independent of the noise parameters. Under an analysis using Theorem \ref{th:approx_EAT}, however, the noise parameters would determine the approximation parameter and hence the smoothing parameter. We discuss the possibility of improving this theorem to decouple these parameters in Sec. \ref{sec:decouple_possible}. Nonetheless, Theorem \ref{th:approx_EAT} is a valuable theoretical tool in scenarios where the approximation parameter can be made arbitrarily small, as demonstrated in the security analysis of parallel DIQKD \cite{Marwah24-DIQKD}.
\end{sloppypar}
\begin{theorem}
  \label{th:approx_EAT}
  Let $\epsilon \in (0,1)$ and for every $k \in [n]$ the registers $A_k$ and $B_k$ be such that $|A_k| = |A|$ and $|B_k| = |B|$. Suppose, the state $\rho_{A_1^n B_1^n E}$ is such that for every $k \in [n]$, there exists a channel $\cM_k: R_{k} \rightarrow A_k B_k$ such that for all inputs $X_{R_k}$, 
  \begin{align}
    \tr_{A_k} \circ \cM_k \rndBrk{X_{R_k}} = \tr(X) \theta^{(k)}_{B_k}
    \label{eq:side_info_ind}
  \end{align}
  for some state $\theta^{(k)}_{B_k}$\footnote{This condition can be relaxed to requiring channels $\cM_k$ such that for all input states $\sigma_{A_1^{k-1} B_1^{k-1} E R_{k}}$, the output state $\cM_k (\sigma)$ satisfies the Markov chain condition $A_1^{k-1} \leftrightarrow B_1^{k-1} E \leftrightarrow B_k$. This condition seems to imply the independence condition used in the theorem (see Appendix \ref{sec:all_op_MC_cond}).}, and a state $\tilde{\rho}^{(k, 0)}_{A_1^{k-1} B_1^{k-1} E R_{k}}$ for which 
  \begin{align}
    \frac{1}{2}\norm{\rho_{A_1^k B_1^k E} - \cM_k(\tilde{\rho}^{(k, 0)}_{A_1^{k-1} B_1^{k-1} E R_{k}})}_1\leq \epsilon.
    \label{eq:approx_EAT_approx_cond}
  \end{align}  
  Then, as long as $\mu := \rndBrk{\frac{4 \sqrt{\epsilon} + \epsilon}{1- \epsilon/(|A| |B|)^2} \log\frac{|A|^2 |B|^2}{\epsilon}}^{1/3}= O\rndBrk{\epsilon^{1/6}\rndBrk{\log\frac{|A|^2 |B|^2}{\epsilon}}^{1/3}}$ lies in $(0,1)$, we have the following lower bound for the smooth min-entropy of $\rho$:
  \begin{align*}
    H^{\mu + \epsilon'}_{\min} (A_1^n | B_1^n E)_\rho &\geq \sum_{k=1}^n \inf_{\omega_{R_k \tilde{R}_k}} H (A_k | B_k \tilde{R}_k)_{\cM_k(\omega)}  - 3n\sqrt{\mu} \log(1+ 2|A|) \\
    &\qquad \qquad- \frac{\log(1+2|A|)}{\sqrt{\mu}} \rndBrk{\frac{2}{\mu^2} + 2 \log\frac{1}{1- \mu^2} + g_1(\epsilon', \mu)}
    \numberthis
    \label{eq:approx_EAT}
  \end{align*}
  where the infimum is over all the input states $\omega_{R_k \tilde{R}_k}$ to the channel $\cM_k$ and $g_1(x, y):= - \log(1- \sqrt{1-x^2}) - \log (1-y^2)$. 
\end{theorem}
The proof of this theorem follows almost the same approach as that of Theorem \ref{th:Hmin_eps_chain_rule}. The major difference being that here we use the stronger triangle inequality (Eq. \ref{eq:ent_tri_ineq_alpha_renyi}) to bound the smooth min-entropy with the $\alpha$-R\'enyi conditional entropy of an auxiliary state. We then use the chain rule for these entropies along with the independence conditions on $B_k$ to derive the lower bound above similar to \cite{Dupuis20}. 
\begin{proof}
  \textbf{Case 1:} First let's consider states $\rho_{A_1^n B_1^n E}$, which have full rank. Let $\nu \in (0,1)$ be an arbitrarily chosen small parameter. For every $k \in [n]$, define the states
  \begin{align}
    \tilde{\tilde{\rho}}^{(k, 0)}_{A_1^{k-1} B_1^{k-1} E R_{k}} &:= (1-\nu) \tilde{\rho}^{(k, 0)}_{A_1^{k-1} B_1^{k-1} E R_{k}} + \nu \tau_{A_1^{k-1} B_1^{k-1} E R_{k}} \\
    \tilde{\tilde{\rho}}^{(k)}_{A_1^{k} B_1^{k} E} &:= \cM_k\rndBrk{\tilde{\tilde{\rho}}^{(k, 0)}_{A_1^{k-1} B_1^{k-1} E R_{k}}}
  \end{align}
  We make this modification to $\tilde{\rho}^{(k)}$ states so that $\tilde{\tilde{\rho}}^{(k)}_{A_1^{k-1} B_1^{k-1} E}$ is full rank. For each $k$ these states satisfy
  \begin{align*}
    \frac{1}{2}&\norm{\rho_{A_1^k B_1^k E} - \tilde{\tilde{\rho}}^{(k)}_{A_1^{k} B_1^{k} E}}_1 \\
    &= \frac{1}{2}\norm{\rho_{A_1^k B_1^k E} - (1-\nu) \cM_k\rndBrk{\tilde{\rho}^{(k,0)}_{A_1^{k-1} B_1^{k-1} E R_{k}}} - \nu \cM_k\rndBrk{\tau_{A_1^{k-1} B_1^{k-1} E R_{k}}}}_1 \\
    &\leq \frac{1- \nu}{2}\norm{\rho_{A_1^k B_1^k E} - \cM_k\rndBrk{\tilde{\rho}^{(k)}_{A_1^{k-1} B_1^{k-1} E R_{k}}}}_1 + \nu \\
    &\leq \epsilon + \nu. \numberthis
  \end{align*}
  Now, for each $k$, we define the states 
  \begin{align}
    \omega^{(k,0)}_{A_1^{k-1} B_1^{k-1} R_k E} &:= \rho_{A_1^{k-1} B_1^{k-1} E}^{1/2} \rndBrk{\tilde{\tilde{\rho}}^{(k)}_{A_1^{k-1} B_1^{k-1} E}}^{-1/2} \tilde{\tilde{\rho}}^{(k,0)}_{A_1^{k-1} B_1^{k-1} R_k E} \rndBrk{\tilde{\tilde{\rho}}^{(k)}_{A_1^{k-1} B_1^{k-1} E}}^{-1/2} \rho_{A_1^{k-1} B_1^{k-1} E}^{1/2} \\
    \omega^{(k)}_{A_1^{k} B_1^{k} E} &:= \rho_{A_1^{k-1} B_1^{k-1} E}^{1/2} \rndBrk{\tilde{\tilde{\rho}}^{(k)}_{A_1^{k-1} B_1^{k-1} E}}^{-1/2} \tilde{\tilde{\rho}}^{(k)}_{A_1^{k} B_1^{k} E} \rndBrk{\tilde{\tilde{\rho}}^{(k)}_{A_1^{k-1} B_1^{k-1} E}}^{-1/2} \rho_{A_1^{k-1} B_1^{k-1} E}^{1/2} \\
    &= \cM_k\rndBrk{\omega^{(k,0)}_{A_1^{k-1} B_1^{k-1} R_k E}}.
  \end{align}
  Since, we defined $\tilde{\tilde{\rho}}^{(k)}_{A_1^{k-1} B_1^{k-1} E}$ to be full rank, we have that 
  \begin{align}
    \omega^{(k)}_{A_1^{k-1} B_1^{k-1} E} = \rho_{A_1^{k-1} B_1^{k-1} E}.
  \end{align} 
  Using Lemma \ref{lemm:cond_state_dist}, we have that 
  \begin{align*}
    \frac{1}{2}\norm{\rho_{A_1^{k} B_1^{k} E} - \omega^{(k)}_{A_1^{k} B_1^{k} E}}_1 &\leq (\sqrt{2} +1) P(\rho_{A_1^{k} B_1^{k} E}, \tilde{\tilde{\rho}}_{A_1^{k} B_1^{k} E}) \\
    &\leq (\sqrt{2} +1)\sqrt{2(\epsilon + \nu)}\\
    &\leq 4\sqrt{\epsilon + \nu}. \numberthis 
  \end{align*}
  Let $\delta \in (0,1)$ be a small parameter (to be set equal to $\epsilon$ later). Finally, for every $k \in [n]$, we define the states 
  \begin{align}
    \bar{\rho}^{(k)}_{A_1^{k} B_1^{k} E} :=& (1-\delta) \omega^{(k)}_{A_1^{k} B_1^{k} E} + \delta \tau_{A_k B_k} \otimes \rho_{A_1^{k-1} B_1^{k-1} E} \\
    =& (1-\delta) \omega^{(k)}_{A_1^{k} B_1^{k} E} + \delta \tau_{A_k B_k} \otimes \omega^{(k)}_{A_1^{k-1} B_1^{k-1} E}
  \end{align}
  where $\tau_{A_k B_k}$ is the completely mixed state on the registers $A_k$ and $B_k$. Also, define $\bar{\rho}^{(0)}_E: =\rho_E$. Let $\Delta_k : {R_k \rightarrow A_k B_k}$ be the map which traces out the register $R_k$ and simply outputs $\tau_{A_k B_k}$ and let $\cM^{\delta}_k := (1-\delta) \cM_k + \delta \Delta_k$. Then, we have that 
  \begin{align*}
    \bar{\rho}^{(k)}_{A_1^{k} B_1^{k} E} &= (1-\delta) \omega^{(k)}_{A_1^{k} B_1^{k} E} + \delta \tau_{A_k B_k} \otimes \omega^{(k)}_{A_1^{k-1} B_1^{k-1} E} \\
    &= \rndBrk{(1-\delta) \cM_k + \delta \Delta_k} \rndBrk{\omega^{(k,0)}_{A_1^{k-1} B_1^{k-1} R_k E}} \\
    &= \cM^\delta_k \rndBrk{\omega^{(k,0)}_{A_1^{k-1} B_1^{k-1} R_k E}}. \numberthis
  \end{align*}
  We also have that 
  \begin{align}
    \frac{1}{2}\norm{\bar{\rho}^{(k)}_{A_1^{k} B_1^{k} E} - \rho_{A_1^{k} B_1^{k} E}}_1 \leq 4 \sqrt{\epsilon + \nu} + \delta.
  \end{align}
  Using Corollary \ref{cor:TRE_bd}, this gives us
  \begin{align}
    D(\rho_{A_1^k B_1^k E} || \bar{\rho}^{(k)}_{A_1^k B_1^k E}) \leq \frac{4 \sqrt{\epsilon + \nu} + \delta}{1- \delta/(|A| |B|)^2} \log\frac{|A|^2 |B|^2}{\delta}.
  \end{align}
  We define the above bound as $z(\epsilon+ \nu, \delta)$. Once again, using Lemma \ref{lemm:cond_st_rel_ent_bd} (for $A_k \leftarrow A_k B_k$, $B \leftarrow E$ and $\bar{\rho}^{(k)}_{A_1^k B} \leftarrow \bar{\rho}^{(k)}_{A_1^k B_1^k E}$) we have that auxiliary state \ifcomments\textcolor{red}{(Note here $\bar{\rho}^{(k)}_{A_1^{k-1} B_1^{k-1} E} = \rho_{A_1^{k-1} B_1^{k-1} E} $)}\fi
  \begin{align}
    \sigma_{A_1^n B_1^n E} &:= \int_{-\infty}^{\infty} dt \beta_0 (t) \prod_{k=0}^{n-1} \sqBrk{\rndBrk{\bar{\rho}^{(k)}_{A_1^k B_1^k E}}^{\frac{1-it}{2}} \rndBrk{\bar{\rho}^{(k+1)}_{A_1^{k} B_1^k E}}^{-\frac{1-it}{2}}} \cdot \bar{\rho}^{(n)}_{A_1^n B_1^n E} \cdot \prod_{k= n-1}^{0} \sqBrk{\rndBrk{\bar{\rho}^{(k+1)}_{A_1^{k} B_1^k E}}^{-\frac{1+it}{2}} \rndBrk{\bar{\rho}^{(k)}_{A_1^k B_1^k E}}^{\frac{1+it}{2}}} \label{eq:approx_EAT_sigma_def1}\\
    &= \int_{-\infty}^{\infty} dt \beta_0 (t) \rho_E^{\frac{1-it}{2}} \rndBrk{\bar{\rho}^{(1)}_{E}}^{-\frac{1-it}{2}} \rndBrk{\bar{\rho}^{(1)}_{A_1 B_1 E}}^{\frac{1-it}{2}} \rndBrk{\bar{\rho}^{(2)}_{A_1 B_1 E}}^{-\frac{1-it}{2}} \rndBrk{\bar{\rho}^{(2)}_{A_1^2 B_1^2 E}}^{\frac{1-it}{2}}\ \cdots \ \rndBrk{\bar{\rho}^{(n)}_{A_1^{n-1} B_1^{n-1} E}}^{-\frac{1-it}{2}} \cdot \nonumber \\
    & \qquad \qquad \bar{\rho}^{(n)}_{A_1^n B_1^n E} \cdot \rndBrk{\bar{\rho}^{(n)}_{A_1^{n-1} B_1^{n-1} E}}^{-\frac{1+it}{2}}\ \cdots \ \rndBrk{\bar{\rho}^{(2)}_{A_1^2 B_1^2 E}}^{\frac{1+it}{2}} \rndBrk{\bar{\rho}^{(2)}_{A_1 B_1 E}}^{-\frac{1+it}{2}} \rndBrk{\bar{\rho}^{(1)}_{A_1 B_1 E}}^{\frac{1+it}{2}} \rndBrk{\bar{\rho}^{(1)}_{E}}^{-\frac{1+it}{2}} \rho_E^{\frac{1+it}{2}} \label{eq:approx_EAT_sigma_def2}
  \end{align}
  is a normalised state satisfying
  \begin{align}
    D_m(\rho_{A_1^n B_1^n E} || \sigma_{A_1^n B_1^n E}) \leq nz(\epsilon+\nu, \delta)
    \label{eq:approx_EAT_Dmeas_bd}
  \end{align}
  and 
  \begin{align}
    \sigma&_{A_1^{k} B_1^{k} E} \nonumber \\
    &= \int_{-\infty}^{\infty} dt \beta_0 (t) \prod_{j=0}^{k-1} \sqBrk{\rndBrk{\bar{\rho}^{(j)}_{A_1^j B_1^j E}}^{\frac{1-it}{2}} \rndBrk{\bar{\rho}^{(j+1)}_{A_1^{j} B_1^j E}}^{-\frac{1-it}{2}}} \cdot \bar{\rho}^{(k)}_{A_1^k B_1^k E} \cdot \prod_{j=k-1}^{0} \sqBrk{\rndBrk{\bar{\rho}^{(j+1)}_{A_1^{j} B_1^j E}}^{-\frac{1+it}{2}} \rndBrk{\bar{\rho}^{(j)}_{A_1^j B_1^j E}}^{\frac{1+it}{2}}} \label{eq:approx_EAT_sigma_partial_st1}\\
    &= \cM^{\delta}_k \rndBrk{\int_{-\infty}^{\infty} dt \beta_0 (t) \prod_{j=0}^{k-1} \sqBrk{\rndBrk{\bar{\rho}^{(j)}_{A_1^j B_1^j E}}^{\frac{1-it}{2}} \rndBrk{\bar{\rho}^{(j+1)}_{A_1^{j} B_1^j E}}^{-\frac{1-it}{2}}} \cdot \omega^{(k,0)}_{A_1^{k-1} B_1^{k-1} R_k E} \cdot \prod_{j=k-1}^{0} \sqBrk{\rndBrk{\bar{\rho}^{(j+1)}_{A_1^{j} B_1^j E}}^{-\frac{1+it}{2}} \rndBrk{\bar{\rho}^{(j)}_{A_1^j B_1^j E}}^{\frac{1+it}{2}}}}.
    \label{eq:approx_EAT_sigma_partial_st2}
  \end{align}
  Let $\sigma^{(k,0)}_{A_1^{k-1} B_1^{k-1} R_k E}$ be the input state for $\cM^\delta_k$ above, so that
  \begin{align}
    \sigma_{A_1^k B_1^k E} &= \cM^\delta_k\rndBrk{\sigma^{(k,0)}_{A_1^{k-1} B_1^{k-1} R_k E}}
  \end{align}
  Let's define $\mu := z(\epsilon + \nu, \delta)^{1/3}$. Using the substate theorem (Theorem \ref{thm:substate_th}), we get the following bound from the above relative entropy bound
  \begin{align}
    D^{\mu}_{\max}(\rho_{A_1^n B_1^n E} || \sigma_{A_1^n B_1^n E}) \leq n \mu + \frac{1}{\mu^2} + \log\frac{1}{1- \mu^2}. 
  \end{align}
  Let $\epsilon' \in (0,1)$ be an arbitrary small parameter such that $\epsilon' + \mu < 1$ and let $\alpha \in (1,2]$. We can now use the entropic triangle inequality in Eq. \ref{eq:ent_tri_ineq_alpha_renyi} to derive 
  \begin{align}
    H^{\mu + \epsilon'}_{\min} (A_1^n | B_1^n E)_\rho \geq \tilde{H}^{\uparrow}_{\alpha}(A_1^n | B_1^n E)_\sigma - \frac{\alpha}{\alpha-1}n \mu - \frac{1}{\alpha-1} \rndBrk{\frac{\alpha}{\mu^2} + \alpha \log\frac{1}{1- \mu^2} + g_1(\epsilon', \mu)} 
    \label{eq:approx_EAT_tri_ineq_app}
  \end{align}
  Moreover, using Eq. \ref{eq:approx_EAT_sigma_partial_st2}, we can show that $B_k$ is independent of $A_1^{k-1} B_1^{k-1} E$ in $\sigma$. For $\sigma_{A_1^{k-1} B_1^{k} E} = \tr_{A_k} \circ \cM^{\delta}_k (\sigma^{(k, 0)}_{A_1^{k-1} B_1^{k-1} R_k E})$, we have 
  \begin{align}
    \sigma_{A_1^{k-1} B_1^{k} E}
    &= (1-\delta) \tr_{A_k}\circ \cM_k (\sigma^{(k,0)}_{A_1^{k-1} B_1^{k-1} R_k E}) + \delta \tau_{B_k} \otimes \sigma^{(k, 0)}_{A_1^{k-1} B_1^{k-1} E} \\
    &= \rndBrk{(1-\delta) \theta^{(k)}_{B_k} + \delta \tau_{B_k}} \otimes \sigma_{A_1^{k-1} B_1^{k-1} E}.
  \end{align}
  In particular, we have that $\sigma$ satisfies the Markov chain $A_1^{k-1} \leftrightarrow B_1^{k-1} E \leftrightarrow B_k$. So, we can use \cite[Corollary 3.5]{Dupuis20} to show that for every $k \in [n]$
  \begin{align*}
    \tilde{H}^{\downarrow}_{\alpha}(A_1^k | B_1^k E)_\sigma &\geq \tilde{H}^{\downarrow}_{\alpha}(A_1^{k-1} | B_1^{k-1} E)_\sigma + \inf_{\omega_{R_k \tilde{R}_k}} \tilde{H}^{\downarrow}_{\alpha} (A_k | B_k \tilde{R}_k)_{\cM_k^{\delta}(\omega)} \\
    &\geq \tilde{H}^{\downarrow}_{\alpha}(A_1^{k-1} | B_1^{k-1} E)_\sigma + \inf_{\omega_{R_k \tilde{R}_k}} \tilde{H}^{\downarrow}_{\alpha} (A_k | B_k \tilde{R}_k)_{\cM_k(\omega)} \numberthis
  \end{align*}
  where we use the quasi-concavity of $\tilde{H}^{\downarrow}_{\alpha}$ \cite[Pg 73]{TomamichelBook16} in the second line. Consecutively using this bound in Eq. \ref{eq:approx_EAT_tri_ineq_app} gives us
  \begin{align*}
    H^{\mu + \epsilon'}_{\min} (A_1^n | B_1^n E)_\rho &\geq \sum_{k=1}^n \inf_{\omega_{R_k \tilde{R}_k}} \tilde{H}^{\downarrow}_{\alpha} (A_k | B_k \tilde{R}_k)_{\cM_k(\omega)} - \frac{\alpha}{\alpha-1}n \mu - \frac{1}{\alpha-1} \rndBrk{\frac{\alpha}{\mu^2} + \alpha \log\frac{1}{1- \mu^2} + g_1(\epsilon', \mu)}\\
    &\geq \sum_{k=1}^n \inf_{\omega_{R_k \tilde{R}_k}} H (A_k | B_k \tilde{R}_k)_{\cM_k(\omega)}  -n (\alpha-1)\log^2(1+ 2|A|) - \frac{\alpha}{\alpha-1}n \mu \\
    &\qquad \qquad - \frac{1}{\alpha-1} \rndBrk{\frac{\alpha}{\mu^2} + \alpha \log\frac{1}{1- \mu^2} + g_1(\epsilon', \mu)}
  \end{align*}
  where in the second line we have used \cite[Lemma B.9]{Dupuis20} which is valid as long as $\alpha < 1 + 1/\log(1+2|A|)$. Lastly, we choose $\alpha = 1 + \frac{\sqrt{\mu}}{\log(1+ 2|A|)}$ and use $\alpha <2$ as an upper bound to derive
  \begin{align}
    H^{\mu + \epsilon'}_{\min} (A_1^n | B_1^n E)_\rho &\geq \sum_{k=1}^n \inf_{\omega_{R_k \tilde{R}_k}} H (A_k | B_k \tilde{R}_k)_{\cM_k(\omega)}  - 3n\sqrt{\mu} \log(1+ 2|A|) \nonumber\\
    &\qquad \qquad- \frac{\log(1+2|A|)}{\sqrt{\mu}} \rndBrk{\frac{2}{\mu^2} + 2 \log\frac{1}{1- \mu^2} + g_1(\epsilon', \mu)}
  \end{align}
  where $\mu = z(\epsilon + \nu, \delta)^{1/3}$. The above bound holds true for all $\nu > 0$. Therefore, it also holds for $\nu \rightarrow 0$. To derive the bound in the theorem statement, we fix $\delta=\epsilon$.\\

  \textbf{Case 2:} Lastly, if $\rho_{A_1^n B_1^n E}$ were not full rank then we can always select a full rank state $\rho'_{A_1^n B_1^n E}$ $\varepsilon$-close to $\rho$, which would satisfy Eq. \ref{eq:approx_EAT_approx_cond} and hence the above inequality with $\epsilon \rightarrow \epsilon + \varepsilon$. Now we can take $\varepsilon \rightarrow 0$ and use the continuity of $H_{\min}$ to arrive at the above bound for such $\rho$. 
\end{proof}

In Appendix \ref{sec:testing}, we incorporate testing into the approximate EAT proven above. Testing allows one to condition the state on a classical event and prove a smooth min-entropy lower bound for the conditional state. This enables one to prove much stronger and tighter bounds using EAT.  

\subsection{Dependence of smoothing parameter on the approximation parameter}
\label{sec:decouple_possible}

It is evident that the smoothing parameter must depend on the approximation parameter, $\epsilon$, for the unstructured approximate EAT in its current form. To illustrate this, consider a distribution $p_{A_1^n B}$ where $B = 1$ with probability $(1-\epsilon)$ and $B=0$ otherwise. In this distribution, $A_1^n$ is sampled uniformly at random from $\{0,1\}^n$ if $B=1$, and otherwise set to a constant string. For every $k$, this distribution satisfies
\begin{align}
  p_{A_1^k B} \approx_{O(\epsilon)} \Delta(p_{A_1^{k-1} B})
\end{align}
where $\Delta$ is a channel that disregards its input and uniformly samples a bit $A_k$. Thus, this distribution meets the requirements for the unstructured approximate EAT. However, for this distribution, $H^{\epsilon}_{\min}(A_1^n | B)_p = n$, but for any $\epsilon' \leq \epsilon/2$, we have $H^{\epsilon'}_{\min}(A_1^n | B)_p = O(\log 1/\epsilon)$. This example demonstrates the necessity of the smoothing parameter's dependence on $\epsilon$.\\ 

Nevertheless, it appears possible to decouple this dependence in certain interesting cases, such as sequential DIQKD with leakage and parallel DIQKD. It seems that the smoothing parameter's dependence on $\epsilon$ is required to remove the case of ``correlated failure''. In the aforementioned example, it is necessary to exclude the case where $B=0$ and $A_1^n$ are perfectly determined (highly correlated). Consider, however, the state $\rho_{X_1^n Y_1^n A_1^n B_1^n E}$ produced in a DIQKD protocol with imperfections or leakages ($X_1^n$ and $Y_1^n$ represent Alice and Bob's questions, $A_1^n$ and $B_1^n$ their answers, and $E$ the adversary's register), where for each $k$
\begin{align}
    \rho_{X_1^k Y_1^k A_1^k B_1^k E} \approx_{\epsilon} \cM_k\rndBrk{\rho_{X_1^{k-1} Y_1^{k-1} A_1^{k-1} B_1^{k-1} R_k E}^{(k,0)}}
\end{align} 
for a channel $\cM_k$, which plays the CHSH game between Alice and Bob, as one would expect in DIQKD with no leakage. We can use the unstructured EAT to prove that the entropy of the answers with respect to the questions and $E$ is large for such a state. Drawing an analogy with the example distribution $p$ above, we can deduce that the smoothing parameter must depend on $\epsilon$ to remove the case of correlated failure, whereby all the games fail together. However, this can also be achieved through testing, similar to the approach adopted in \cite{Marwah23_src_corr}. By measuring the winning probability of the CHSH game on a random sample of games, we can determine if it is sufficiently high. If so, we can conclude that the CHSH game must have been played using ``good'' entangled quantum states between the two parties, and an event of correlated failure must not have occurred. While some additional assumptions may be required regarding the side information, this approach could potentially allow for an arbitrary smoothing parameter in such cases. This is an interesting line of research to pursue in the future. 

% second proof for the chain rule
\section{Alternative proof for the universal chain rule}
\label{sec:alternate_proof}

The proofs so far in this paper have followed the approach of creating an auxiliary state using certain conditional states, proving that this state has a small relative entropy from the original state and reducing the original problem to a simpler one in terms of this auxiliary state. This is also the general approach we follow in \cite{Marwah23}. In \cite[Section 4]{Marwah23} we provide a simple and intuitive alternate proof technique that lower bounds the smooth min-entropy for the classical approximately independent registers problem in an \emph{elementwise} fashion. At the time of writing \cite{Marwah23}, we were unable to generalise this proof to the quantum setting. However, using the ideas developed in this paper, we can now generalise this proof. Here we use it to provide an alternative proof of the universal chain rule for the smooth min-entropy. \\

\subsection{Classical proof}
\label{sec:alt_pf_cl}
We begin by once again first sketching the proof in the classical case. Recall that we would like to prove that for a probability distribution $p_{A_1^n B}$,
\begin{align}
    H^{g_1(\epsilon)}_{\min}(A_1^n | B)_{p} &\geq \sum_{k=1}^n H_{\min}^{\downarrow, \epsilon}(A_k | A_1^{k-1} B)_p - n g_2(\epsilon) - k(\epsilon).
\end{align}
Let $\lambda_k := H_{\min}^{\downarrow, \epsilon}(A_k | A_1^{k-1} B)_p$. Following Sec. \ref{sec:uni_ch_rul_cl1}, for $k \in [n]$, let $q^{(k)}_{A_1^k B}$ be the distribution, such that 
\begin{align}
    & \frac{1}{2}\norm{p_{A_1^k B} - q^{(k)}_{A_1^k B}}_1 \leq \epsilon \\
    & H^{\downarrow, \epsilon}_{\min}(A_k | A_1^{k-1} B)_p = \lambda_k = H^{\downarrow}_{\min}(A_k | A_1^{k-1} B)_{q^{(k)}}.
\end{align}
Similar to Eq. \ref{eq:cl_ch_rule_dist_bd}, we also have that 
\begin{align}
  \frac{1}{2} \norm{p_{A_1^k B} - p_{A_1^{k-1} B} q^{(k)}_{A_k|A_1^{k-1} B}}_1 \leq 2\epsilon.
\end{align}
We will use the following lemma from \cite{Marwah23}.
\begin{lemma}[{\cite[Lemma 4.1]{Marwah23}}]
    Suppose $p, q$ are probability distributions on $\mathcal{X}$ such that $\frac{1}{2}\norm{p-q}_1 \leq \epsilon$, then $S \subseteq \mathcal{X}$ defined as $S:= \{x \in \mathcal{X}: p(x) \leq (1+\epsilon^{1/2}) q(x) \}$ is such that $q(S) \geq 1- \epsilon^{1/2}$ and $p(S)\geq 1- \epsilon^{1/2}- \epsilon$.
    \label{lemm:small_tr_norm_implies_small_Dmax}
\end{lemma}
Using the Lemma above, for every $k \in [n]$, we know that the set 
\begin{align*}
    B_k :&= \curlyBrk{(a_1^n, b): p(a_1^k b) > (1+\sqrt{2\epsilon})p(a_1^{k-1} b)q^{(k)}(a_k|a_1^{k-1} b)} \\
    &= \curlyBrk{(a_1^n, b): p(a_k | a_1^{k-1} b) > (1+\sqrt{2\epsilon})q^{(k)}(a_k |a_1^{k-1} b)}
\end{align*}
satisfies $\Pr_p(B_k) \leq 3\sqrt{\epsilon}$. We can now define $L = \sum_{k=1}^n \charFn{B_k}$, which is a random variable that simply counts the number of bad sets $B_k$ an element $(a_1^n, b)$ belongs to. Using the Markov inequality, we have
\begin{align*}
    \Pr_{p}\sqBrk{L > n\epsilon^{\frac{1}{4}}} \leq \frac{\Expect_p[L]}{n\epsilon^{\frac{1}{4}}} \leq 3\epsilon^{\frac{1}{4}}.
\end{align*}
We can define the bad set $\mathcal{B} := \curlyBrk{(a_1^n, b): L(a_1^n, b)>n\epsilon^{\frac{1}{4}}}$. Then for the subnormalised distribution $\tilde{p}_{A_1^n B}$ defined as 
\begin{align*}
    \tilde{p}_{A_1^n B} (a_1^n, b) = 
    \begin{cases}
        p_{A_1^n B} (a_1^n, b) & (a_1^n, b) \not\in \mathcal{B} \\
        0 & \text{else}
    \end{cases},
\end{align*}
we have $P(\tilde{p}_{A_1^n B}, p_{A_1^n B}) \leq \sqrt{6}\epsilon^{1/8}$. Further, note that for every $(a_1^n, b) \not\in \mathcal{B}$, we have 
\begin{align*}
    p(a_1^n | b) &= \prod_{k=1}^n p(a_k | a_1^{k-1}, b) \\
    &= \prod_{k: (a_1^n, b) \not\in B_k} p(a_k | a_1^{k-1}, b) \prod_{k: (a_1^n, b) \in B_k} p(a_k | a_1^{k-1}, b) \\
    &\leq (1+\sqrt{2\epsilon})^n \prod_{k: (a_1^n, b) \not\in B_k} q^{(k)}(a_k | a_1^{k-1}, b) \prod_{k: (a_1^n, b) \in B_k} e^{\log |A| - \lambda_k}\\
    &\leq (1+\sqrt{2\epsilon})^n \prod_{k: (a_1^n, b) \not\in B_k} e^{-\lambda_k} \prod_{k: (a_1^n, b) \in B_k} e^{\log |A| - \lambda_k}\\
    &\leq (1+\sqrt{2\epsilon})^n\ e^{n\epsilon^{1/4}\log |A|} \prod_{k=1}^n e^{-\lambda_k}
\end{align*}
where in the third line we have used the fact that if $(a_1^n, b) \not\in B_k$, then $p(a_k | a_1^{k-1}b) \leq (1+\sqrt{\epsilon})q^{(k)}(a_k | a_1^{k-1}b)$ and if $(a_1^n, b) \in B_k$ then $p(a_k | a_1^{k-1}b) \leq 1 \leq \exp(\log |A| -\lambda_k)$ since $\lambda_k \leq \log |A|$. In the last line we have used the fact that for $(a_1^k, b) \not\in \mathcal{B}$, we have $|\{k \in [n]: (a_1^{n},b) \in B_k\}| = L(a_1^n, b) \leq n\epsilon^{\frac{1}{4}}$. This proves the following lower bound for the smooth min-entropy of $\rho$
\begin{align}
    H^{\downarrow, \sqrt{6}\epsilon^{1/8}}_{\min} (A_1^n | B) \geq \sum_{k=1}^n H^{\downarrow, \epsilon}_{\min}(A_k | A_1^{k-1}B) - n\epsilon^{1/4}\log |A| - n \log (1+\sqrt{\epsilon}).
\end{align}
There are many hurdles to generalising this proof to the quantum setting. Firstly, the correct generalisation of Lemma \ref{lemm:small_tr_norm_implies_small_Dmax} is quite difficult. Secondly, the proof above bounds the conditional probability $p(a_1^n | b)$ for every element $(a_1^n, b) \not\in \mathcal{B}$ differently depending on which bad sets $(\mathcal{B}_k)_k$ it belongs in. It is not obvious how such an argument can be carried out for quantum states. For example, one cannot simply use the eigenvalues of $\rho_{A_1^n B}$ instead of $p(a_1^n b)$ since the eigenvectors for each of the partial states $\rho_{A_1^k B}$ differ. Lastly, in the quantum case it does not seem that the Markov inequality can be used to identify the bad sets and the smoothed state as we did above.\\

The correct generalisation of Lemma \ref{lemm:small_tr_norm_implies_small_Dmax} was proven by \cite{fedja23}. We state it and reproduce its proof in Lemma \ref{lemm:good_proj}. To address the second and third hurdles mentioned above, we modify the classical proof before proceeding to the quantum proof. \\

Instead of figuring out the bad set $\mathcal{B}$ and eliminating the elements belonging to this set to construct the smoothed distribution, we use the substate theorem and the entropic triangle inequality. Define the auxiliary subnormalised distribution 
\begin{align}
    q(a_1^n b) := \delta^{L(a_1^n b)} p(a_1^n b)
    \label{eq:defn_exp_aux_dist}
\end{align}
for some small $\delta \in (0,\frac{1}{|A|})$. The $\delta^{L(a_1^n b)}$ factor in $q$ simply guarantees that elements which have a large $L$ (are \emph{bad}) are significantly damped. This ensures that for all $a_1^n, b$, we have that 
\begin{align*}
    q(a_1^n b) &= \delta^{L(a_1^n b)} p(a_1^n b) \\
    &= p(b) \prod_{k=1}^{n} \delta^{\chi_{B_k}(a_1^n b)} p(a_k | a_1^{k-1}b) \\
    &\leq p(b) \prod_{k=1}^{n} (1+\sqrt{2\epsilon}) e^{-\lambda_k}.
    \numberthis
\end{align*} 
The last inequality above is guaranteed both when $a_1^n, b \not\in \mathcal{B}_k$ and when $a_1^n, b \in \mathcal{B}_k$ since we chose $\delta < \frac{1}{|A|}$. This gives us a lower bound the min-entropy of $q$. Moreover, we can easily show that the relative entropy between $p$ and $q$ is small:
\begin{align*}
    D(p_{A_1^n B} || q_{A_1^n B}) &= \mathbb{E}_{p}\sqBrk{\log \frac{p(A_1^n B)}{q({A_1^n B})}} \\
    &= \mathbb{E}_{p}\sqBrk{L \log \frac{1}{\delta}}\\
    &= 3n\sqrt{\epsilon}\log \frac{1}{\delta}.
    \numberthis
\end{align*} 
Now, we can bound the smooth min-entropy of $p$ in terms of the min-entropy of $q$ by using the substate theorem and the entropic triangle inequality. We will generalise this proof to the quantum case. 

\subsection{Lemmas}

In this section, we will primarily use the following variant of the smooth min-entropy, which has previously appeared in \cite{Anshu20}:
\begin{align}
    \bar{H}_{\min}^{\epsilon}(A|B)_{\rho} &:= \sup \curlyBrk{\lambda \in \mathbb{R} : \text{ for } \tilde{\rho}_{AB} \in B_{\epsilon}(\rho_{AB}) \text{ such that } \tilde{\rho}_{AB} \leq e^{-\lambda}\Id_A \otimes {\rho}_{B}}.
\end{align}
Note that $\bar{H}_{\min}^{\epsilon}(A|B)_{\rho} \geq -\log |A|$ because $\rho_{AB} \leq |A|\Id_A \otimes \rho_B$. Also, $\bar{H}_{\min}^{\epsilon}(A|B)_{\rho} \leq \log \frac{|A|}{1-\epsilon^2}$. To see this note that for any $\tilde{\rho}_{AB}$ and $\lambda$, such that 
\begin{align*}
    & \tilde{\rho}_{AB} \leq e^{-\lambda}\Id_A \otimes {\rho}_{B} \\
    \Rightarrow & \tr \tilde{\rho}_{AB} \leq e^{-\lambda}|A| \\
    \Rightarrow & \lambda \leq \log \frac{|A|}{\tr \tilde{\rho}_{AB}} \numberthis
\end{align*}
where we have simply taken a trace in the second line. Finally, use the fact that $\tr \tilde{\rho}_{AB} \geq 1- \epsilon^2$ for all $\tilde{\rho}_{AB} \in B_{\epsilon}(\rho_{AB})$.\\

This smooth min-entropy can be lower bounded using the $H_{\min}^{\downarrow, \epsilon}$ min-entropy as the following lemma shows. 

\begin{lemma}
    For a normalized quantum state $\rho_{AB}$ and $\epsilon \geq 0$, we have 
    \begin{align}
        \bar{H}_{\min}^{2\epsilon}(A|B)_{\rho} \geq H_{\min}^{\downarrow, \epsilon}(A|B)_{\rho}.
    \end{align}
    \label{lemm:bar_Hmin_Hmin_dn_reln}
\end{lemma}
\begin{proof}
    Let $\lambda = H_{\min}^{\downarrow, \epsilon}(A|B)_{\rho}$ and the state $\tilde{\rho}_{AB} \in B_{\epsilon}(\rho_{AB})$ be such that 
    \begin{align}
        \tilde{\rho}_{AB} \leq e^{-\lambda}\Id_A \otimes \tilde{\rho}_{B}.
    \end{align}
    By Lemma \ref{lemm:cond_state_dist_extra}, we have that $\eta_{AB} := \rho_{B}^{1/2} U_B \tilde{\rho}_{B}^{-1/2}\tilde{\rho}_{AB}\tilde{\rho}_{B}^{-1/2}U_B^{\dagger}\rho_{B}^{1/2}$ satisfies $P(\rho_{AB}, \eta_{AB}) \leq 2\epsilon$. Clearly, this state also satisfies
    \begin{align}
        \eta_{AB} \leq e^{-\lambda}\Id_A \otimes \rho_B.
    \end{align}
\end{proof}

We also require the following operator inequalities relating an operator and its compressions.

\begin{lemma}[Asymmetric pinching {\cite[Lemma 5.2]{Marwah23}}]
    For $t>0$, a positive semidefinite operator $X \geq 0 $ and orthogonal projections $\Pi$ and $\Pi_\perp = \Id - \Pi$, we have that 
    \begin{align}
        X \leq (1+t) \Pi X \Pi + \rndBrk{1+ \frac{1}{t}}\Pi_\perp X \Pi_\perp.
    \end{align}
    \label{lemm:asymm_pinching}
\end{lemma}

\begin{lemma}
    \label{lemm:proj_switching}
    For a positive operator $X \geq 0$ and orthogonal projectors $\Pi$ and $\Pi_{\perp} = \Id - \Pi$, we have 
    \begin{align}
        \Pi_{\perp} X \Pi_{\perp} \leq 2 X + 2 \Pi X \Pi.
    \end{align}
\end{lemma}
\begin{proof}
    We will write the operator $X$ as the block matrix
    \begin{align}
        X = \begin{pmatrix}
            X_1 & X_2 \\
            X_2^\ast & X_3 
        \end{pmatrix}
    \end{align}
    where the blocks are partitioned according to the direct sum $\text{im}(\Pi)\oplus \text{im}(\Pi_{\perp})$. The statement in the Lemma is now equivalent to proving that 
    \begin{align}
        0 \leq  \begin{pmatrix}
            4X_1 & 2X_2 \\
            2X_2^\ast & X_3 
        \end{pmatrix}.
    \end{align}
    Call the matrix above $X'$. Note that using the Schur's complement \cite[Exercise 1.3.5]{Bhatia07} $X \geq 0$ implies that $X_3 \geq 0$, $\text{im}(X_2^\ast) \subseteq \text{im}(X_3)$ and 
    \begin{align}
        X_1 \geq X_2 X_3^{-1} X_2^\ast. 
    \end{align}
    Using Schur's complement characterization again for $X'$, we see that $X'\geq 0$ is equivalent to $X_3 \geq 0$, $\text{im}(X_2^\ast) \subseteq \text{im}(X_3)$ and 
    \begin{align*}
        4 X_1 \geq (2X_2) X_3^{-1} (2X_2)^\ast
    \end{align*}
    which are all true because of the corresponding relations for $X$ itself. 
\end{proof}
The following lemma correctly generalises Lemma \ref{lemm:small_tr_norm_implies_small_Dmax}. It was proven by user:fedja in response to a question by user:noel (pseudonym used by AM) on MathOverflow \cite{fedja23}. Its proof has been reproduced here for the purpose of completeness.
\begin{lemma}[{\cite{fedja23}}]
    For $\epsilon \in [0,1]$, and subnormalized quantum states $\rho$ and $\sigma$ on the finite dimensional Hilbert space $\mathcal{X}$ such that $\frac{1}{2}\norm{\rho - \sigma}_1 \leq \epsilon$, there exists an orthogonal projector $\Pi$ such that 
    \begin{align}
        &\Pi \rho \Pi \leq (1+g_1(\epsilon))\sigma 
        \label{eq:good_proj1}
    \end{align}
    and 
    \begin{align}
        & \tr((\Id - \Pi) \rho ) \leq g_2(\epsilon)
        \label{eq:good_proj2}
    \end{align}
    for the small functions
    \begin{align}
        g_1(\epsilon) &:= \frac{8}{3}(1+\epsilon^{1/3})\epsilon^{1/3}\log\frac{1}{\epsilon} + (1+\epsilon^{1/3} + \epsilon^{2/3}) \epsilon^{1/3} = O\rndBrk{\epsilon^{1/3} \log \frac{1}{\epsilon}} \\
        g_2(\epsilon) &:= 4(1+\epsilon^{1/3})\epsilon^{1/3} + 2\epsilon = O(\epsilon^{1/3})
    \end{align}
    \label{lemm:good_proj}
\end{lemma}
\begin{proof}
    Let $\sigma = \sum_{i} q_i \ket{x_i}\bra{x_i}$ be the eigenvalue decomposition of $\sigma$. Let $\delta_1 \in (0,1)$ be a small parameter, which will be specified later. For every $k \geq 0$, define 
    \begin{align}
        & \mathcal{X}_k := \text{span}\curlyBrk{\ket{x_i} : q_i \in ((1+\delta_1)^{-(k+1)}, (1+\delta_1)^{-k}]} \\
        & d_k := \text{dim}(\mathcal{X}_k)
    \end{align}
    We have $\mathcal{X} = \bigoplus_{k=0}^\infty \mathcal{X}_k$ and $\sigma = \bigoplus_{k=0}^\infty \sigma|_{\mathcal{X}_k}$. Let $P_k$ be the projector on the space $\mathcal{X}_k$ for every $k$. Note that $P_k$ commute with $\sigma$. If we restrict $\sigma$ to the space $\mathcal{X}_k$, we have
    \begin{align}
        \frac{1}{(1+\delta_1)^{k+1}} \Id_{\mathcal{X}_k} \leq \sigma|_{\mathcal{X}_k} \leq \frac{1}{(1+\delta_1)^{k}} \Id_{\mathcal{X}_k}
    \end{align}
    for every $k$. Further, for any projector $\Pi_k$ in the space $\mathcal{X}_k$, we have 
    \begin{align*}
        \Vert \sigma|_{\mathcal{X}_k}^{-\frac{1}{2}} \Pi_k \sigma|_{\mathcal{X}_k} \Pi_k \sigma|_{\mathcal{X}_k}^{-\frac{1}{2}} \Vert_{\infty} &\leq \Vert\sigma|_{\mathcal{X}_k}\Vert_{\infty} \Vert \sigma|_{\mathcal{X}_k}^{-1}\Vert_{\infty} \\
        &\leq 1+ \delta_1
    \end{align*}
    which implies that for any projector $\Pi_k$ in $\mathcal{X}_k$
    \begin{align}
        \Pi_k \sigma|_{\mathcal{X}_k} \Pi_k \leq (1+ \delta_1) \sigma|_{\mathcal{X}_k}.
    \end{align}
    We will choose the projector $\Pi$ to satisfy the lemma to be of the form $\Pi = \bigoplus_{k=0}^\infty \Pi_k$, where each $\Pi_k$ is a projector in the space $\mathcal{X}_k$. With such a projector, we have
    \begin{align*}
        \Pi \sigma \Pi &= \bigoplus_{k=0}^\infty \Pi_k \sigma|_{\mathcal{X}_k} \Pi_k \\
        &\leq (1+ \delta_1) \bigoplus_{k=0}^\infty  \sigma|_{\mathcal{X}_k}\\
        & = (1+ \delta_1) \sigma. 
        \label{eq:proj_sigma_bdd_by_sigma}
        \numberthis
    \end{align*}

    Define $\Delta := (\rho - \sigma)^+$ to be the positive part of $(\rho - \sigma)$ and observe that $\tr(\Delta)\leq 2\epsilon$. Note that $\rho \leq \sigma + \Delta$. Further, let $\Delta_k := P_k \Delta P_k$ and $\mu_k := \tr(\Delta P_k)$ for every $k$. Then, we have that 
    \begin{align*}
        \sum_{k=0}^\infty \mu_k &= \sum_{k=0}^\infty \tr(\Delta P_k) \\
        &= \tr(\Delta) \\
        &\leq 2\epsilon. \numberthis \label{eq:sum_of_mu}
    \end{align*}
    Let $\Delta_k = \sum_{i=1}^{d_k} \nu_{ki} \ket{y_{ki}}\bra{y_{ki}}$ be the eigenvalue decomposition of $\Delta_k$. Let $\delta_2 \in (0,1)$ and $K \in \mathbb{N}$ be two more parameters to be chosen later. Define the subspace $\mathcal{Y}_k$ of $\mathcal{X}_k$ as 
    \begin{align}
        \mathcal{Y}_k := \text{span}\curlyBrk{\ket{y_{ki}} : \nu_{ki} \leq \delta_2 (1+\delta_1)^{-(k+1)}}.
    \end{align}
    Note that the codimension of $\mathcal{Y}_k$ in the space $\mathcal{X}_k$, is at most $\frac{\mu_k}{\delta_2}(1+ \delta_1)^{k+1}$. We further restrict the subspace $\mathcal{Y}_k$ to the space
    \begin{align}
        \mathcal{Z}_k := \mathcal{Y}_k \cap \bigcap_{j=0}^{k-K-1} \text{ker}(P_j \Delta).
    \end{align}
    Since, rank of $P_j$ is $d_j$, the codimension of $\mathcal{Z}_k$ in $\mathcal{Y}_k$ is at most $\sum_{j=0}^{k-K-1} d_j$. Define $\Pi_k$ to be the projector on $\mathcal{Z}_k$. Note that since $\mathcal{Z}_k$ is a subspace of $\mathcal{X}_k$, $\Pi_k$ is a projector in the space $\mathcal{X}_k$. We define the projector 
    \begin{align}
        \Pi := \bigoplus_{k=0}^\infty \Pi_k,
    \end{align}
    which is of the form promised and show that this satisfies the conditions of the lemma. \\
    
    \noindent We begin by showing that $\Pi \Delta \Pi$ can be bounded by a small multiple of $\Pi \sigma \Pi$, specifically
    \begin{align*}
        \Pi \Delta \Pi \leq (2K+1) \delta_2 \Pi \sigma \Pi.
    \end{align*}
    It is sufficient to show that for $\ket{v} = \bigoplus_{k=0}^\infty \ket{v_k}$ such that $\ket{v_k} \in \mathcal{Z}_k$ for every $k$, we have
    \begin{align}
        \braket{v | \Delta | v} \leq (2K+1) \delta_2 \braket{v | \sigma | v}.
    \end{align} 
    The left-hand side above can be expanded as
    \begin{align*}
        \braket{v | \Delta | v} &= \sum_{i=0}^\infty \sum_{j=0}^\infty \braket{v_i | \Delta | v_j} \\
        &= \sum_{i=0}^\infty \braket{v_i | \Delta | v_i} + 2 \sum_{i=0}^\infty \sum_{j=0}^{i-1} \text{Re}\curlyBrk{\braket{v_i | \Delta | v_j}}\\
        &\leq \sum_{i=0}^\infty \braket{v_i | \Delta | v_i} + 2 \sum_{i=0}^\infty \sum_{j=0}^{i-1} |\braket{v_i | \Delta | v_j}|\\
        &= \sum_{i=0}^\infty \braket{v_i | \Delta | v_i} + 2 \sum_{i=0}^\infty \sum_{j=i-K}^{i-1} |\braket{v_i | \Delta | v_j}| + 2 \sum_{i=0}^\infty \sum_{j=0}^{i-K-1} |\braket{v_i | \Delta | v_j}|. \label{eq:three_sums}\numberthis
    \end{align*}
    The first summation above can be bounded as
    \begin{align*}
        \sum_{i=0}^\infty \braket{v_i | \Delta | v_i} &\leq \delta_2 \sum_{i=0}^\infty \frac{1}{(1+\delta_1)^{i+1}} \braket{v_i | v_i}  \\
        &\leq \delta_2 \sum_{i=0}^\infty \braket{v_i | \sigma | v_i} \label{eq:sum1_bd} \numberthis
    \end{align*}
    where we have used the fact that $\ket{v_i} \in \mathcal{Y}_i$ and the maximum eigenvalue of $P_i \Delta P_i = \Delta_i$ in this subspace is at most $\delta_2 \frac{1}{(1+\delta_1)^{i+1}}$, and the fact that the minimum eigenvalue of $\sigma$ in $\mathcal{Z}_i \subseteq \mathcal{X}_i$ is at least $\frac{1}{(1+\delta_1)^{i+1}}$. \\
    
    \noindent We can bound the second term as 
    \begin{align*}
        \sum_{i=0}^\infty \sum_{j=i-K}^{i-1} |\braket{v_i | \Delta | v_j}| &\leq \sum_{i=0}^\infty \sum_{j=i-K}^{i-1} \norm{\Delta^{\frac{1}{2}} \ket{v_j}} \norm{ \Delta^{\frac{1}{2}} \ket{v_i}} \\
        &\leq \sum_{i=0}^\infty \frac{1}{2} \sum_{j=i-K}^{i-1} \rndBrk{\norm{\Delta^{\frac{1}{2}} \ket{v_j}}^2 + \norm{ \Delta^{\frac{1}{2}} \ket{v_i}}^2} \\
        &= K \sum_{i=0}^\infty \norm{ \Delta^{\frac{1}{2}} \ket{v_i}}^2 \\
        &= K \sum_{i=0}^\infty\braket{v_i| \Delta | v_i}
    \end{align*}
    where we have used the Cauchy-Schwarz inequality in the first line along with the fact that $\Delta \geq 0$ and the AM-GM inequality in the second line. This summation can now be bounded using Eq. \ref{eq:sum1_bd}. \\

    The last term in Eq. \ref{eq:three_sums} is a summation over inner products $|\braket{v_i | \Delta | v_j}| = |\braket{v_j | P_j \Delta | v_i}|$ for $j \leq i-K-1$. By our definition for the subspace $\mathcal{Z}_i$, we have that $\ket{v_i} \in \text{ker}(P_j \Delta)$ for these choices of the index $j$. \\

    Putting these together, we have the bound 
    \begin{align*}
        \braket{v | \Delta | v} &\leq \rndBrk{2K+1} \delta_2 \sum_{i=0}^\infty \braket{v_i | \sigma | v_i} \\
        &= \rndBrk{2K+1} \delta_2 \braket{v | \sigma | v}
    \end{align*}
    for every $\ket{v} = \Pi \ket{v}$. Above we have used the fact that $\ket{v_i} \in \mathcal{X}_i$ and $\sigma = \bigoplus_{i=0}^\infty \sigma|_{\mathcal{X}_i}$. Therefore, we have 
    \begin{align}
        \Pi \Delta \Pi \leq \rndBrk{2K+1} \delta_2 \Pi \sigma \Pi
    \end{align}
    and also
    \begin{align*}
        \Pi \rho \Pi &\leq \Pi \rndBrk{\sigma + \Delta} \Pi \\
        & \leq \rndBrk{1 + \rndBrk{2K+1} \delta_2} \Pi \sigma \Pi.\numberthis
        \label{eq:proj_sandwich_bd}
    \end{align*}
    Using Eq. \ref{eq:proj_sigma_bdd_by_sigma} we can remove the sandwiching projectors from $\sigma$ to get the bound
    \begin{align}
        \Pi \rho \Pi &\leq \rndBrk{1 + \rndBrk{2K+1} \delta_2} (1+\delta_1) \sigma.
    \end{align}
    Now, we will show that $\tr(\Pi \sigma)$ is large. Recall that the codimension of $\mathcal{Z}_k$ in $\mathcal{X}_k$ is at most
    \begin{align}
        \frac{\mu_k}{\delta_2}(1+\delta_1)^{k+1} + \sum_{j=0}^{k-K-1} d_j.
        \label{eq:codim_bd}
    \end{align}
    Hence,
    \begin{align*}
        \tr\rndBrk{(\Id - \Pi)\sigma} &= \sum_{k=0}^\infty \tr\rndBrk{(P_k - \Pi_k) \sigma|_{\mathcal{X}_k}} \\
        &\leq \sum_{k=0}^\infty \frac{1}{(1+\delta_1)^{k}} \tr(P_k - \Pi_k) \\
        &\leq \sum_{k=0}^\infty \frac{1}{(1+\delta_1)^{k}} \rndBrk{\frac{\mu_k}{\delta_2} (1+\delta_1)^{k+1} + \sum_{j=0}^{k-K-1} d_j} \\
        &= \frac{1}{\delta_2 }(1+ \delta_1) \sum_{k=0}^\infty \mu_k + \sum_{j=0}^{\infty} d_j \sum_{k=j+K+1}^{\infty} (1+\delta_1)^{-k} \\
        &= \frac{2 \epsilon (1+ \delta_1)}{\delta_2} + \frac{1}{\delta_1}(1+ \delta_1)^{-(K-1)} \sum_{j=0}^{\infty} d_j \frac{1}{(1+ \delta_1)^{j+1}} \\
        &= (1+ \delta_1) \rndBrk{\frac{2\epsilon}{\delta_2}  + \frac{1}{\delta_1}(1+ \delta_1)^{-K}} \numberthis \label{eq:proj_overlap_bd}
    \end{align*}
    where in the second line we have used $\sigma|_{\mathcal{X}_k} \leq (1+ \delta_1)^{-k} \Id_{\mathcal{X}_k}$, in the third line we have used the bound on the codimension of $\mathcal{Z}_k$ in Eq. \ref{eq:codim_bd}, in the fifth line we have used the bound in Eq. \ref{eq:sum_of_mu} and in the last line we have used the fact that $\sum_{j=0}^{\infty} d_j (1+ \delta_1)^{-(j+1)} \leq \tr{\sigma}\leq 1$. \\

    Now all that is left to do is to select the parameters $\delta_1, \delta_2$ and $K$. We choose
    \begin{align}
        \delta_1 &:= \epsilon^{1/3} \\
        \delta_2 &:= \epsilon^{2/3} \\
        K &:= \left\lfloor \frac{2}{\epsilon^{1/3}}\log \frac{1}{\epsilon^{2/3}} \right\rfloor.
    \end{align}
    For these parameters, we get 
    \begin{align}
        \Pi \rho \Pi &\leq \rndBrk{1+\frac{8}{3}\epsilon^{1/3}\log\frac{1}{\epsilon} + \epsilon^{2/3}}(1+ \epsilon^{1/3}) \sigma.
    \end{align}
    Note that $(1+\delta_1)^{-1} \leq (1-\delta_1/2)\leq e^{-\delta_1/2}$ since $\delta_1 \in (0,1)$. This implies 
    \begin{align*}
        \tr\rndBrk{(\Id - \Pi)\sigma} &\leq (1+ \epsilon^{1/3}) \rndBrk{2\epsilon^{1/3}  + \frac{1}{\epsilon^{1/3}}e^{-K \frac{\epsilon^{1/3}}{2}}} \\
        &\leq (1+ \epsilon^{1/3}) \rndBrk{2\epsilon^{1/3}  + \frac{1}{\epsilon^{1/3}}e^{-\frac{\epsilon^{1/3}}{2} \rndBrk{\frac{2}{\epsilon^{1/3}}\log \frac{1}{\epsilon^{2/3}}-1}}} \\
        &\leq (1+ \epsilon^{1/3}) \rndBrk{2\epsilon^{1/3}  + \frac{1}{\epsilon^{1/3}}2\epsilon^{2/3}} \\
        &\leq 4(1+ \epsilon^{1/3}) \epsilon^{1/3}.
    \end{align*}
    Finally, using the fact that $\frac{1}{2}\norm{\rho - \sigma}_1 \leq \epsilon$, we get 
    \begin{align*}
        \tr\rndBrk{(\Id - \Pi) \rho} &\leq \tr\rndBrk{(\Id - \Pi) \sigma} + 2\epsilon \\
        &\leq 4(1+ \epsilon^{1/3}) \epsilon^{1/3} + 2\epsilon. 
        \numberthis
    \end{align*}
\end{proof}

\subsection{Quantum proof}

In this section, we will generalise the classical proof presented earlier to the quantum case. We will prove a statement of the following form for all normalised quantum states $\rho_{A_1^n B}$:
\begin{align}
    H^{g'_1(\epsilon)}_{\min} (A_1^n | B)_{\rho} \geq \sum_{k=1}^n \bar{H}^{\epsilon}_{\min} (A_k | A_1^{k-1} B)_{\rho} - n g'_2(\epsilon) - k(\epsilon)
    % H^{\downarrow, g_1(\epsilon)}_{\min}(A_1^n | B)_{\rho} \geq \sum_{k=1}^n H^{\downarrow, \epsilon}_{\min}(A_k | A_1^{k-1} B)_{\rho} - n g_2(\epsilon) - k(\epsilon),
\end{align}
where $g'_1$ and $g'_2$ are small functions of $\epsilon$ and $k$ is a general function of $\epsilon$. This will be sufficient to prove the universal chain rule. \\

For every $k$, let $\lambda_k := \bar{H}^{\epsilon}_{\min} (A_k | A_1^{k-1} B)_\rho$ and $\tilde{\rho}^{(k)}_{A_1^k B}$ be the subnormalised state such that 
\begin{align}
    & P(\rho_{A_1^k B} , \tilde{\rho}^{(k)}_{A_1^k B}) \leq \epsilon 
    \label{eq:tilde_rho_dist}\\
    & \tilde{\rho}^{(k)}_{A_1^k B} \leq e^{-\lambda_k} \Id_{A_k} \otimes \rho_{A_1^{k-1} B}.
    \label{eq:tilde_rho_min_ent}
\end{align}
Using the Fuchs-van de Graaf inequality, for each $k$, we also have 
\begin{align}
    \frac{1}{2}\norm{\rho_{A_1^k B} - \tilde{\rho}^{(k)}_{A_1^k B}}_1 \leq \epsilon.
\end{align}
We will now define the projectors $P^{(k, l)}_{A_1^k B}$ for every $k \in [n]$ and the label $l \in \{g,b\}$ ($\{\text{good, bad}\}$). We define the \emph{good} projector $P^{(k, g)}_{A_1^k B}$ to be the projector given by Lemma \ref{lemm:good_proj} when applied to the states $\rho_{A_1^k B}$ and $\tilde{\rho}^{(k)}_{A_1^k B}$, so that it satisfies
\begin{align}
    P^{(k, g)}_{A_1^k B} \rho_{A_1^k B} P^{(k, g)}_{A_1^k B} \leq (1+g_1(\epsilon)) \tilde{\rho}^{(k)}_{A_1^k B} 
    \label{eq:good_proj_prop}
\end{align}
and 
\begin{align}
    \tr(P^{(k, g)}_{A_1^k B} \rho_{A_1^k B}) \geq 1 - g_2(\epsilon)
\end{align}
for $g_1$ and $g_2$ as defined in Lemma \ref{lemm:good_proj}. Further, define its orthogonal complement as the \emph{bad} projector,
\begin{align}
    P^{(k,b)}_{A_1^k B} := \Id_{A_1^k B} - P^{(k,g)}_{A_1^k B}.
\end{align}
The label $l$ in $P^{(k,l)}_{A_1^k B}$ will allow us to succinctly refer to these two projectors together. Define the subnormalised state 
\begin{align}
    \eta_{A_1^n B} := P^{(1,g)}_{A_1 B} P^{(2,g)}_{A_1^2 B} \cdots P^{(n,g)}_{A_1^n B}\ \rho_{A_1^n B}\ P^{(n,g)}_{A_1^n B} \cdots P^{(2,g)}_{A_1^2 B} P^{(1,g)}_{A_1 B}.
\end{align}
For this state, observe that 
\begin{align*}
    P^{(1,g)}_{A_1 B} & P^{(2,g)}_{A_1^2 B} \cdots P^{(n,g)}_{A_1^n B}\ \rho_{A_1^n B}\ P^{(n,g)}_{A_1^n B} \cdots P^{(2,g)}_{A_1^2 B} P^{(1,g)}_{A_1 B} \\
    &\leq (1+g_1(\epsilon))P^{(1,g)}_{A_1 B} P^{(2,g)}_{A_1^2 B} \cdots P^{(n-1,g)}_{A_1^{n-1} B}\ \tilde{\rho}^{(n)}_{A_1^n B}\ P^{(n-1,g)}_{A_1^{n-1} B} \cdots P^{(2,g)}_{A_1^2 B} P^{(1,g)}_{A_1 B} \\
    &\leq e^{-\lambda_n} (1+g_1(\epsilon))P^{(1,g)}_{A_1 B} P^{(2,g)}_{A_1^2 B} \cdots P^{(n-1,g)}_{A_1^{n-1} B}\ \Id_{A_n} \otimes \rho_{A_1^{n-1} B}\ P^{(n-1,g)}_{A_1^{n-1} B} \cdots P^{(2,g)}_{A_1^2 B} P^{(1,g)}_{A_1 B} \\
    &= e^{-\lambda_n} (1+g_1(\epsilon)) \Id_{A_n} \otimes P^{(1,g)}_{A_1 B} P^{(2,g)}_{A_1^2 B} \cdots P^{(n-1,g)}_{A_1^{n-1} B}\ \rho_{A_1^{n-1} B}\ P^{(n-1,g)}_{A_1^{n-1} B} \cdots P^{(2,g)}_{A_1^2 B} P^{(1,g)}_{A_1 B} \\
    &\leq \cdots \\
    &\leq e^{-\sum_{k=1}^n \lambda_k} (1+g_1(\epsilon))^n \Id_{A_1^n} \otimes \rho_B
\end{align*}
which implies that $H_{\min}(A_1^n| B)_{\eta} \geq \sum_{k=1}^n \bar{H}^{\epsilon}_{\min} (A_k | A_1^{k-1} B)_{\rho} - n \log(1+ g_1(\epsilon))$. Moreover, the distance between $\eta$ and $\rho$ is 
\begin{align*}
    P(\rho_{A_1^n B},\ \eta_{A_1^n B}) &= P(\rho_{A_1^n B},\ P^{(1,g)}_{A_1 B} P^{(2,g)}_{A_1^2 B} \cdots P^{(n,g)}_{A_1^n B}\ \rho_{A_1^n B}\ P^{(n)}_{A_1^n B} \cdots P^{(2,g)}_{A_1^2 B} P^{(1,g)}_{A_1 B}) \\
    &\leq P(\rho_{A_1^n B},\ P^{(1,g)}_{A_1 B}\ \rho_{A_1^n B}\ P^{(1,g)}_{A_1 B}) \\
    &\qquad \qquad + P(P^{(1,g)}_{A_1 B}\ \rho_{A_1^n B}\ P^{(1,g)}_{A_1 B},\ P^{(1,g)}_{A_1 B} P^{(2,g)}_{A_1^2 B} \cdots P^{(n,g)}_{A_1^n B}\ \rho_{A_1^n B}\ P^{(n,g)}_{A_1^n B} \cdots P^{(2,g)}_{A_1^2 B} P^{(1,g)}_{A_1 B}) \\
    &\leq P(\rho_{A_1^n B}, P^{(1,g)}_{A_1 B}\ \rho_{A_1^n B}\ P^{(1,g)}_{A_1 B}) + P(\rho_{A_1^n B},\ P^{(2,g)}_{A_1^2 B} \cdots P^{(n,g)}_{A_1^n B}\ \rho_{A_1^n B}\ P^{(n,g)}_{A_1^n B} \cdots P^{(2,g)}_{A_1^2 B}) \\
    &\leq \cdots \\
    &\leq \sum_{k=1}^n P(\rho_{A_1^n B},\ P^{(k,g)}_{A_1^k B}\ \rho_{A_1^n B}\ P^{(k,g)}_{A_1^k B}) \\
    &\leq n \sqrt{2 g_2(\epsilon)}
\end{align*}
where we have used the gentle measurement lemma \cite[Proposition 3.14]{Watrous18} for the last line. The distance between $\eta$ and $\rho$ grows like $n$ times a small function and the min-entropy of $\eta$ satisfies $H_{\min}(A_1^n| B)_{\eta} \geq \sum_{k=1}^n \bar{H}^{\epsilon}_{\min} (A_k | A_1^{k-1} B)_{\rho} - n \log(1+g_1(\epsilon))$. If we were able to show that the relative entropy distance between these two states also grew like $n$ times a small function, similar to the purified distance then we could use the substate theorem and the entropic triangle inequality to prove a bound of the required form. However, we can't directly prove that the relative entropy between these two states is small because in general the projectors sandwiching the state $\rho$ in $\eta$ could lead to a situation, where $\rho \not\ll \eta$, which would cause their relative entropy diverge. To remedy this, one could imagine adding a small amount of the complement projector to these projectors, so that we sandwich with $P^{(k,g)}_{A_1^k B} + \delta P^{(k,b)}_{A_1^k B}$ instead of simply $P^{(k)}_{A_1^k B}$. Under no further assumptions on $\rho$, the question of whether the divergence of $\rho$ and $\eta$ is finite or not in this case reduces to the question: given a state $\sigma$, a projector $\Pi$ and its complement $\Pi_{\perp} := \Id - \Pi$, is
\begin{align}
    \sigma \ll (\Pi + \delta \Pi_{\perp}) \sigma (\Pi + \delta \Pi_{\perp}) ?
\end{align}
The answer to this is easily seen to be negative when one consider the pure state $\sigma = \ket{u}\bra{u}$. In this case, the above is equivalent to asking if $\ket{u}\bra{u} \ll \ket{v}\bra{v}$ for some vector $\ket{v} \neq \ket{u}$, which is not true. \\

The next simple remedy that comes to mind is to use a ``pinching type'' disturbance to ensure that the divergence is finite. For $\delta \in (0,1)$, a projector $\Pi$ and its complement $\Pi_{\perp}$, define the CP map $\mathcal{P}_\delta[\Pi]$ as
\begin{align}
    \mathcal{P}_\delta[\Pi](\sigma) := \Pi \sigma \Pi + \delta \Pi_{\perp} \sigma \Pi_{\perp}.
    \label{eq:asymm_pinched_map}
\end{align}
Then, by the pinching inequality, for every $\Pi$, we have 
\begin{align*}
    \sigma &\leq \frac{2}{\delta} \rndBrk{\delta \Pi \sigma \Pi + \delta \Pi_{\perp} \sigma \Pi_{\perp}} \\
    &\leq \frac{2}{\delta} \mathcal{P}_\delta[\Pi](\sigma)
\end{align*}
which implies $D_{\max}(\sigma || \mathcal{P}_\delta[\Pi](\sigma)) \leq O(\log {1}/{\delta})$-- a bound that should be sufficient for our purposes. Therefore, one can try proving the chain rule using the subnormalised state 
\begin{align*}
    \mathcal{P}_\delta[P^{(1,g)}_{A_1 B}]& \circ \mathcal{P}_\delta[P^{(2,g)}_{A_1^2 B}] \circ \cdots \circ \mathcal{P}_\delta[P^{(n,g)}_{A_1^n B}] (\rho_{A_1^n B})\\
    &= \sum_{l_1^n \in \{g,b\}^n} \delta^{\omega_b (l_1^n)} P^{(1,l_1)}_{A_1 B} P^{(2,l_2)}_{A_1^2 B} \cdots P^{(n,l_n)}_{A_1^n B}\ \rho_{A_1^n B}\ P^{(n,l_n)}_{A_1^n B} \cdots P^{(2,l_2)}_{A_1^2 B} P^{(1,l_1)}_{A_1 B} 
    \label{eq:pinching_based_st} 
    \numberthis
\end{align*}
where $\omega_b(l_1^n):= |\{i: l_i = b \}|$ is the weight of $b$ labels in the string $l_1^n$. In our proof below, we do not attempt to directly identify the correct modified state ourselves. Similar to the proofs before, we instead leave this question for the generalised GT inequality. We will use an exponential generalisation of the distribution $q$ in Eq. \ref{eq:defn_exp_aux_dist}. Variants of both of the above remedies make an appearance during the following proof.

\begin{lemma}
    For a full rank state $\rho_{A_1^n B}$ and $\epsilon \in (0,1)$, we have the chain rule
    \begin{align*}
        H^{\mu}_{\min}(A_1^n|B)_{\rho} &\geq \sum_{k=1}^n \bar{H}^\epsilon_{\min}(A_k | A_1^{k-1}B)_{\rho} - n\log\rndBrk{1 + 4|A|^2 \frac{\epsilon}{1-\epsilon^2}} - n\log(1+g_1(\epsilon)) \\
        &  - n \log(1+\epsilon) - n (\mu+\mu^3) - \frac{1}{\mu^2} - \log\frac{1}{1-\mu^2}
        \numberthis
    \end{align*}
    where 
    \begin{align}
        & \mu := \rndBrk{8(1+\epsilon^{1/3})\log\frac{1}{\epsilon}}^{1/3}\epsilon^{1/9} = O\rndBrk{\epsilon^{1/9}\rndBrk{\log\frac{1}{\epsilon}}^{1/3}} \\
        & g_1(\epsilon):= \frac{8}{3}(1+\epsilon^{1/3})\epsilon^{1/3}\log\frac{1}{\epsilon} + (1+\epsilon^{1/3} + \epsilon^{2/3}) \epsilon^{1/3} = O\rndBrk{\epsilon^{1/3} \log \frac{1}{\epsilon}}
    \end{align}
    as long as $\mu \in (0,1)$.
\end{lemma}
\begin{proof}
    We retain the definitions of  $\tilde{\rho}_{A_1^k B},\ \lambda_k$, and $P^{(k, l)}_{A_1^k B}$ from the discussion above. Following the classical proof in Sec. \ref{sec:alt_pf_cl}, define $L_{A_1^n B}:= \sum_{k=1}^n P^{(k,b)}_{A_1^k B}$ as the sum of the ``bad'' projectors. Then, we have that 
    \begin{align*}
        \tr(L_{A_1^n B}\ \rho_{A_1^n B}) &= \sum_{k=1}^n \tr(P^{(k,b)}_{A_1^k B}\ \rho_{A_1^k B}) \\
        &\leq n g_2(\epsilon).
        \numberthis
        \label{eq:alt_pf_exp_bd}
    \end{align*}
    Let $\delta \in (0,1)$ be a parameter to be chosen later. We have that 
    \begin{align*}
        n& g_2(\epsilon)\log \frac{1}{\delta}\\
        &\geq D\rndBrk{\rho_{A_1^n B} || \exp \rndBrk{\log \rho_{A_1^n B} + \log (\delta) L_{A_1^n B}}} \\
        &= \sup_{\omega_{A_1^n B}>0} \curlyBrk{\tr(\rho_{A_1^n B}\ \log \omega_{A_1^n B}) + 1 - \tr\exp\rndBrk{\log \omega_{A_1^n B} + \log (\delta) \sum_{k=1}^n P^{(k,b)}_{A_1^k B} + \log \rho_{A_1^n B} }} \numberthis
        \label{eq:var_rel_ent_bd}
    \end{align*}
    where we have used Eq. \ref{eq:alt_pf_exp_bd} in the first line and the variational expression in Eq. \ref{eq:rel_ent_var_rep} in the second line. For the trace of exponential term in the last expression above, we can use the generalised Golden-Thompson inequality (Theorem \ref{thm:gen_golden_thompson}) as
    \begin{align*}
        &\tr\exp\rndBrk{\log \omega_{A_1^n B} + \log (\delta) \sum_{k=1}^n P^{(k,b)}_{A_1^k B} + \log \rho_{A_1^n B}} \\
        \ifcomments & \textcolor{red}{(e^{Id_A \tensor X_B} = \Id_A \tensor e^{X_B})}\\ \fi
        &\leq \int_{-\infty}^{\infty} dt \beta_0 (t) \tr\rndBrk{\omega_{A_1^n B}\ e^{\frac{1-it}{2}\log(\delta)P^{(1,b)}_{A_1 B}}\cdots\ e^{\frac{1-it}{2}\log(\delta)P^{(n,b)}_{A_1^{n} B}}\ \rho_{A_1^n B}\ e^{\frac{1+it}{2}\log(\delta)P^{(n,b)}_{A_1^{n} B}} \cdots\ e^{\frac{1+it}{2}\log(\delta)P^{(1,b)}_{A_1 B}}} \\
        % &= \int_{-\infty}^{\infty} dt \beta_0 (t) \tr\rndBrk{\omega_{A_1^n B}\ \delta^{\frac{1-it}{2}P^{(1,b)}_{A_1 B}}\cdots\ \delta^{\frac{1-it}{2}P^{(n,b)}_{A_1^{n} B}}\ \rho_{A_1^n B}\ \delta^{\frac{1+it}{2} P^{(n,b)}_{A_1^{n} B}} \cdots\ \delta^{\frac{1+it}{2} P^{(1,b)}_{A_1 B}}} \\
        &= \int_{-\infty}^{\infty} dt \beta_0 (t) \tr\big(\omega_{A_1^n B}\ (\delta^{\frac{1-it}{2}} P^{(1,b)}_{A_1 B} + P^{(1,g)}_{A_1 B})\cdots\ (\delta^{\frac{1-it}{2}} P^{(n,b)}_{A_1^n B} + P^{(n,g)}_{A_1^n B})\ \rho_{A_1^n B}\\
        &\qquad \qquad \qquad \qquad(\delta^{\frac{1+it}{2}} P^{(n,b)}_{A_1^n B} + P^{(n,g)}_{A_1^n B}) \cdots\ (\delta^{\frac{1+it}{2}} P^{(1,b)}_{A_1 B} + P^{(1,g)}_{A_1 B})\big), \numberthis
        \label{eq:gold_thomp_app}
    \end{align*}
    where we have used the fact that for a projector $P$ and $\beta \in \mathbb{C}$, $\exp(\beta P) = e^{\beta} P + (I-P)$ in the last line. Define the subnormalised state 
    \begin{align}
        \eta_{A_1^n B} &:= \int_{-\infty}^{\infty} dt \beta_0 (t)(\delta^{\frac{1-it}{2}} P^{(1,b)}_{A_1 B} + P^{(1,g)}_{A_1 B})\cdots\ (\delta^{\frac{1-it}{2}} P^{(n,b)}_{A_1^n B} + P^{(n,g)}_{A_1^n B})\ \rho_{A_1^n B} \nonumber \\
        &\qquad \qquad \qquad \qquad (\delta^{\frac{1+it}{2}} P^{(n,b)}_{A_1^n B} + P^{(n,g)}_{A_1^n B}) \cdots\ (\delta^{\frac{1+it}{2}} P^{(1,b)}_{A_1 B} + P^{(1,g)}_{A_1 B})
    \end{align}
    Observe that this is just a clever (and correct) way of implementing the first remedy we discussed in the motivation. From Eq. \ref{eq:gold_thomp_app}, we have
    \begin{align}
        &\tr\exp\rndBrk{\log \rho_{A_1^n B} + \log (\delta) \sum_{k=1}^n P^{(k,b)}_{A_1^k B} + \log \omega_{A_1^n B}} \leq \tr(\omega_{A_1^n B}\eta_{A_1^n B})
    \end{align}
    Plugging this in Eq. \ref{eq:var_rel_ent_bd}, we get 
    \begin{align*}
        n g_2(\epsilon)\log \frac{1}{\delta} &\geq \sup_{\omega_{A_1^n B}>0} \curlyBrk{\tr(\rho_{A_1^n B}\ \log \omega_{A_1^n B}) + 1 - \tr(\omega_{A_1^n B}\eta_{A_1^n B})} \\
        &= D_{m}(\rho_{A_1^n B} || \eta_{A_1^n B}). \numberthis
        \label{eq:Dmeas_bd_rho_eta}
    \end{align*}
    We will now show that the state $\eta$ has a small smooth max-relative entropy distance from $\rho$. We use the substate theorem for this. However, we need to be a bit careful since $\eta$ is not normalised.

    \begin{claim}
        For $\mu(\epsilon, \delta) := \rndBrk{g_2(\epsilon)\log\frac{1}{\delta}}^{1/3}$ (we will use the shorthand $\mu$ going forth)
        \begin{align}
            D^{\mu}_{\max}\rndBrk{\rho_{A_1^n B} || \eta_{A_1^n B}} &\leq n (\mu+\mu^3) + \frac{1}{\mu^2} + \log\frac{1}{1-\mu^2}.
        \end{align}
    \end{claim}
    \begin{proof}
        Let $Z := \tr(\eta_{A_1^n B})$. Then, using the data processing inequality on the above $D_{m}(\rho_{A_1^n B} || \eta_{A_1^n B})$ bound, we see that 
        \begin{align}
            \log \frac{1}{Z} \leq n g_2(\epsilon)\log \frac{1}{\delta}.
            \label{eq:Z_lower_bd}
        \end{align}
        We can also upper bound $Z$ as 
        \begin{align*}
            \tr(\eta_{A_1^n B}) &= \int_{-\infty}^{\infty} dt \beta_0 (t) \tr\bigg((\delta^{\frac{1-it}{2}} P^{(1,b)}_{A_1 B} + P^{(1,g)}_{A_1 B})\cdots\ (\delta^{\frac{1-it}{2}} P^{(n,b)}_{A_1^n B} + P^{(n,g)}_{A_1^n B})\ \rho_{A_1^n B} \nonumber \\
            &\qquad \qquad \qquad \qquad (\delta^{\frac{1+it}{2}} P^{(n,b)}_{A_1^n B} + P^{(n,g)}_{A_1^n B}) \cdots\ (\delta^{\frac{1+it}{2}} P^{(1,b)}_{A_1 B} + P^{(1,g)}_{A_1 B})\bigg) \\
            &\leq \int_{-\infty}^{\infty} dt \beta_0 (t) \norm{\delta P^{(1,b)}_{A_1 B} + P^{(1,g)}_{A_1 B}}_{\infty} \tr\bigg((\delta^{\frac{1-it}{2}} P^{(2,b)}_{A_1^2 B} + P^{(2,g)}_{A_1^2 B})\cdots\ (\delta^{\frac{1-it}{2}} P^{(n,b)}_{A_1^n B} + P^{(n,g)}_{A_1^n B})\ \rho_{A_1^n B} \nonumber \\
            &\qquad \qquad \qquad \qquad (\delta^{\frac{1+it}{2}} P^{(n,b)}_{A_1^n B} + P^{(n,g)}_{A_1^n B}) \cdots\ (\delta^{\frac{1+it}{2}} P^{(2,b)}_{A_1^2 B} + P^{(2,g)}_{A_1^2 B})\bigg) \\
            &\leq \int_{-\infty}^{\infty} dt \beta_0 (t) \tr\bigg((\delta^{\frac{1-it}{2}} P^{(2,b)}_{A_1^2 B} + P^{(2,g)}_{A_1^2 B})\cdots\ (\delta^{\frac{1-it}{2}} P^{(n,b)}_{A_1^n B} + P^{(n,g)}_{A_1^n B})\ \rho_{A_1^n B} \nonumber \\
            &\qquad \qquad \qquad \qquad (\delta^{\frac{1+it}{2}} P^{(n,b)}_{A_1^n B} + P^{(n,g)}_{A_1^n B}) \cdots\ (\delta^{\frac{1+it}{2}} P^{(2,b)}_{A_1^2 B} + P^{(2,g)}_{A_1^2 B})\bigg) \\
            &\leq \cdots \\
            &\leq \tr \rho_{A_1^n B} \\
            &=1. \numberthis
        \end{align*}
        Using the substate theorem, we have
        \begin{align*}
            D^{\mu}_{\max}\rndBrk{\rho_{A_1^n B} || \frac{\eta_{A_1^n B}}{Z}} &\leq \frac{D_m \rndBrk{\rho_{A_1^n B} || \frac{\eta_{A_1^n B}}{Z}} + 1}{\mu^2} + \log\frac{1}{1-\mu^2}\\
            &= \frac{D_m \rndBrk{\rho_{A_1^n B} || \eta_{A_1^n B}} - \log\frac{1}{Z} + 1}{\mu^2} + \log\frac{1}{1-\mu^2}\\
            &\leq \frac{n\mu^3 + 1}{\mu^2} + \log\frac{1}{1-\mu^2}\\
            &\leq n \mu + \frac{1}{\mu^2} + \log\frac{1}{1-\mu^2}
            \numberthis
        \end{align*}
        where we have used $Z\leq 1$ in the third line. \\

        We can get rid of $Z$ normalisation factor by using Eq. \ref{eq:Z_lower_bd}
        \begin{align*}
            D^{\mu}_{\max}\rndBrk{\rho_{A_1^n B} || \eta_{A_1^n B}} &= D^{\mu}_{\max}\rndBrk{\rho_{A_1^n B} || \frac{\eta_{A_1^n B}}{Z}} + \log\frac{1}{Z} \\
            &\leq n (\mu+\mu^3) + \frac{1}{\mu^2} + \log\frac{1}{1-\mu^2} 
            \numberthis
            \label{eq:smooth_Dmax_rho_eta}
        \end{align*}
    \end{proof}

    Though $\eta$ has a nice form, it is still difficult to use the properties of the projectors $P^{(k,g)}_{A_1^k B}$ (Eq. \ref{eq:good_proj_prop}) with it. We will now show how we can dominate $\eta$ using a state of the form in Eq. \ref{eq:pinching_based_st}. 
    \begin{claim}
        Define the subnormalised state $\sigma_{A_1^n B} := P_{\sqrt{\delta}}[P^{(1,g)}_{A_1 B}] \circ \cdots \circ P_{\sqrt{\delta}}[P^{(n,g)}_{A_1^n B}](\rho_{A_1^n B})$ for $P_{\sqrt{\delta}}[\Pi](X) = \Pi X \Pi + \sqrt{\delta} \Pi_{\perp} X \Pi_{\perp}$ as defined in Eq. \ref{eq:asymm_pinched_map}. For this state, we have 
        \begin{align}
            D^{\mu}_{\max}\rndBrk{\rho_{A_1^n B} || \sigma_{A_1^n B}} &\leq n (\mu+\mu^3) + n\log(1+\sqrt{\delta}) + \frac{1}{\mu^2} + \log\frac{1}{1-\mu^2}.
            \label{eq:Dmax_rho_sigma_alt_pf}
        \end{align}
    \end{claim}
    \begin{proof}
        Let $X$ be an arbitrary positive operator, and $\Pi$ and $\Pi_{\perp} = \Id - \Pi$ be orthogonal projectors. Asymmetric pinching (Lemma \ref{lemm:asymm_pinching}) with the projectors $\Pi$ and $\Pi_{\perp}$, and parameter $\sqrt{\delta}$ shows that for all $t \in \mathbb{R}$
        \begin{align*}
            (\Pi + \delta^{\frac{1-it}{2}} \Pi_{\perp}) X (\Pi + \delta^{\frac{1+it}{2}} \Pi_{\perp}) &\leq (1+\sqrt{\delta})\Pi X \Pi + \rndBrk{1+\frac{1}{\sqrt{\delta}}} \delta \Pi_{\perp} X \Pi_{\perp} \\
            &= (1+\sqrt{\delta})P_{\sqrt{\delta}}[\Pi](X)
            \numberthis
            \label{eq:reduction_to_pinched_st}
        \end{align*} 
        Using Eq. \ref{eq:reduction_to_pinched_st} repeatedly, we have that for every $t\in \mathbb{R}$
        \begin{align*}
            &(\delta^{\frac{1-it}{2}} P^{(1,b)}_{A_1 B} + P^{(1,g)}_{A_1 B})\cdots\ (\delta^{\frac{1-it}{2}} P^{(n,b)}_{A_1^n B} + P^{(n,g)}_{A_1^n B})\ \rho_{A_1^n B}\ (\delta^{\frac{1+it}{2}} P^{(n,b)}_{A_1^n B} + P^{(n,g)}_{A_1^n B}) \cdots\ (\delta^{\frac{1+it}{2}} P^{(1,b)}_{A_1 B} + P^{(1,g)}_{A_1 B}) \\
            &\leq (1+\sqrt{\delta})(\delta^{\frac{1-it}{2}} P^{(1,b)}_{A_1 B} + P^{(1,g)}_{A_1 B})\cdots\ (\delta^{\frac{1-it}{2}} P^{(n-1,b)}_{A_1^{n-1} B} + P^{(n-1,g)}_{A_1^{n-1} B})\ P_{\sqrt{\delta}}[P^{(n,g)}_{A_1^n B}](\rho_{A_1^n B})\  \\
            & \qquad \qquad \qquad(\delta^{\frac{1+it}{2}} P^{(n-1,b)}_{A_1^{n-1} B} + P^{(n-1,g)}_{A_1^{n-1} B}) \cdots\ (\delta^{\frac{1+it}{2}} P^{(1,b)}_{A_1 B} + P^{(1,g)}_{A_1 B}) \\
            &\leq \cdots \\
            &\leq (1+\sqrt{\delta})^n P_{\sqrt{\delta}}[P^{(1,g)}_{A_1 B}] \circ \cdots \circ P_{\sqrt{\delta}}[P^{(n,g)}_{A_1^n B}](\rho_{A_1^n B})
            \numberthis
        \end{align*}
        which implies that 
        \begin{align}
            \eta_{A_1^n B} \leq (1+\sqrt{\delta})^n P_{\sqrt{\delta}}[P^{(1,g)}_{A_1 B}] \circ \cdots \circ P_{\sqrt{\delta}}[P^{(n,g)}_{A_1^n B}](\rho_{A_1^n B}).
            \label{eq:eta_bd_pinched_st}
        \end{align} 
        This bound essentially says that the $D_{\max}$ between these two states is small:
        \begin{align}
            D_{\max}(\eta_{A_1^n B} || \sigma_{A_1^n B}) \leq n\log(1+\sqrt{\delta})
            \label{eq:Dmax_eta_sigma_bd}
        \end{align}
        Combining this with the smooth $D_{\max}$ bound between $\rho$ and $\eta$ (Eq. \ref{eq:smooth_Dmax_rho_eta}) shows the bound in the claim.  
        \end{proof}

    The only thing left to do now is to show that min-entropy of the state $\sigma$ is large, i.e., $\gtrsim \sum_{k=1}^n \lambda_k$ (recall $\lambda_k := \bar{H}_{\min}^{\epsilon}(A_k | A_1^{k-1}B)_{\rho}$). Note that
    \begin{align}
        \sigma_{A_1^n B} &= \sum_{l_1^n \in \{g,b\}^n} \delta^{\frac{1}{2}\omega_b (l_1^n)} P^{(1,l_1)}_{A_1 B} P^{(2,l_2)}_{A_1^2 B} \cdots P^{(n,l_n)}_{A_1^n B}\ \rho_{A_1^n B}\ P^{(n,l_n)}_{A_1^n B} \cdots P^{(2,l_2)}_{A_1^2 B} P^{(1,l_1)}_{A_1 B} 
        \label{eq:sigma_projector_sum_rep}
    \end{align}
    where $\omega_b(l_1^n):= |\{i: l_i = b \}|$ the weight of $b$ labels in the string $l_1^n$. We will now bound the min-entropy of each of the terms in this summation. Roughly speaking, whenever the label $l_k = g$ (a good projector is applied), $\lambda_k$ amount of min-entropy will be accumulated but when the projector is bad $O(\log|A|)$ amount of min-entropy will be lost. The reason, we are still able to accumulate a large amount of min-entropy for the state $\sigma_{A_1^n B}$ is because these terms are also weighted with the factor $\delta^{\frac{1}{2}\omega_b (l_1^n)}$. If the number of bad projectors in a term is large, then this factor ensures that the contribution of this term is small. \\

    \begin{claim}
        For every $k \in [n]$ and $l_k \in \{g, b\}$, we have 
        \begin{align}
            P^{(k,l_k)}_{A_1^k B}\ \rho_{A_1^k B}\ P^{(k,l_k)}_{A_1^k B} \leq \begin{cases}
                (1+g_1(\epsilon))e^{-\lambda_k} \Id_{A_k} \otimes \rho_{A_1^{k-1} B} & \text{ if } l_k =g \\
                4|A|(1+g_1(\epsilon)) \Id_{A_k} \otimes \rho_{A_1^{k-1} B} & \text{ if } l_k =b
            \end{cases}
        \end{align} 
        which can succinctly be written as 
        \begin{align}
            P^{(k,l_k)}_{A_1^k B}\ \rho_{A_1^k B}\ P^{(k,l_k)}_{A_1^k B} \leq (1+g_1(\epsilon))e^{-\lambda_k \delta(l_k, g)} (4|A|)^{\delta(l_k, b)} \Id_{A_k} \otimes \rho_{A_1^{k-1} B}
        \end{align}
        where $\delta(x,y)$ is the Kronecker delta function ($\delta(x,y) = 1$ if $x=y$ else it is $0$).        
    \end{claim}
    \begin{proof}
        Let's first consider the case when $l_k =g$. In this case, we have 
        \begin{align*}
            P^{(k,g)}_{A_1^k B}\ \rho_{A_1^k B}\ P^{(k,g)}_{A_1^k B} &\leq (1+g_1(\epsilon)) \tilde{\rho}^{(k)}_{A_1^k B} \\
            &\leq (1+g_1(\epsilon))e^{-\lambda_k} \Id_{A_k} \otimes \rho_{A_1^{k-1} B}
            \numberthis
            \label{eq:good_proj_sandwich_bd}
        \end{align*}
        where the first line follows from the definition of the good projectors (Eq. \ref{eq:good_proj_prop}) and the second line follows from Eq. \ref{eq:tilde_rho_min_ent}.\\

        When $l_k= b$, we have 
        \begin{align*}
            P^{(k,b)}_{A_1^k B}\ \rho_{A_1^k B}\ P^{(k,b)}_{A_1^k B} &\leq 2 \rho_{A_1^k B} + 2 P^{(k,g)}_{A_1^k B}\ \rho_{A_1^k B}\ P^{(k,g)}_{A_1^k B} \\
            &\leq 2 \rho_{A_1^k B} + 2(1+g_1(\epsilon)) \tilde{\rho}^{(k)}_{A_1^k B} \\
            &\leq 2(1+g_1(\epsilon)) \rndBrk{|A| \Id_{A_k} \otimes \rho_{A_1^{k-1} B} + e^{-\lambda_k} \Id_{A_k} \otimes \rho_{A_1^{k-1} B}} \\
            &\leq 4|A|(1+g_1(\epsilon)) \Id_{A_k} \otimes \rho_{A_1^{k-1} B}
            \numberthis
            \label{eq:bad_proj_sandwich_bd}
        \end{align*}
        where in the first line we have used Lemma \ref{lemm:proj_switching}, in the second line we have used Eq. \ref{eq:good_proj_prop}, in the third line we have used $\rho_{A_1^k B} \leq |A| \Id_{A_k} \otimes \rho_{A_1^{k-1} B}$ and Eq. \ref{eq:tilde_rho_min_ent}, and in the last line we use $\lambda_k \geq -\log |A|$. 
    \end{proof}
    For every $k$, let $c_k(l_k) := e^{-\lambda_k \delta(l_k, g)} (4|A|)^{\delta(l_k, b)}$ so that 
    \begin{align}
        P^{(k,l_k)}_{A_1^k B}\ \rho_{A_1^k B}\ P^{(k,l_k)}_{A_1^k B} \leq (1+g_1(\epsilon)) c_k(l_k) \Id_{A_k} \otimes \rho_{A_1^{k-1} B}.
    \end{align}
    Now, observe that for each term in the summation in Eq. \ref{eq:sigma_projector_sum_rep}
    \begin{align*}
        P&^{(1,l_1)}_{A_1 B} P^{(2,l_2)}_{A_1^2 B} \cdots P^{(n,l_n)}_{A_1^n B}\ \rho_{A_1^n B}\ P^{(n,l_n)}_{A_1^n B} \cdots P^{(2,l_2)}_{A_1^2 B} P^{(1,l_1)}_{A_1 B} \\
        &\leq (1+g_1(\epsilon)) c_n(l_n) P^{(1,l_1)}_{A_1 B} P^{(2,l_2)}_{A_1^2 B} \cdots P^{(n-1,l_{n-1})}_{A_1^{n-1} B}\ \Id_{A_n} \otimes \rho_{A_1^{n-1} B}\ P^{(n-1,l_{n-1})}_{A_1^{n-1} B} \cdots P^{(2,l_2)}_{A_1^2 B} P^{(1,l_1)}_{A_1 B} \\
        &= (1+g_1(\epsilon)) c_n(l_n) \Id_{A_n} \otimes P^{(1,l_1)}_{A_1 B} P^{(2,l_2)}_{A_1^2 B} \cdots P^{(n-1,l_{n-1})}_{A_1^{n-1} B}\ \rho_{A_1^{n-1} B}\ P^{(n-1,l_{n-1})}_{A_1^{n-1} B} \cdots P^{(2,l_2)}_{A_1^2 B} P^{(1,l_1)}_{A_1 B} \\
        &\leq \cdots \\
        &\leq (1+g_1(\epsilon))^n \rndBrk{\prod_{k=1}^n c_k(l_k)} \Id_{A_1^n} \otimes \rho_{B}.
        \numberthis
    \end{align*}
    Plugging this bound into the expression for $\sigma$ in Eq. \ref{eq:sigma_projector_sum_rep}, we get 
    \begin{align*}
        \sigma_{A_1^n B} &= \sum_{l_1^n \in \{g,b\}^n} \delta^{\frac{1}{2}\omega_b (l_1^n)} P^{(1,l_1)}_{A_1 B} P^{(2,l_2)}_{A_1^2 B} \cdots P^{(n,l_n)}_{A_1^n B}\ \rho_{A_1^n B}\ P^{(n,l_n)}_{A_1^n B} \cdots P^{(2,l_2)}_{A_1^2 B} P^{(1,l_1)}_{A_1 B} \\
        &\leq (1+g_1(\epsilon))^n \rndBrk{\sum_{l_1^n \in \{g,b\}^n} \delta^{\frac{1}{2}\omega_b (l_1^n)} \prod_{k=1}^n c_k(l_k)} \Id_{A_1^n} \otimes \rho_{B}. 
        \numberthis
        \label{eq:Dmax_sigma_rhoB_partial_bd}
    \end{align*}
    Let us now bound the expression
    \begin{align*}
        \sum_{l_1^n \in \{g,b\}^n} \delta^{\frac{1}{2}\omega_b (l_1^n)} \prod_{k=1}^n c_k(l_k) &= \sum_{l_1^n \in \{g,b\}^n} \prod_{k=1}^n \sqrt{\delta}^{\delta(l_k, b)} \prod_{k=1}^n e^{-\lambda_k \delta(l_k, g)} (4|A|)^{\delta(l_k, b)} \\
        &= \sum_{l_1^n \in \{g,b\}^n} \prod_{k=1}^n e^{-\lambda_k \delta(l_k, g)} (4|A|\sqrt{\delta})^{\delta(l_k, b)} \\
        &= \prod_{k=1}^n \rndBrk{e^{-\lambda_k} + 4|A|\sqrt{\delta}}\\
        &= \prod_{k=1}^n e^{-\lambda_k} \rndBrk{1 + 4|A|e^{\lambda_k}\sqrt{\delta}} \\
        &\leq \prod_{k=1}^n e^{-\lambda_k} \rndBrk{1 + 4|A|^2 \frac{\sqrt{\delta}}{1-\epsilon^2}}\\
        &= \rndBrk{1 + 4|A|^2 \frac{\sqrt{\delta}}{1-\epsilon^2}}^n e^{-\sum_{k=1}^n \lambda_k}
        \numberthis
        \label{eq:Hmin_sum_prod_bd}
    \end{align*}
    where we have used $\lambda_k \leq \log\frac{|A|}{1- \epsilon^2}$ in the second last line. Combining Eq. \ref{eq:Dmax_sigma_rhoB_partial_bd} and Eq. \ref{eq:Hmin_sum_prod_bd}, we get 
    \begin{align}
        H_{\min}(A_1^n |B)_{\sigma} &\geq - D_{\max}(\sigma_{A_1^n B}||\Id_{A_1^n} \otimes \rho_{B}) \nonumber\\
        &\geq \sum_{k=1}^n \lambda_k - n\log\rndBrk{1 + 4|A|^2 \frac{\sqrt{\delta}}{1-\epsilon^2}} - n\log(1+g_1(\epsilon)) \nonumber\\
        &\geq \sum_{k=1}^n \bar{H}^\epsilon_{\min}(A_k | A_1^{k-1}B)_{\rho} - n\log\rndBrk{1 + 4|A|^2 \frac{\sqrt{\delta}}{1-\epsilon^2}} - n\log(1+g_1(\epsilon))
    \end{align}
    Finally, we can use the entropic triangle inequality in Eq. \ref{eq:ent_tri_ineq_simp} along with Eq. \ref{eq:Dmax_rho_sigma_alt_pf} to get
    \begin{align*}
        H^{\mu}_{\min}(A_1^n|B)_{\rho} &\geq \sum_{k=1}^n \bar{H}^\epsilon_{\min}(A_k | A_1^{k-1}B)_{\rho} - n\log\rndBrk{1 + 4|A|^2 \frac{\sqrt{\delta}}{1-\epsilon^2}} - n\log(1+g_1(\epsilon)) \\
        &  - n \log(1+\sqrt{\delta}) - n (\mu+\mu^3) - \frac{1}{\mu^2} - \log\frac{1}{1-\mu^2}
        \numberthis
    \end{align*}
    where $\mu = \rndBrk{g_2(\epsilon)\log\frac{1}{\delta}}^{1/3}$. We choose the parameter $\delta=\epsilon^2$, so that we have
    \begin{align*}
        H^{\mu}_{\min}(A_1^n|B)_{\rho} &\geq \sum_{k=1}^n \bar{H}^\epsilon_{\min}(A_k | A_1^{k-1}B)_{\rho} - n\log\rndBrk{1 + 4|A|^2 \frac{\epsilon}{1-\epsilon^2}} - n\log(1+g_1(\epsilon)) \\
        &  - n \log(1+\epsilon) - n (\mu+\mu^3) - \frac{1}{\mu^2} - \log\frac{1}{1-\mu^2}
        \numberthis
    \end{align*}
    for $\mu = \rndBrk{8(1+\epsilon^{1/3})\log\frac{1}{\epsilon}}^{1/3}\epsilon^{1/9} = O(\epsilon^{1/9}\rndBrk{\log\frac{1}{\epsilon}}^{1/3})$. 
\end{proof}

We complete the proof of the universal chain rule by transforming all the entropies to $H^{\downarrow, \epsilon}_{\min}$ in the following theorem.

\begin{theorem}
    For a state $\rho_{A_1^n B}$ and $\epsilon \in (0,1)$, we have the chain rule
    \begin{align*}
        H^{\downarrow, 2\mu + \epsilon/4}_{\min}(A_1^n|B)_{\rho} &\geq \sum_{k=1}^n H^{\downarrow, \epsilon/4}_{\min}(A_k | A_1^{k-1}B)_{\rho} - n\log\rndBrk{1 + 4|A|^2 \frac{\epsilon}{1-\epsilon^2}} - n\log(1+g_1(\epsilon)) \\
        &  - n \log(1+\epsilon) - n (\mu+\mu^3) - \frac{1}{\mu^2} - \log\frac{1}{1-\mu^2} - \log\rndBrk{\frac{2}{\mu^2} + \frac{1}{1-\mu}}
        \numberthis
    \end{align*}
    where 
    \begin{align}
        & \mu := \rndBrk{8(1+\epsilon^{1/3})\log\frac{1}{\epsilon}}^{1/3}\epsilon^{1/9} = O\rndBrk{\epsilon^{1/9}\rndBrk{\log\frac{1}{\epsilon}}^{1/3}} \\
        & g_1(\epsilon):= \frac{8}{3}(1+\epsilon^{1/3})\epsilon^{1/3}\log\frac{1}{\epsilon} + (1+\epsilon^{1/3} + \epsilon^{2/3}) \epsilon^{1/3} = O\rndBrk{\epsilon^{1/3} \log \frac{1}{\epsilon}}
    \end{align}
    as long as $\mu \in (0,1)$.
\end{theorem}
\begin{proof}
    \textbf{Case 1:} If $\rho_{A_1^n B}$ is full rank, then we can use the lemma above along with Eq. \ref{eq:Hmin_up_dn_reln} and Lemma \ref{lemm:bar_Hmin_Hmin_dn_reln} to show that 
    \begin{align*}
        H^{\downarrow, 2\mu}_{\min}(A_1^n|B)_{\rho} &\geq \sum_{k=1}^n H^{\downarrow, \epsilon/2}_{\min}(A_k | A_1^{k-1}B)_{\rho} - n\log\rndBrk{1 + 4|A|^2 \frac{\epsilon}{1-\epsilon^2}} - n\log(1+g_1(\epsilon)) \\
        &  - n \log(1+\epsilon) - n (\mu+\mu^3) - \frac{1}{\mu^2} - \log\frac{1}{1-\mu^2} - \log\rndBrk{\frac{2}{\mu^2} + \frac{1}{1-\mu}}.
        \numberthis
    \end{align*}
    \textbf{Case 2:} Now, we can follow the same argument as the one in Case 2 of the proof of Theorem \ref{th:Hmin_eps_chain_rule} to derive the bound in the theorem statement for all states $\rho$. 
\end{proof}

\section{Conclusion}

We develop a powerful proof technique combining the entropic triangle inequality, the generalised GT inequality and the substate theorem for proving entropic bounds for approximation chains. We use this technique to prove novel chain rules-- the universal smooth min-entropy chain rule and the unstructured approximate entropy accumulation theorem. Importantly, both of these chain rules can be used meaningfully for arbitrarily large number of systems. \\

As far as applications are concerned, we use the unstructured approximate EAT to prove the security of parallel DI-QKD in our companion work \cite{Marwah24-DIQKD}. We expect the universal chain rule to aid in transforming von Neumann entropy based arguments to one-shot arguments. Furthermore, it provides a straightforward solution to problems such as the approximately independent registers problem discussed in \cite[Section 4]{Marwah23}. It is also our conviction that the proof technique introduced in this chapter will be useful for tackling other problems. For instance, it should be possible to use it to derive similar chain rules for one-shot variants (e.g., $I_{\max}$) of the (multipartite) mutual information. \\

As discussed in Sec. \ref{sec:decouple_possible}, it seems possible to decouple the smoothing parameter and the approximation parameter in certain scenarios, especially those involving DIQKD. We leave the problem of determining whether these parameters can be decoupled using testing for future work. This is an interesting and important question, which could potentially lead to significant improvements in the security proof of parallel DIQKD and the analysis of DIQKD with leakage.\\

Finally, we did not attempt to optimise our bounds here, but it should be interesting to study the absolute limits of the entropic error terms and the smoothing errors in our chain rules.   

\section*{Acknowledgments}

We would like to thank user:fedja on MathOverflow, who provided a proof for Lemma \ref{lemm:good_proj}-- a challenge that had eluded us for a long time-- and graciously permitted its reproduction. AM was supported by bourse d'excellence Google. This work was also supported by the Natural Sciences and Engineering Research Council of Canada.

% appendix
\appendix

\addcontentsline{toc}{section}{APPENDICES}
\section*{APPENDICES}

\section{Additional lemmas}

The following lemma provides a tighter bound for the distance as compared to Lemma \ref{lemm:cond_state_dist}. It is proven in \cite[Lemma B.3]{Dupuis14}. We could replace the use of Lemma \ref{lemm:cond_state_dist} in Theorem \ref{th:approx_EAT} and \ref{th:approx_EAT_wtest} with the following. This would slightly improve some constants. We use Lemma \ref{lemm:cond_state_dist} instead for notational clarity, since we do not need to keep track of the additional unitary introduced in the following lemma. 
\begin{lemma}
    For a normalised state $\rho_{AB}$ and a subnormalised state $\tilde{\rho}_{AB}$ such that $P(\rho_{AB}, \tilde{\rho}_{AB}) \leq \epsilon$, there exists a unitary $U_B$ on the register $B$ such that the state 
    \begin{align}
        \eta_{AB} := \rho_{B}^{1/2} U_B \tilde{\rho}_{B}^{-1/2}\tilde{\rho}_{AB}\tilde{\rho}_{B}^{-1/2}U_B^{\dagger}\rho_{B}^{1/2}
        \label{eq:eta_cond_state_dist_defn}
    \end{align}
    ($\tilde{\rho}_{B}^{-1/2}$ is the Moore-Penrose pseudo-inverse above) satisfies $P(\tilde{\rho}_{AB}, \eta_{AB}) \leq \epsilon$ and $P(\rho_{AB}, \eta_{AB}) \leq 2\epsilon$. Note that if $\tilde{\rho}_B$ is full rank, then $\eta_B = \rho_B$. 
    \label{lemm:cond_state_dist_extra}
\end{lemma}
  
\begin{proof}
Note that since $\rho_{AB}$ is normalised, we have $F(\tilde{\rho}_{AB}, \rho_{AB})\geq 1- \epsilon^2$. Let $\ket{\tilde{\rho}}_{ABR}$ be an arbitrary purification of $\tilde{\rho}_{AB}$. Let $U_B$ be any unitary for now and let $\eta_{AB}$ be defined as in Eq. \ref{eq:eta_cond_state_dist_defn} above. We will choose $U_B$ so that $F(\tilde{\rho}_{AB}, \eta_{AB})$ is large. \\

Observe that the pure state $\ket{\eta}_{ABR} := \rho_{B}^{1/2} U_B \tilde{\rho}_{B}^{-1/2} \ket{\tilde{\rho}}_{ABR}$ is a purification of $\eta_{AB}$. Using Uhlmann's theorem \cite[Theorem 3.22]{Watrous18}, we have 
\begin{align*}
    F(\tilde{\rho}_{AB}, \eta_{AB}) &\geq |\braket{\tilde{\rho} | \eta}|^2 \\
    &= \left\vert \braket{\tilde{\rho} | \rho_{B}^{1/2} U_B \tilde{\rho}_{B}^{-1/2} | \tilde{\rho}} \right\vert^2 \\
    &=\left\vert \tr(\rho_{B}^{1/2} U_B \tilde{\rho}_{B}^{-1/2} \tilde{\rho}_{ABR})\right\vert^2 \\
    &=\left\vert \tr(U_B \tilde{\rho}_{B}^{1/2} \rho_{B}^{1/2})\right\vert^2.
\end{align*}
Say the polar decomposition of $\tilde{\rho}_{B}^{1/2} \rho_{B}^{1/2} = V_B \left\vert \tilde{\rho}_{B}^{1/2} \rho_{B}^{1/2}\right\vert$. We can now select $U_B$ to be $V_B^{\dagger}$, so that 
\begin{align*}
    F(\tilde{\rho}_{AB}, \eta_{AB}) &\geq \rndBrk{\tr\left\vert \tilde{\rho}_{B}^{1/2} \rho_{B}^{1/2}\right\vert}^2\\
    &= F(\tilde{\rho}_B, \rho_{B})\\
    &\geq 1-\epsilon^2 \numberthis
\end{align*}
where we have used $F(\tilde{\rho}_B, \rho_{B}) \geq F(\tilde{\rho}_{AB}, \rho_{AB})$. Further, we have 
\begin{align*}
    P(\tilde{\rho}_{AB}, \eta_{AB}) &= \sqrt{1- F_\ast(\tilde{\rho}_{AB}, \eta_{AB})} \\
    &\leq \sqrt{1- F(\tilde{\rho}_{AB}, \eta_{AB})} \\
    &\leq \epsilon.
\end{align*}
Using the triangle inequality, for this choice of $U_B$, we get
\begin{align*}
    P({\rho}_{AB}, \eta_{AB}) &\leq P({\rho}_{AB}, \tilde{\rho}_{AB}) + P(\tilde{\rho}_{AB}, \eta_{AB})\\
    &\leq 2\epsilon.
\end{align*}
\end{proof}

\section{Comparison with previous chain rules}
\label{sec:ch_rule_comp}
It is instructive to understand why the previously developed chain rule fails to yield a universal chain rule. \cite[Lemma A.8]{Dupuis14} proved that
\begin{align}
  H_{\min}^{\epsilon_1+ 2\epsilon_2+\delta}(A_1 A_2|B)_{\rho} \geq H_{\min}^{\epsilon_1}(A_1|B)_{\rho} + H_{\min}^{\epsilon_2}(A_2| A_1 B)_{\rho} - k(\delta)
\end{align}
where $k(\delta)=O\rndBrk{\log\frac{1}{\delta}}$. Firstly, each time this chain rule is used, a large entropy loss of $k(\delta) = O\rndBrk{\log\frac{1}{\delta}}$ is incurred. If this was applied $n$ times, this would result in a loss of $n k(\delta)$ entropy. Secondly, if one repeatedly applies the chain rule above, then the smoothing parameter on the left-hand side would grow linearly with $n$. This second problem is fundamental to the purified distance based technique used to develop the previous chain rules. The triangle inequality based approach developed in this paper helps us overcome this. On the other hand, one can modify the technique in \cite[Lemma A.8]{Dupuis14} to absorb the $k(\delta)$ term into the smooth min-entropy, so that one of the terms is $H^{\downarrow, \epsilon}_{\min}$ in the lower bound. We do this in the following lemma. 

\begin{lemma}
  For $\epsilon_1, \epsilon_2 \in (0,1)$, and a state $\rho_{ABC}$, we have 
  \begin{align}
    H^{\epsilon_1 + 2\epsilon_2}_{\min}(AB | C) \geq H^{\epsilon_1}_{\min}(B | C) + H^{\downarrow, \epsilon_2}_{\min}(A | B C)
  \end{align}
\end{lemma}
\begin{proof}
  Let $\tilde{\rho}_{BC}$ be a subnormalised state such that $P(\rho_{BC}, \tilde{\rho}_{BC}) \leq \epsilon_1$ and $H^{\epsilon_1}_{\min}(B | C)_{\rho} = H_{\min}(B | C)_{\tilde{\rho}} =: \lambda_1$. Also, let $\bar{\rho}_{ABC}$ be a subnormalised state such that $P(\rho_{ABC}, \bar{\rho}_{ABC}) \leq \epsilon_2$ and $H^{\downarrow, \epsilon_2}_{\min}(A | BC)_{\rho} = H_{\min}(A | BC)_{\bar{\rho}} =: \lambda_2$. Let $\tilde{\rho}_{ABC}$ be an extension of $\tilde{\rho}_{BC}$ such that $P(\tilde{\rho}_{ABC}, {\rho}_{ABC}) = P(\tilde{\rho}_{BC}, {\rho}_{BC})$ (see \cite[Corollary 3.1]{TomamichelBook16}). Using Lemma \ref{lemm:cond_state_dist_extra}, there exists a unitary $U_{BC}$ such that 
  \begin{align}
    \eta_{ABC} := \tilde{\rho}_{BC}^{1/2} U_{BC} \bar{\rho}_{BC}^{-1/2} \bar{\rho}_{ABC} \bar{\rho}_{BC}^{-1/2}U_{BC}^\dagger \tilde{\rho}_{BC}^{1/2}
  \end{align}
  satisfies $P(\bar{\rho}_{ABC}, \eta_{ABC}) \leq \epsilon_1 + \epsilon_2$ since $P(\tilde{\rho}_{ABC}, \bar{\rho}_{ABC}) \leq \epsilon_1 + \epsilon_2$. Now, we have that 
  \begin{align}
    \bar{\rho}_{BC}^{-1/2} \bar{\rho}_{ABC} \bar{\rho}_{BC}^{-1/2} \leq e^{-\lambda_2} \Id_{ABC}
  \end{align}
  which implies that 
  \begin{align}
    \eta_{ABC} = \tilde{\rho}_{BC}^{1/2} U_{BC} \bar{\rho}_{BC}^{-1/2} \bar{\rho}_{ABC} \bar{\rho}_{BC}^{-1/2} U_{BC}^\dagger \tilde{\rho}_{BC}^{1/2} &\leq e^{-\lambda_2} \Id_{A} \otimes \tilde{\rho}_{BC} \\\
    &\leq e^{-(\lambda_1 + \lambda_2)} \Id_{AB} \otimes \sigma_C
  \end{align}
  for some state $\sigma_C$. Thus, we have that $H^{\epsilon_1 + 2 \epsilon_2}_{\min}(AB|C)_\rho \geq H_{\min}(AB|C)_\eta \geq \lambda_1 + \lambda_2$.
\end{proof}

\section{Classical approximate chain rule for the relative entropy}
\label{sec:cl_approx_ch_rule_for_D}

\begin{lemma}
    Let $p$ and $p'$ be probability distributions over $\mathcal{X}$. Then, for a function $f: \mathcal{X} \rightarrow \mathbb{R}$, we have 
    \begin{align}
        \left\vert \Expect_{X\sim p}[f(X)] - \Expect_{X\sim p'}[f(X)] \right\vert \leq \max_{x \in \mathcal{X}} |f(x)| \norm{p-p'}_1.
        \label{eq:cty_of_epectations}
    \end{align} 
    \label{lemm:cty_of_epectations}
\end{lemma}
\begin{proof}
    \begin{align*}
        \left\vert \Expect_{X\sim p}[f(X)] - \Expect_{X\sim p'}[f(X)] \right\vert &= \vert \sum_{x \in \mathcal{X}} p(x)f(x) - \sum_{x \in \mathcal{X}} p'(x)f(x)\vert \\
        &= \vert \sum_{x \in \mathcal{X}} (p(x)- p'(x))f(x) \vert \\
        &\leq \sum_{x \in \mathcal{X}} \vert p(x)- p'(x) \vert \vert f(x) \vert \\
        &\leq \max_{x \in \mathcal{X}} |f(x)| \norm{p-p'}_1
    \end{align*}
\end{proof}
We will call the following lemma an approximate chain rule for the relative entropy. It allows us to switch the probability distribution $p_{AB}$ in the term $D(p_{AB} || p_{B}q_{A|B})$ of the chain rule for $D(p_{AB}||q_{AB})$ to $D(p'_{AB} || p'_{B}q_{A|B})$, where $p \approx_{\epsilon} p'$ while incurring a penalty which depends only on the size of the register $A$. 
\begin{lemma}
    Suppose that $\delta >0$ and $q_{AB}$ is a probability distribution over $\mathcal{A}\times \mathcal{B}$ such that for all $a,b \in \mathcal{A}\times \mathcal{B}$, we have $q(a|b)\geq \delta$. Then, for $0<\epsilon<\frac{1}{2}$ and a $p'_{AB}$ such that $\frac{1}{2} \norm{p_{AB}- p'_{AB}}_1 \leq \epsilon$, we have that 
    \begin{align}
        D(p_{AB}||q_{AB}) \leq D(p_{B}||q_{B}) +  D(p'_{AB} ||p'_{B} q_{A|B})+ z(\epsilon, \delta)
        \label{eq:approx_chain_rule_for_D}
    \end{align}
    where $z(\epsilon, \delta) := h(2\epsilon)+ 6\epsilon\log\frac{1}{\delta} + 4\epsilon\log|\mathcal{A}|$. Equivalently, we have
    \begin{align}
        D(p_{AB}||q_{AB}) \leq D(p_{B}||q_{B}) +  \inf_{\norm{p'_{AB} - p_{AB}}_1 \leq 2\epsilon} D(p'_{AB} ||p'_{B} q_{A|B})+ z(\epsilon, \delta)
    \end{align}
    \label{lemm:approx_chain_rule_for_D}
\end{lemma}
\begin{proof}
    If $\supp(p_B) \not\subseteq \supp(q_B)$, then the right-hand side is infinite and the identity is trivially true. We suppose $\supp(p_B) \subseteq \supp(q_B)$ here on. \\

    We will show that for every $p'_{AB}$, which is $\epsilon$-close to $p_{AB}$ the right-hand side in Eq. \ref{eq:approx_chain_rule_for_D} is greater than the left-hand side. Classically, we have the chain rule
    \begin{align}
        D(p_{AB}||q_{AB}) = D(p_{B}||q_{B}) + \Expect_{B \sim p_B} [D(p_{A|B}||q_{A|B})].
        \label{eq:simple_ch_rule}
    \end{align}
    Note that both the above terms are finite ($q(a|b)\geq \delta$ is given). We will bound $\max_b D(p_{A|b}||q_{A|b})$ and then use Lemma \ref{lemm:cty_of_epectations} to create a bound for the expectation in terms of $p'$. For a given $b \in \mathcal{B}$, we have 
    \begin{align*}
        |D(p_{A|b}||q_{A|b})| &= \left\vert \sum_a p(a|b) \log\frac{p(a|b)}{q(a|b)} \right\vert \\
        &\leq \left\vert \sum_a p(a|b) \log p(a|b)\right\vert + \left\vert\sum_a p(a|b) \log\frac{1}{q(a|b)} \right\vert\\
        &= H(A|B=b)_{p} + \sum_a p(a|b) \log\frac{1}{q(a|b)} \\
        &\leq \log(|\mathcal{A}|) + \log\frac{1}{\delta}. 
    \end{align*} 
    Now, using Lemma \ref{lemm:cty_of_epectations}, 
    \begin{align}
        \Expect_{B \sim p_B} [D(p_{A|B}||q_{A|B})] \leq \Expect_{B \sim p'_B} [D(p_{A|B}||q_{A|B})] + 2\epsilon\rndBrk{\log(|\mathcal{A}|) + \log\frac{1}{\delta}}.
        \label{eq:D_ch_rule_intermediate_ch_rule}
    \end{align}
    Finally, we need to change the $p_{A|B}$ in $D(p_{A|B}||q_{A|B})$ to $p'_{A|B}$. 
    \begin{align*}
        D(p_{A|b}||q_{A|b}) &= -H(A|B=b)_{p} + \Expect_{A\sim p_{A|b}} \sqBrk{\log\frac{1}{q(A|b)}} \\
        &\leq -H(A|B=b)_{p'} + h\rndBrk{\frac{1}{2}\norm{p_{A|b}-p'_{A|b}}_1} + \frac{1}{2}\norm{p_{A|b}-p'_{A|b}}_1 \log(|\mathcal{A}|)\\
        & \quad\quad+ \Expect_{A\sim p'_{A|b}} \sqBrk{\log\frac{1}{q(A|b)}} + \norm{p_{A|b}-p'_{A|b}}_1 \log\frac{1}{\delta}\\
        &= D(p'_{A|b}||q_{A|b}) + h\rndBrk{\frac{1}{2} \norm{p_{A|b}-p'_{A|b}}_1} + \frac{1}{2} \norm{p_{A|b}-p'_{A|b}}_1 \rndBrk{2\log\frac{1}{\delta} + \log(|\mathcal{A}|)}
    \end{align*}
    where we used the Fannes-Audenaert continuity bound \cite[Theorem 11.10.2]{Wilde13} and Lemma \ref{lemm:cty_of_epectations} in the second line. Taking the expectation over $p'_B$, we get 
    \begin{align*}
        & \Expect_{B \sim p'_B} [D(p_{A|B}||q_{A|B})] \\
        & \leq \Expect_{B \sim p'_B} \sqBrk{D(p'_{A|B}||q_{A|B}) + h\rndBrk{\frac{1}{2}\norm{p_{A|B}-p'_{A|B}}_1} + \frac{1}{2}\norm{p_{A|B}-p'_{A|B}}_1 \rndBrk{2\log\frac{1}{\delta} + \log(|\mathcal{A}|)}} \\
        & \leq \Expect_{B \sim p'_B} \sqBrk{D(p'_{A|B}||q_{A|B})} + h\rndBrk{\frac{1}{2}\norm{p'_B p_{A|B} - p'_{AB}}_1}+ \frac{1}{2} \norm{p'_B p_{A|B} - p'_{AB}}_1 \rndBrk{2\log\frac{1}{\delta} + \log|\mathcal{A}|} \\
        & \leq \Expect_{B \sim p'_B} \sqBrk{D(p'_{A|B}||q_{A|B})} + h(2\epsilon)+ 2\epsilon \rndBrk{2\log\frac{1}{\delta} + \log|\mathcal{A}|} \\
        &= D(p'_{AB}|| p'_B q_{A|B}) + h(2\epsilon)+ 2\epsilon \rndBrk{2\log\frac{1}{\delta} + \log|\mathcal{A}|}
    \end{align*}
    where we have used the fact that if $\norm{p_{AB} - p'_{AB}}_1 \leq \epsilon$, then $\norm{p'_{B}p_{A|B} - p'_{AB}}_1 \leq 2\epsilon$. This can be derived using the triangle inequality. Putting this in Eq. \ref{eq:D_ch_rule_intermediate_ch_rule}, we get 
    \begin{align*}
        \Expect_{B \sim p_B} [D(p_{A|B}||q_{A|B})] \leq D(p'_{AB}|| p'_B q_{A|B}) + h(2\epsilon)+ 6\epsilon\log\frac{1}{\delta} + 4\epsilon\log|\mathcal{A}|.
    \end{align*}
    Therefore, using Eq. \ref{eq:simple_ch_rule}, we get 
    \begin{align*}
        D(p_{AB}||q_{AB}) \leq D(p_B || q_B) + D(p'_{AB}|| p'_B q_{A|B}) + h(2\epsilon)+ 6\epsilon\log\frac{1}{\delta} + 4\epsilon\log|\mathcal{A}|.
    \end{align*}
\end{proof}

\section{Markov chain condition for all input states implies independence}
\label{sec:all_op_MC_cond}

In the following lemma, we show that if all outputs of a channel satisfy a certain Markov chain condition with the reference registers, then under some dimension constraints the output of the channel is independent of the input. Due to this fact, we choose to state the unstructured approximate EAT using the independence condition for the side information (Eq. \ref{eq:side_info_ind}). We expect that this lemma can be improved further.

\begin{lemma}
  Let $A$, $B$ and $R$ be registers such that $|A| = |B|$ and $|R| = |A||B|$. Let $\cM: R \rightarrow C$ be a channel such that for all input states $\rho^{(0)}_{A B R}$ the output $\rho_{ABC} = \cM(\rho^{(0)})$ satisfies the Markov chain $A \leftrightarrow B \leftrightarrow C$. Then, we have that $\cM(X_R) = \tr(X) \omega_C$ for some state $\omega_C$.
\end{lemma}
\begin{proof}
  Since $|R| = |A||B|$, we can view $R$ as the registers $A' B'$, where $A \equiv A$ and $B' \equiv B$. We can construct the Choi matrix of this channel as 
  \begin{align}
    J_{ABC} = \cM\rndBrk{\ket{\Phi}\bra{\Phi}_{A B A' B'}}
  \end{align}
  where $\ket{\Phi}$ is used to denote the unnormalised maximally entangled state, i.e., $\ket{\Phi}_{A B A' B'} := \sum_{a,b} \ket{ab}_{A B}\ket{ab}_{A' B'} = \ket{\Phi}_{A A'} \otimes \ket{\Phi}_{B B'}$ and $\cM$ is viewed as a channel from $A' B' \rightarrow A B$. Since, all outputs of $\cM$ satisfy the Markov chain $A \leftrightarrow B \leftrightarrow C$, we have that
  \begin{align*}
    J_{ABC} &= J_{AB} J_B^{-1} J_{BC} \\
    &= \Id_A \otimes \Id_B\ \cdot\ |A|^{-1} \Id_B\ \cdot\ \cM(\Id_{A'} \otimes \ket{\Phi}\bra{\Phi}_{B B'}) \\
    &= \Id_A \otimes \cM(\tau_{A'} \otimes \ket{\Phi}\bra{\Phi}_{B B'})
    \numberthis
  \end{align*}
  where $\tau_{A'} = |A|^{-1} \Id_{A'}$ is the maximally mixed state on $A$. Let's define $\cN:{B'\rightarrow C}$ as 
  \begin{align}
    \cN(X_{B'}) = \cM\rndBrk{\tau_{A'} \otimes X_{B'}}.
  \end{align}
  Since, the Choi matrix is unique, we can see that 
  \begin{align}
    \cM_{A' B' \rightarrow C} = \cN_{B' \rightarrow C} \circ \tr_{A'}.
    \label{eq:M_tr_A}
  \end{align}
  Let $W_{A'B'}$ be the swap unitary matrix, i.e., $W_{A'B'}\ket{a b} = \ket{b a}$. Then, note that for all input states $\rho_{A B A' B'}$, the output $\cM(W_{A'B'} \rho_{A B A' B'} W_{A'B'}^{\dagger})$ satisfies the Markov chain $A \leftrightarrow B \leftrightarrow C$ according to the hypthosis in the lemma statement. In particular, we can carry out the above argument using the channel $\cM(W_{A'B'}\ \cdot\ W_{A'B'}^{\dagger})$ and that gives us that for all operators $X_{A'B'}$
  \begin{align}
    \cM(W_{A'B'} X_{A'B'} W_{A'B'}^{\dagger}) &=  \cN'_{B' \rightarrow C} \rndBrk{\tr_{A'}(X_{A'B'})}
  \end{align}
  which implies that 
  \begin{align}
    \cM( X_{A'B'}) &=  \cN'_{B' \rightarrow C} \rndBrk{\tr_{A'}(W_{A'B'} X_{A'B'} W_{A'B'}^{\dagger})} \nonumber \\
    &= \cN'_{A' \rightarrow C} \rndBrk{\tr_{B'}(X_{A'B'})}
    \label{eq:M_tr_B}
  \end{align}
  Using Eq. \ref{eq:M_tr_A} and \ref{eq:M_tr_B} for $X_{A'B'} = \sigma_{A'} \otimes \sigma_{B'}$ where $\sigma_{A'}$ and $\sigma_{B'}$ are arbitrary states, we have that 
  \begin{align}
    \cN_{B' \rightarrow C} \rndBrk{\sigma_{B'}} = \cN'_{A' \rightarrow C} \rndBrk{\sigma_{A'}}
  \end{align}
  which implies that both of these must be equal to a constant state. Let's call the state $\omega_C$. Using the fact that $\cM$ is trace-preserving, we then have that $\cM_{A'B' \rightarrow C}(X_{A'B'}) = \tr(X)\omega_C$. 
\end{proof}

\section{Testing for unstructured approximate entropy accumulation}
\label{sec:testing}

To incorporate testing to Theorem \ref{th:approx_EAT}, we follow \cite{Metger22}, which is itself based on \cite{Dupuis19}. \\

First, we will define the testing channels $\mathcal{T}_k$. These channels measure the outputs $A_k$ and $B_k$ of the state $\rho$ and output a result $X_k$ based on these measurements. Concretely, for every $k \in [n]$ the channel $\mathcal{T}_k : A_k B_k \rightarrow A_k B_k X_k$ is of the form
\begin{align}
    \mathcal{T}_k (\omega_{A_k B_k}) = \sum_{a, b} \Pi_{A_k}^{(a)} \otimes \Pi_{B_k}^{(b)} \omega_{A_k B_k} \Pi_{A_k}^{(a)} \otimes \Pi_{B_k}^{(b)} \otimes \ket{x(a,b)}\bra{x(a,b)}_{X_k}
    \label{eq:test_maps}
\end{align}
where $\{\Pi_{A_k}^{(a)}\}_a$ and $\{\Pi_{B_k}^{(b)}\}_b$ are orthogonal projectors and $x(\cdot)$ is some deterministic function which uses the measurements $a$ and $b$ to create the output register $X_k$.\\

Using these orthogonal projectors, we further define the orthogonal projectors on the registers $A_1^k B_1^k$
\begin{align}
  \Pi_{A_1^k B_1^k}^{(a_1^k, b_1^k)} := \bigotimes_{i=1}^k \rndBrk{\Pi_{A_i}^{(a_i)} \otimes \Pi_{B_i}^{(b_i)}}
\end{align}
for every $a_1^k, b_1^k$. Together these form a measurement on the registers $A_1^k B_1^k$.\\

Next we define the min-tradeoff functions. Let $\mathbb{P}$ be the set of probability distributions over the alphabet of $X$ registers. Let $R$ be any register isomorphic to $R_{k}$. For a probability $q \in \mathbb{P}$ and a channel $\cN_{k} : R_{k} \rightarrow A_k B_k$, we define the set 
\begin{align}
    \Sigma_k (q | \cN_{k}) := \curlyBrk{\nu_{A_k B_k X_k R} = \mathcal{T}_k \circ \cN_{k}(\omega_{R_{k}R}): \text{ for a state } \omega_{R_{k-1}R} \text{ such that }\nu_{X_k}= q}.
\end{align}

\begin{definition}
    A function $f: \mathbb{P} \rightarrow \mathbb{R}$ is called a min-tradeoff function for the channels $\{\cN_k \}_{k=1}^n$ if for every $k$, it satisfies
    \begin{align}
        f(q) \leq \inf_{\nu \in \Sigma_k (q| \cN_{k})} H(A_k | B_k R)_{\nu}.
    \end{align}
\end{definition}
We will also need the definitions of the following simple properties of the min-tradeoff functions for our entropy accumulation theorem:
\begin{align}
    &\text{max}(f) := \max_{q \in \mathbb{P}} f(q)\\
    &\text{min}(f) := \min_{q \in \mathbb{P}} f(q).
\end{align}
We now state the unstructured approximate EAT with testing. 
\begin{theorem}
  \label{th:approx_EAT_wtest}
    Let $\epsilon \in (0,1)$ and for every $k \in [n]$, the registers $A_k$ and $B_k$ be such that $|A_k| = |A|$ and $|B_k| = |B|$. Suppose, the state $\rho_{A_1^n B_1^n X_1^n E}$ is such that 
    \begin{enumerate}
      \item The registers $X_1^n$ can be recreated by applying the testing maps to the registers $A_1^n$ and $B_1^n$, that is, 
      \begin{align}
        \rho_{A_1^n B_1^n X_1^n E} = \mathcal{T}_n \circ \cdots \circ \mathcal{T}_1 (\rho_{A_1^n B_1^n E})
      \end{align}
      \item For every $k \in [n]$, there exists a channel $\cM_k: R_{k} \rightarrow A_k B_k$ and a state $\theta_{B_k}^{(k)}$ such that 
      \begin{align}
        &\tr_{X_k}\circ \mathcal{T}_k \circ \cM_k = \cM_k \\
        &\tr_{A_k}\circ \cM_k (X_{R_k}) = \tr(X) \theta_{B_k}^{(k)} \quad \text{ for all operators } X_{R_k}
      \end{align}
      and a state $\tilde{\rho}^{(k, 0)}_{A_1^{k-1} B_1^{k-1} E R_{k}}$ for which 
      \begin{align}
        \frac{1}{2}\norm{\rho_{A_1^k B_1^k E} - \cM_k(\tilde{\rho}^{(k, 0)}_{A_1^{k-1} B_1^{k-1} E R_{k}})}_1\leq \epsilon.
      \end{align} 
    \end{enumerate}
    Then, for an event $\Omega$ defined using $X_1^n$, an affine min-tradeoff function $f$ for $\{\cM_k\}_{k=1}^n$ such that for every $x_1^n \in \Omega$, $f(\text{freq}(x_1^n)) \geq h$, we have
    \begin{align}
      H_{\min}^{\mu' + \epsilon'}(A_1^n|B_1^n E)_{\rho_{|\Omega}} &\geq n(h  - V(3\sqrt{\mu} + 4\epsilon) -g_2 (2\epsilon)) \nonumber \\
      & \quad - \frac{V}{\sqrt{\mu}}\rndBrk{2\log\frac{1}{P_{\rho}(\Omega) - \mu} + \frac{2}{\mu^2} + 2 \log \frac{1}{1- \mu^2} + g_1(\epsilon', \mu')}
    \end{align}
    where 
    \begin{align}
      &\mu := \rndBrk{\frac{8 \sqrt{\epsilon} + 2\epsilon}{1- \epsilon^2/(|A| |B|)^2} \log\frac{|A| |B|}{\epsilon}}^{1/3}\\
      &\mu' := 2\sqrt{\frac{\mu}{P_{\rho}(\Omega)}} \\
      &V := \log \rndBrk{1 + 2|A|} + \lceil \max(f) - \min(f)\rceil\\
      &g_1(x, y):= - \log(1- \sqrt{1-x^2}) - \log (1-y^2) \\
      &g_2(x) :=  x\log \frac{1}{x} + (1+x)\log (1+x)
    \end{align}
    and $\epsilon' \in (0,1)$ such that $\mu' + \epsilon' <1$.
\end{theorem}
\begin{proof}
  We first define the sets 
  \begin{align}
    \mathcal{A}_k := \curlyBrk{X_{A_1^k B_1^k E} : X_{A_1^k B_1^k E} = \sum_{a_1^k, b_1^k} \Pi^{(a_1^k b_1^k)}_{A_1^k B_1^k} X_{A_1^k B_1^k E} \Pi^{(a_1^k b_1^k)}_{A_1^k B_1^k}}
  \end{align}
  for every $0 \leq k \leq n$. It should be noted that each of these sets is a \emph{unital algebra}, i.e., a vector space over the field $\mathbb{C}$ closed under the matrix product containing identity $\Id_{A_1^k B_1^k E}$. It is also easy to see that 
  \begin{align}
    &\mathcal{A}_k \otimes \Id_{A_{k+1}B_{k+1}} := \curlyBrk{X_{A_1^k B_1^k E} \otimes \Id_{A_{k+1}B_{k+1}} : X_{A_1^k B_1^k E} \in \mathcal{A}_k} \subseteq \mathcal{A}_{k+1} \label{eq:algebra_tensor_Id}\\
    &\tr_{A_k B_k}(\mathcal{A}_k) := \curlyBrk{X_{A_1^{k-1} B_1^{k-1} E}: X_{A_1^k B_1^k E} \in \mathcal{A}_k} \subseteq \mathcal{A}_{k-1}. \label{eq:algebra_part_trace}
  \end{align}
  Finally, it should be noted that for all operators $X_{A_1^k B_1^k E} \in \mathcal{A}_k$, we have 
  \begin{align}
    \tr_{X_1^k} \circ \mathcal{T}_k \circ \cdots \circ \mathcal{T}_1 \rndBrk{X} = X.
  \end{align}
  The above equation in particular implies that one can measure a state in $\mathcal{A}_k$ using the channel $\mathcal{T}_k \circ \cdots \circ \mathcal{T}_1$ without disturbing the state. All the auxiliary states defined during this proof will be contained inside an appropriate set $\mathcal{A}_k$. Noting that these sets are, in fact, an algebra will make it easier to see this. For example, observe that each partial state $\rho_{A_1^k B_1^k E} \in \mathcal{A}_k$. \\

  \textbf{Case 1:} To begin, we restrict our attention to states $\rho_{A_1^n B_1^n E}$, which have full rank. Let $\nu \in (0,1)$ be an arbitrarily chosen small parameter. For every $k \in [n]$, define the states
  \begin{align}
    \tilde{\tilde{\rho}}^{(k, 0)}_{A_1^{k-1} B_1^{k-1} E R_{k}} &:= (1-\nu) \tilde{\rho}^{(k, 0)}_{A_1^{k-1} B_1^{k-1} E R_{k}} + \nu \tau_{A_1^{k-1} B_1^{k-1} E R_{k}} \\
    \tilde{\tilde{\rho}}^{(k, 1)}_{A_1^{k-1} B_1^{k-1} X_1^{k-1} E R_{k}} &:= \mathcal{T}_{k-1} \circ \cdots \circ \mathcal{T}_{1} (\tilde{\tilde{\rho}}^{(k, 0)}_{A_1^{k-1} B_1^{k-1} E R_{k}}) \\
    \tilde{\tilde{\rho}}^{(k)}_{A_1^{k} B_1^{k} X_1^k E} &:= \mathcal{T}_{k} \circ \cM_k\rndBrk{\tilde{\tilde{\rho}}^{(k, 1)}_{A_1^{k-1} B_1^{k-1} X_1^{k-1} E R_{k}}}.
  \end{align}
  Since the maps $\tr_{X_i}\circ \mathcal{T}_{i}$ are unital, for every $k$, we have 
  \begin{align}
    \tilde{\tilde{\rho}}^{(k)}_{A_1^{k-1} B_1^{k-1} E} &= \tilde{\tilde{\rho}}^{(k, 1)}_{A_1^{k-1} B_1^{k-1} E} \label{eq:tilde_tilde_rho_partial_st_eq}\\
    &\geq \nu \tau_{A_1^{k-1} B_1^{k-1} E}.
  \end{align} 
  Thus, the states $\tilde{\tilde{\rho}}^{(k)}_{A_1^{k-1} B_1^{k-1} E}$ are also full rank. Moreover, we have that 
  \begin{align}
    \tilde{\tilde{\rho}}^{(k)}_{A_1^{k} B_1^{k} E} \in \mathcal{A}_k. 
    \label{eq:tilde_tilde_rho_in_algA}
  \end{align}
  For each $k$ these states satisfy
  \begin{align*}
    \frac{1}{2}\norm{\rho_{A_1^{k} B_1^{k} X_1^k E} - \tilde{\tilde{\rho}}^{(k)}_{A_1^{k} B_1^{k} X_1^k E}}
    &= \frac{1}{2}\norm{\rho_{A_1^k B_1^k X_1^k E} - \mathcal{T}_k \circ \cM_k\rndBrk{\tilde{\tilde{\rho}}^{(k,1)}_{A_1^{k-1} B_1^{k-1} X_1^{k-1} E R_{k}}}}_1 \\
    &= \frac{1- \nu}{2}\norm{\mathcal{T}_k \circ \cdots \circ \mathcal{T}_1 (\rho_{A_1^k B_1^k E}) - \mathcal{T}_k \circ \cdots \circ \mathcal{T}_1\circ  \cM_k\rndBrk{\tilde{\rho}^{(k,0)}_{A_1^{k-1} B_1^{k-1} E R_{k}}}}_1 + \nu \\
    &\leq \frac{1- \nu}{2}\norm{\rho_{A_1^k B_1^k E} - \cM_k\rndBrk{\tilde{\rho}^{(k,0)}_{A_1^{k-1} B_1^{k-1} E R_{k}}}}_1 + \nu \\
    &\leq \epsilon + \nu. \numberthis
  \end{align*}
  Now, for each $k \in [n]$, we define the normalised states
  \begin{align}
    &\omega^{(k,1)}_{A_1^{k-1} B_1^{k-1} R_k E} := \rho_{A_1^{k-1} B_1^{k-1} E}^{1/2} \rndBrk{\tilde{\tilde{\rho}}^{(k)}_{A_1^{k-1} B_1^{k-1} E}}^{-1/2} \tilde{\tilde{\rho}}^{(k,1)}_{A_1^{k-1} B_1^{k-1} R_k E} \rndBrk{\tilde{\tilde{\rho}}^{(k)}_{A_1^{k-1} B_1^{k-1} E}}^{-1/2} \rho_{A_1^{k-1} B_1^{k-1} E}^{1/2}
  \end{align}
  and
  \begin{align}
    \omega^{(k)}_{A_1^{k} B_1^{k} E} &:= \cM_k\rndBrk{\omega^{(k,1)}_{A_1^{k-1} B_1^{k-1} R_k E}} \\
    &= \rho_{A_1^{k-1} B_1^{k-1} E}^{1/2} \rndBrk{\tilde{\tilde{\rho}}^{(k)}_{A_1^{k-1} B_1^{k-1} E}}^{-1/2} \tr_{X_k}\circ \mathcal{T}_k \circ \cM_k \rndBrk{\tilde{\tilde{\rho}}^{(k,1)}_{A_1^{k-1} B_1^{k-1} R_k E}} \rndBrk{\tilde{\tilde{\rho}}^{(k)}_{A_1^{k-1} B_1^{k-1} E}}^{-1/2} \rho_{A_1^{k-1} B_1^{k-1} E}^{1/2} \\
    &= \rho_{A_1^{k-1} B_1^{k-1} E}^{1/2} \rndBrk{\tilde{\tilde{\rho}}^{(k)}_{A_1^{k-1} B_1^{k-1} E}}^{-1/2} \tilde{\tilde{\rho}}^{(k)}_{A_1^{k} B_1^{k} E} \rndBrk{\tilde{\tilde{\rho}}^{(k)}_{A_1^{k-1} B_1^{k-1} E}}^{-1/2} \rho_{A_1^{k-1} B_1^{k-1} E}^{1/2}.
  \end{align}
  We used $\tr_{X_k}\circ \mathcal{T}_k \circ \cM_k = \cM_k$ above. Observe that $\omega^{(k)}_{A_1^{k} B_1^{k} E} \in \mathcal{A}_k$ (using Eq. \ref{eq:tilde_tilde_rho_in_algA}, \ref{eq:algebra_tensor_Id} and \ref{eq:algebra_part_trace}). Since, we defined $\tilde{\tilde{\rho}}^{(k)}_{A_1^{k-1} B_1^{k-1} E}$ to be full rank, we have that 
  \begin{align}
    \omega^{(k)}_{A_1^{k-1} B_1^{k-1}  E} = \rho_{A_1^{k-1} B_1^{k-1}  E}.
  \end{align} 
  Using Lemma \ref{lemm:cond_state_dist}, we have that 
  \begin{align*}
    \frac{1}{2}\norm{\rho_{A_1^{k} B_1^{k} E} - \omega^{(k)}_{A_1^{k} B_1^{k} E}}_1 &\leq (\sqrt{2} +1) P(\rho_{A_1^{k} B_1^{k} E}, \tilde{\tilde{\rho}}_{A_1^{k} B_1^{k} E}) \\
    &\leq (\sqrt{2} +1)\sqrt{2(\epsilon + \nu)}\\
    &\leq 4\sqrt{\epsilon + \nu}. \numberthis 
  \end{align*}
  Let $\delta \in (0,1)$ be a small parameter (to be set equal to $\epsilon$ later). Define the states 
  \begin{align}
    \rho^{(k, \delta)}_{A_k B_k} := (1-\delta)\rho_{A_k B_k} + \delta \tau_{A_k B_k}.
  \end{align}
  Finally, for every $k \in [n]$, we define the states 
  \begin{align}
    \bar{\rho}^{(k)}_{A_1^{k} B_1^{k} E} :=& (1-\delta) \omega^{(k)}_{A_1^{k} B_1^{k} E} + \delta \rho^{(k, \delta)}_{A_k B_k} \otimes \rho_{A_1^{k-1} B_1^{k-1} E} \\
    =& (1-\delta) \omega^{(k)}_{A_1^{k} B_1^{k} E} + \delta \rho^{(k, \delta)}_{A_k B_k} \otimes \omega^{(k)}_{A_1^{k-1} B_1^{k-1} E}
  \end{align}
  Also, define $\bar{\rho}^{(0)}_E: =\rho_E$. For each $0 \leq k \leq n$, the state $\bar{\rho}^{(k)}_{A_1^{k} B_1^{k} E} \in \mathcal{A}_k$ using Eq. \ref{eq:algebra_tensor_Id}. We have taken a slight diversion from the proof of Theorem \ref{th:approx_EAT} in defining the above state. This has been done to ensure we are able to bound the entropy in Eq. \ref{eq:per_rnd_Halpha_bd2}.\\

  Let $\Delta_k : {R_k \rightarrow A_k B_k}$ be the map which traces out the register $R_k$ and outputs $\rho^{(k, \delta)}_{A_k B_k}$. Further, let 
  \begin{align}
    \cM^{\delta}_k := (1-\delta) \cM_k + \delta \Delta_k.
  \end{align}
  Note that similar to $\cM_k$, $\cM^{\delta}_k$ also satisfies 
  \begin{align}
    \cM^{\delta}_k = \tr_{X_k} \circ \mathcal{T}_k \circ \cM^{\delta}_k.
  \end{align}
  Then, we have that 
  \begin{align*}
    \bar{\rho}^{(k)}_{A_1^{k} B_1^{k} E} &= (1-\delta) \omega^{(k)}_{A_1^{k} B_1^{k} E} + \delta \rho^{(k, \delta)}_{A_k B_k} \otimes \omega^{(k)}_{A_1^{k-1} B_1^{k-1} E} \\
    &= \rndBrk{(1-\delta) \cM_k + \delta \Delta_k}\rndBrk{\omega^{(k,1)}_{A_1^{k-1} B_1^{k-1} R_k E}} \\
    &=  \cM^\delta_k \rndBrk{\omega^{(k,1)}_{A_1^{k-1} B_1^{k-1} R_k E}} . \numberthis
  \end{align*}
  We also have that 
  \begin{align}
    \frac{1}{2}\norm{\bar{\rho}^{(k)}_{A_1^{k} B_1^{k} E} - \rho_{A_1^{k} B_1^{k} E}}_1 \leq 4 \sqrt{\epsilon + \nu} + \delta
  \end{align}
  and by the definition of $\bar{\rho}^{(k)}_{A_1^{k} B_1^{k} E}$
  \begin{align}
    \delta^2 \tau_{A_k B_k} \otimes \rho_{A_1^{k-1} B_1^{k-1} E} \leq \bar{\rho}^{(k)}_{A_1^{k} B_1^{k} E}
  \end{align}
  Using Corollary \ref{cor:TRE_bd}, gives us that 
  \begin{align}
    D(\rho_{A_1^k B_1^k E} || \bar{\rho}^{(k)}_{A_1^k B_1^k E})
    &\leq \frac{8 \sqrt{\epsilon + \nu} + 2\delta}{1- \delta^2/(|A| |B|)^2} \log\frac{|A||B|}{\delta}.
  \end{align}
  We define the above bound as $z(\epsilon+ \nu, \delta)$. Using Lemma \ref{lemm:cond_st_rel_ent_bd}, for the normalised auxiliary state
  \begin{align}
    \sigma_{A_1^n B_1^n E} &:= \int_{-\infty}^{\infty} dt \beta_0 (t) \prod_{k=0}^{n-1} \sqBrk{\rndBrk{\bar{\rho}^{(k)}_{A_1^k B_1^k E}}^{\frac{1-it}{2}} \rndBrk{\bar{\rho}^{(k+1)}_{A_1^{k} B_1^k E}}^{-\frac{1-it}{2}}} \cdot \bar{\rho}^{(n)}_{A_1^n B_1^n E} \cdot \prod_{k= n-1}^{0} \sqBrk{\rndBrk{\bar{\rho}^{(k+1)}_{A_1^{k} B_1^kE}}^{-\frac{1+it}{2}} \rndBrk{\bar{\rho}^{(k)}_{A_1^k B_1^k E}}^{\frac{1+it}{2}}} \label{eq:approx_EAT_wtest_sigma_def1}
  \end{align}
  we have that 
  \begin{align}
    D_m (\rho_{A_1^n B_1^n E} || \sigma_{A_1^n B_1^n E}) \leq n z(\epsilon + \nu, \delta)
  \end{align}
  and
  \begin{align}
    \sigma&_{A_1^k B_1^k E} \nonumber \\
    &= \int_{-\infty}^{\infty} dt \beta_0 (t) \prod_{j=0}^{k-1} \sqBrk{\rndBrk{\bar{\rho}^{(j)}_{A_1^j B_1^j E}}^{\frac{1-it}{2}} \rndBrk{\bar{\rho}^{(j+1)}_{A_1^{j} B_1^j E}}^{-\frac{1-it}{2}}} \cdot \bar{\rho}^{(k)}_{A_1^k B_1^k E} \cdot \prod_{j= k-1}^{0} \sqBrk{\rndBrk{\bar{\rho}^{(j+1)}_{A_1^{j} B_1^jE}}^{-\frac{1+it}{2}} \rndBrk{\bar{\rho}^{(j)}_{A_1^j B_1^j E}}^{\frac{1+it}{2}}} \label{eq:approx_EAT_wtest_sigma_partial_st1} \\
    &= \cM^\delta_k \bigg(\int_{-\infty}^{\infty} dt \beta_0 (t) \prod_{j=0}^{k-1} \sqBrk{\rndBrk{\bar{\rho}^{(j)}_{A_1^j B_1^j E}}^{\frac{1-it}{2}} \rndBrk{\bar{\rho}^{(j+1)}_{A_1^{j} B_1^j E}}^{-\frac{1-it}{2}}} \cdot \omega^{(k,1)}_{A_1^{k-1} B_1^{k-1} R_k E} \cdot \prod_{j= k-1}^{0} \sqBrk{\rndBrk{\bar{\rho}^{(j+1)}_{A_1^{j} B_1^jE}}^{-\frac{1+it}{2}} \rndBrk{\bar{\rho}^{(j)}_{A_1^j B_1^j E}}^{\frac{1+it}{2}}} \bigg)
    \label{eq:approxEAT_test_sigma_map_op}
  \end{align}
  for all $k \in [n]$. Let $\sigma^{(k,0)}_{A_1^{k-1} B_1^{k-1} R_k E}$ be the input state for $\cM^\delta_k$ above, so that
  \begin{align}
    \sigma_{A_1^k B_1^k E} &= \cM^\delta_k\rndBrk{\sigma^{(k,0)}_{A_1^{k-1} B_1^{k-1} R_k E}}
  \end{align}
  It is also easy to see using the properties of the algebras $\mathcal{A}_k$ that $\sigma_{A_1^k B_1^k E} \in \mathcal{A}_k$ for each $k$. Thus, we can extend the state $\sigma$ as follows using measurements $(\mathcal{T}_i)_i$
  \begin{align*}
    \sigma_{A_1^n B_1^n X_1^n E} &= \mathcal{T}_n \circ \cdots \circ \mathcal{T}_1\rndBrk{\sigma_{A_1^n B_1^n E}}.
  \end{align*}
  Note that the partial state
  \begin{align*}
    \sigma_{A_1^k B_1^k X_1^k E} &= \tr_{A_{k+1}^n B_{k+1}^n X_{k+1}^n} \rndBrk{\mathcal{T}_n \circ \cdots \circ \mathcal{T}_1\rndBrk{\sigma_{A_1^n B_1^n  E}}}\\
    &= \mathcal{T}_{k} \circ \cdots \circ \mathcal{T}_1\rndBrk{\tr_{A_{k+1}^n B_{k+1}^n X_{k+1}^n}\rndBrk{\mathcal{T}_{n} \circ \cdots \circ \mathcal{T}_{k+1}\rndBrk{\sigma_{A_1^n B_1^n E}}}} \\
    &= \mathcal{T}_{k} \circ \cdots \circ \mathcal{T}_1\rndBrk{\sigma_{A_1^k B_1^k E}} \numberthis 
    \label{eq:X_for_sigma_const}\\
    &= \mathcal{T}_{k} \circ \cM^\delta_k \circ \mathcal{T}_{k-1} \circ \cdots \circ \mathcal{T}_1\rndBrk{\sigma^{(k,0)}_{A_1^{k-1} B_1^{k-1} R_k E}} 
    \numberthis
    \label{eq:sigma_in_corr_form}
  \end{align*} 
  Let's define $\mu := z(\epsilon + \nu, \delta)^{1/3}$. Using the substate theorem (Theorem \ref{thm:substate_th}), we get the following bound from the above relative entropy bound
  \begin{align}
    D^{\mu}_{\max}(\rho_{A_1^n B_1^n  E} || \sigma_{A_1^n B_1^n  E}) \leq n \mu + \frac{1}{\mu^2} + \log\frac{1}{1- \mu^2}
  \end{align}
  which using data processing also implies that 
  \begin{align}
    D^{\mu}_{\max}(\rho_{A_1^n B_1^n  X_1^n E} || \sigma_{A_1^n B_1^n X_1^n E}) \leq n \mu + \frac{1}{\mu^2} + \log\frac{1}{1- \mu^2}
  \end{align}
  The bound above implies that there exists a state ${\rho}'_{A_1^n B_1^n X_1^n E}$, which is also classical on $X_1^n$ such that 
  \begin{align}
      P\rndBrk{\rho_{A_1^n B_1^n  X_1^n E}, {\rho}'_{A_1^n B_1^n  X_1^n E}} \leq \mu
  \end{align}
  and 
  \begin{align}
      {\rho}'_{A_1^n B_1^n  X_1^n E} \leq \frac{e^{n \mu + \frac{1}{\mu^2}}}{1- \mu^2}\sigma_{A_1^n B_1^n  X_1^n E}. 
      \label{eq:EAT_prime_rho_op_ineq}
  \end{align}
  The registers $X_1^n$ for ${\rho}'$ can be chosen to be classical, since the channel measuring $X_1^n$ only decreases the distance between ${\rho}'$ and $\rho$, and the new state produced would also satisfy Eq. \ref{eq:EAT_prime_rho_op_ineq}. As the registers $X_1^n$ are classical for both $\sigma$ and ${\rho}'$, we can condition these states on the event $\Omega$. We will call the probability of the event $\Omega$ for the state $\sigma$ and ${\rho}'$ $P_{\sigma}(\Omega)$ and $P_{{\rho}'}(\Omega)$ respectively. Using \cite[Lemma G.1]{Marwah23} and the Fuchs-van de Graaf inequality, we have 
  \begin{align}
      P\rndBrk{\rho_{A_1^n B_1^n X_1^n E|\Omega} , {\rho}'_{A_1^n B_1^n X_1^n E|\Omega}} \leq 2 \sqrt{\frac{\mu}{P_{{\rho}}(\Omega)}}.
  \end{align}
  Conditioning Eq. \ref{eq:EAT_prime_rho_op_ineq} on $\Omega$, we get 
  \begin{align}
      P_{{\rho}'}(\Omega) {\rho}'_{A_1^n B_1^n X_1^n E|\Omega} \leq \frac{e^{n \mu + \frac{1}{\mu^2}}}{1- \mu^2} P_{\sigma}(\Omega) \sigma_{A_1^n B_1^n X_1^n E|\Omega}. 
      \label{eq:EAT_prime_rho_op_ineq_cond}
  \end{align}
  Together, the above two equations imply that 
  \begin{align}
      D^{\mu'}_{\max}(\rho_{A_1^n B_1^n X_1^n E|\Omega} ||\sigma_{A_1^n B_1^n X_1^n E|\Omega}) \leq n \mu + \frac{1}{\mu^2} + \log\frac{P_{\sigma} (\Omega)}{P_{{\rho}'}(\Omega)} + \log \frac{1}{1- \mu^2}
  \end{align}
  for $\mu' := 2 \sqrt{\frac{\mu}{P_{{\rho}}(\Omega)}}$. \\

  For $\epsilon' >0$ such that $\mu' + \epsilon' < 1$ and $\alpha \in (1, 2]$, we can plug the above in the bound provided by Lemma \ref{lemm:ent_tri_ineq} to get
  \begin{align}
      H_{\min}^{\mu' + \epsilon'}(A_1^n|B_1^n E)_{\rho_{|\Omega}} &\geq \tilde{H}^{\uparrow}_{\alpha}(A_1^n|B_1^n E)_{\sigma_{|\Omega}} - \frac{\alpha}{\alpha-1} n \mu \nonumber \\
      & \qquad - \frac{1}{\alpha-1}\rndBrk{\alpha\log\frac{P_{\sigma}(\Omega)}{P_{{\rho}'}(\Omega)}+ \frac{\alpha}{\mu^2} + \alpha\log\frac{1}{1-\mu^2} +g_1(\epsilon', \mu')}.
      \label{eq:Hmin_rho_to_Halpha_sigma_wtest}
  \end{align}
  Now, note that using Eq. \ref{eq:test_maps} and \ref{eq:X_for_sigma_const}, and \cite[Lemma B.7]{Dupuis20} we have 
  \begin{align}
      \tilde{H}^{\uparrow}_{\alpha}(A_1^n|B_1^n E)_{\sigma_{|\Omega}} = \tilde{H}^{\uparrow}_{\alpha}(A_1^n X_1^n|B_1^n E)_{\sigma_{|\Omega}}.
      \label{eq:Halpha_sigma_adding_X}
  \end{align}
  For every $k$, we introduce a register $D_k$ of dimension $|D| := \lceil e^{\max(f)- \min(f)} \rceil$ and the channels $\mathcal{D}_k : X_k \rightarrow X_k D_k$ defined as 
  \begin{align}
      \mathcal{D}_k(\omega) := \sum_x \braket{x|\omega|x} \ket{x}\bra{x}\otimes \nu_x
  \end{align}
  where for every $x$, the state $\nu_x$ is a mixture between a uniform distribution on $\{1, 2, \cdots\ , \lfloor e^{\max(f)- f(\delta_x)} \rfloor\}$ and a uniform distribution on $\{1, 2, \cdots\ , \lceil e^{\max(f)- f(\delta_x)} \rceil\}$, so that
  \begin{align}
      H(D_k)_{\nu_x} = \max(f)- f(\delta_x)
  \end{align}
  where $\delta_x$ is the distribution with unit weight at element $x$. Define the state 
  \begin{align}
      \bar{\sigma}_{A_1^n B_1^n X_1^n D_1^n E} := {\mathcal{D}}_n \circ \cdots \circ {\mathcal{D}}_1 ({\sigma}_{A_1^n B_1^n X_1^n E})
      \label{eq:sigma_bar_defn}
  \end{align}
  Now \cite[Lemma 4.5]{Metger22} implies that this satisfies
  \begin{align}
      \tilde{H}^{\uparrow}_{\alpha}(A_1^n X_1^n|B_1^n E)_{\bar{\sigma}_{|\Omega}} \geq \tilde{H}^{\uparrow}_{\alpha}(A_1^n X_1^n D_1^n |B_1^n E)_{\bar{\sigma}_{|\Omega}} - \max_{x_1^n \in \Omega} H_{\alpha}(D_1^n)_{\bar{\sigma}_{|x_1^n}}
      \label{eq:Halpha_AX_to_AXD_bd}
  \end{align}
  For $x_1^n \in \Omega$, we have 
  \begin{align*}
      H_{\alpha}(D_1^n)_{\bar{\sigma}_{|x_1^n}} &\leq H (D_1^n)_{\bar{\sigma}_{|x_1^n}} \\
      &= \sum_{k=1}^n H(D_k)_{\nu_{x_k}} \\
      &= \sum_{k=1}^n \max(f)- f(\delta_{x_k}) \\
      &= n \max(f) - n f(\text{freq}(x_1^n)) \\
      &\leq n \max(f) - n h.
      \numberthis
      \label{eq:Halpha_on_Omega}
  \end{align*}
  We can get rid of the conditioning on the right-hand side of Eq. \ref{eq:Halpha_AX_to_AXD_bd} by using \cite[Lemma B.5]{Dupuis20}. This gives us
  \begin{align}
      \tilde{H}^{\uparrow}_{\alpha}(A_1^n X_1^n D_1^n |B_1^n E)_{\bar{\sigma}_{|\Omega}} &\geq \tilde{H}^{\uparrow}_{\alpha}(A_1^n X_1^n D_1^n |B_1^n E)_{\bar{\sigma}} - \frac{\alpha}{\alpha-1}\log \frac{1}{P_{\sigma}(\Omega)}
      \label{eq:rem_cond_Halpha_AXD_bd}
  \end{align}
  Moreover, using Eq. \ref{eq:sigma_in_corr_form}, we can show that $B_k$ is independent of $A_1^{k-1} X_1^{k-1} D_1^{k-1} B_1^{k-1} E$ in $\bar{\sigma}$. Firstly, for $\sigma_{A_1^{k-1} B_1^{k} E} = \tr_{A_k} \circ \cM^{\delta}_k (\sigma^{(k, 0)}_{A_1^{k-1} B_1^{k-1} R_k E})$, we have 
  \begin{align}
    \sigma_{A_1^{k-1} B_1^{k} E}
    &= (1-\delta) \tr_{A_k}\circ \cM_k (\sigma^{(k,0)}_{A_1^{k-1} B_1^{k-1} R_k E}) + \delta \rho^{(k, \delta)}_{B_k} \otimes \sigma_{A_1^{k-1} B_1^{k-1} E} \\
    &= \rndBrk{(1-\delta) \theta^{(k)}_{B_k} + \delta \rho^{(k, \delta)}_{B_k}} \otimes \sigma_{A_1^{k-1} B_1^{k-1} E}.
  \end{align}
  Since, $\bar{\sigma}_{A_1^{k-1} X_1^{k-1} D_1^{k-1} B_1^{k} E} = \mathcal{D}_{k-1} \circ \mathcal{T}_{k-1} \cdots \circ \mathcal{D}_1 \circ \mathcal{T}_1 (\sigma_{A_1^{k-1} B_1^{k} E})$, we can easily see that $B_k$ is independent of the other registers. In particular, $\bar{\sigma}$ satisfies the Markov chain $A_1^{k-1} X_1^{k-1} D_1^{k-1} \leftrightarrow B_1^{k-1} E \leftrightarrow B_k$. Let's define the channels 
  \begin{align}
    &\bar{\cM}_k := \mathcal{D}_k \circ \mathcal{T}_k \circ \cM_k \\
    &\bar{\cM}^{\delta}_k := \mathcal{D}_k \circ \mathcal{T}_k \circ \cM^{\delta}_k.
  \end{align}
  Now we can use \cite[Corollary 3.5]{Dupuis20} to show that for every $k \in [n]$
  \begin{align*}
    \tilde{H}^{\downarrow}_{\alpha}(A_1^k X_1^k D_1^k | B_1^k E)_{\bar{\sigma}} &\geq \tilde{H}^{\downarrow}_{\alpha}(A_1^{k-1} X_1^{k-1} D_1^{k-1} | B_1^{k-1} E)_{\bar{\sigma}} + \inf_{\omega_{R_k \tilde{R}_k}} \tilde{H}^{\downarrow}_{\alpha} (A_k X_k D_k| B_k \tilde{R}_k)_{\bar{\cM}_k^{\delta}(\omega)} \\
    &\geq \tilde{H}^{\downarrow}_{\alpha}(A_1^{k-1} X_1^{k-1} D_1^{k-1} | B_1^{k-1} E)_{\bar{\sigma}} \\ 
    & \qquad \qquad + \min\curlyBrk{\inf_{\omega_{R_k \tilde{R}_k}} \tilde{H}^{\downarrow}_{\alpha} (A_k D_k | B_k \tilde{R}_k)_{\bar{\cM}_k(\omega)}, \tilde{H}^{\downarrow}_{\alpha} (A_k D_k | B_k)_{\mathcal{D}_k \circ \mathcal{T}_k(\rho^{(k,\delta)}_{A_k B_k})}}.
    \numberthis
    \label{eq:approxEAT_wtest_entropy_bd}
  \end{align*}
  where we have used the quasi-concavity of R\'enyi conditional entropies \cite[Pg 73]{TomamichelBook16} and the fact that $X_k$ are classical in the second line. \\

  We now lower bound the two terms in the minimum above. Using \cite[Lemma B.9]{Dupuis20}, for a state $\nu = \bar{\cM}_k(\omega_{R_k \tilde{R}_k})$ and $1 < \alpha < \frac{1}{\log (1 + 2|A||D|)}$ we have that 
  \begin{align*}
    \tilde{H}^{\downarrow}_{\alpha} (A_k D_k | B_k \tilde{R}_k)_{\nu} &\geq H (A_k D_k | B_k \tilde{R}_k)_{\nu} - \rndBrk{\alpha-1} \log^2 \rndBrk{1 + 2|A||D|} \\
    &= H (A_k | B_k \tilde{R}_k)_{\nu} + H(D_k | A_k B_k \tilde{R}_k)_{\nu} - \rndBrk{\alpha-1} \log^2 \rndBrk{1 + 2|A||D|} \\
    &= H (A_k | B_k \tilde{R}_k)_{\nu} + H(D_k | X_k)_{\nu} - \rndBrk{\alpha-1} \log^2 \rndBrk{1 + 2|A||D|}\\
    &= H (A_k | B_k \tilde{R}_k)_{\nu} + \sum_x \nu(x) (\text{max}(f) - f(\delta_x)) - \rndBrk{\alpha-1} \log^2 \rndBrk{1 + 2|A||D|}\\
    &= H (A_k | B_k \tilde{R}_k)_{\nu} + \text{max}(f) - f(\nu_X) - \rndBrk{\alpha-1} \log^2 \rndBrk{1 + 2|A||D|}\\
    &\geq \text{max}(f) - \rndBrk{\alpha-1} \log^2 \rndBrk{1 + 2|A||D|}
    \numberthis
    \label{eq:per_rnd_Halpha_bd1}
  \end{align*}
  For the second term, and $1 < \alpha < \frac{1}{\log (1 + 2|A||D|)}$, we have
  \begin{align*}
    \tilde{H}^{\downarrow}_{\alpha} (A_k D_k | B_k)_{\mathcal{D}_k \circ \mathcal{T}_k(\rho^{(k, \delta)}_{A_k B_k})} &\geq H (A_k D_k | B_k )_{\mathcal{D}_k \circ \mathcal{T}_k(\rho^{(k, \delta)}_{A_k B_k})} - \rndBrk{\alpha-1} \log^2 \rndBrk{1 + 2|A||D|} \\
    &\geq H (A_k D_k | B_k )_{\bar{\cM}_k(\tilde{\rho}^{(k,0)}_{R_k})} - \rndBrk{\alpha-1} \log^2 \rndBrk{1 + 2|A||D|} \\
    &\qquad \qquad  - 2(\epsilon + \delta) \log |A||D| - g(\epsilon + \delta)\\
    &\geq \text{max}(f) - \rndBrk{\alpha-1} \log^2 \rndBrk{1 + 2|A||D|} - 2(\epsilon + \delta) \log |A||D| - g_2(\epsilon + \delta)
    \numberthis
    \label{eq:per_rnd_Halpha_bd2}
  \end{align*}
  where we have used \cite[Lemma B.9]{Dupuis20} in the first line, the AFW continuity bound \cite[Theorem 11.10.3]{Wilde13} in the second line to convert the entropy on $\mathcal{D}_k \circ \mathcal{T}_k(\rho^{(k, \delta)}_{A_k B_k})$ to an entropy on $\bar{\cM}_k(\tilde{\rho}^{(k,0)}_{R_k})$ and finally we use bound in Eq. \ref{eq:per_rnd_Halpha_bd1} for states of this form. \\

  Plugging these in Eq. \ref{eq:approxEAT_wtest_entropy_bd}, gives us that for every $k \in [n]$
  \begin{align}
    \tilde{H}^{\downarrow}_{\alpha}(A_1^k X_1^k D_1^k | B_1^k E)_{\bar{\sigma}} &\geq \tilde{H}^{\downarrow}_{\alpha}(A_1^{k-1} X_1^{k-1} D_1^{k-1} | B_1^{k-1} E)_{\bar{\sigma}} \nonumber\\ 
    & \qquad \qquad + \text{max}(f)  - \rndBrk{\alpha-1} \log^2 \rndBrk{1 + 2|A||D|}\nonumber \\
    & \qquad \qquad - 2(\epsilon + \delta) \log |A||D| - g_2(\epsilon + \delta)
  \end{align}
  Consecutively using this bound gives us
  \begin{align*}
    \tilde{H}^{\uparrow}_{\alpha}(A_1^n X_1^n D_1^n | B_1^n E)_{\bar{\sigma}} &\geq \tilde{H}^{\downarrow}_{\alpha}(A_1^n X_1^n D_1^n | B_1^n E)_{\bar{\sigma}} \\
    &\geq  n\text{max}(f)  - n\rndBrk{\alpha-1} \log^2 \rndBrk{1 + 2|A||D|}\\
    &\qquad \qquad - 2n(\epsilon + \delta) \log |A||D| - ng_2(\epsilon + \delta)
    \numberthis
    \label{eq:Halpha_sigma_AXD_bd}
  \end{align*}
  Combining Eq. \ref{eq:Halpha_sigma_adding_X}, \ref{eq:Halpha_AX_to_AXD_bd}, \ref{eq:Halpha_on_Omega}, \ref{eq:rem_cond_Halpha_AXD_bd} and \ref{eq:Halpha_sigma_AXD_bd}, we get 
  \begin{align*}
    \tilde{H}^{\uparrow}_{\alpha}(A_1^n | B_1^n E)_{\sigma_{|\Omega}} &\geq nh  - n\rndBrk{\alpha-1} \log^2 \rndBrk{1 + 2|A||D|} - 2n(\epsilon + \delta) \log |A||D| - ng_2(\epsilon + \delta) \\
    &\qquad \qquad - \frac{\alpha}{\alpha-1}\log \frac{1}{P_{\sigma}(\Omega)}.
    \numberthis
  \end{align*}
  Plugging this into Eq. \ref{eq:Hmin_rho_to_Halpha_sigma_wtest}, we get
  \begin{align*}
    H_{\min}^{\mu' + \epsilon'}(A_1^n|B_1^n E)_{\rho_{|\Omega}} &\geq nh  - n\rndBrk{\alpha-1} \log^2 \rndBrk{1 + 2|A||D|} - \frac{\alpha}{\alpha-1} n \mu \\
      & \qquad - 2n(\epsilon + \delta) \log |A||D| - ng_2(\epsilon + \delta) \\
      & \qquad - \frac{1}{\alpha-1}\rndBrk{\alpha\log\frac{1}{P_{{\rho}'}(\Omega)} + \frac{\alpha}{\mu^2} + \alpha \log \frac{1}{1- \mu^2} + g_1(\epsilon', \mu')}.
      \numberthis
  \end{align*}
  Finally, we choose $\alpha = 1 + \frac{\sqrt{\mu}}{\log \rndBrk{1 + 2|A||D|}}$ and use the $\alpha <2$ as an upper bound to derive
  \begin{align}
    H_{\min}^{\mu' + \epsilon'}(A_1^n|B_1^n E)_{\rho_{|\Omega}} &\geq n(h  - 3\sqrt{\mu} \log \rndBrk{1 + 2|A||D|} - 2(\epsilon + \delta) \log |A||D| - g_2(\epsilon + \delta)) \nonumber \\
    & \quad - \frac{\log \rndBrk{1 + 2|A||D|}}{\sqrt{\mu}}\rndBrk{2\log\frac{1}{P_{\rho}(\Omega) - \mu} + \frac{2}{\mu^2} + 2 \log \frac{1}{1- \mu^2} + g_1(\epsilon', \mu')}.
    \label{eq:almost_final_bd}
  \end{align}
  Note that $\log \rndBrk{1 + 2|A||D|} \leq \log \rndBrk{1 + 2|A|} + \log |D| \leq \log \rndBrk{1 + 2|A|} + \lceil \max(f) - \min(f)\rceil = V$. Recall that $\mu = z(\epsilon + \nu, \delta)^{1/3}$. Since, $\nu$ can be chosen arbitrarily greater than $0$, if we let $\nu \rightarrow 0$, the bound in Eq. \ref{eq:almost_final_bd} is still valid. Finally, choose $\delta = \epsilon$ to derive the bound in the theorem. \\

  \textbf{Case 2:} Finally, for a state $\rho_{A_1^n B_1^n E}$ which satisfies the assumptions of the Theorem but is not full rank, we can consider the full rank states $\rho_{A_1^n B_1^n E}^{(\varepsilon)} := (1- \varepsilon)\rho_{A_1^n B_1^n E} + \varepsilon \tau_{A_1^n B_1^n E}$ for an arbitrarily small $\varepsilon > 0$. $\rho_{A_1^n B_1^n E}^{(\varepsilon)}$ will be full rank and satisfy the assumptions of the Theorem and the bound in Eq. \ref{eq:almost_final_bd} for $\epsilon \rightarrow \epsilon + \varepsilon$. One can then prove the lower bound above for the state $\rho_{A_1^n B_1^n E}$ by taking the limit $\varepsilon \rightarrow 0$. 
\end{proof}

\bibliographystyle{halpha}
\bibliography{bib}

\end{document}